\documentclass[letterpaper,twocolumn,10pt]{article}
\usepackage{usenix2019_v3,epsfig,endnotes}
\usepackage{bbm}
\usepackage{epsfig,endnotes}
\usepackage{multirow}
\usepackage{subfig}
\usepackage{graphicx}
\usepackage{grffile}

\newcounter{subcopyrightbox@save}
\usepackage[font=bf]{caption}
\usepackage{color, url}
\usepackage{xspace} 
\usepackage{mathrsfs}
\usepackage{amssymb}
\usepackage{amsmath}
\usepackage{amsthm}
\usepackage{epstopdf}
\usepackage{balance}
\usepackage{bm}

\usepackage{thm-restate,thmtools}

\usepackage[ruled,linesnumbered,noend]{algorithm2e}
\SetKwRepeat{Do}{do}{while}
\usepackage{mathtools}
\DeclarePairedDelimiter{\ceil}{\lceil}{\rceil}
\DeclarePairedDelimiter{\floor}{\lfloor}{\rfloor}

\usepackage{cleveref}
\allowdisplaybreaks

\newcommand{\myparatight}[1]{\smallskip\noindent{\bf {#1}:}~}

\newcommand{\RN}[1]{%
  \textup{\uppercase\expandafter{\romannumeral#1}}%
}

\DeclareMathOperator*{\argmax}{arg\,max}

\newenvironment{packeditemize}{\begin{list}{$\bullet$}{\setlength{\itemsep}{0.2pt}\addtolength{\labelwidth}{-4pt}\setlength{\leftmargin}{\labelwidth}\setlength{\listparindent}{\parindent}\setlength{\parsep}{1pt}\setlength{\topsep}{0pt}}}{\end{list}}

\AtBeginDocument{%
  \providecommand\BibTeX{{%
    \normalfont B\kern-0.5em{\scshape i\kern-0.25em b}\kern-0.8em\TeX}}}

\begin{document}
\thispagestyle{headings}
\markright{\hfill To appear in the 30th Usenix Security Symposium, August 2021, Vancouver, B.C., Canada\hfill}

\title{Data Poisoning Attacks to Local Differential Privacy Protocols}

\author{
\rm{Xiaoyu Cao, Jinyuan Jia, Neil Zhenqiang Gong} \\
Duke University\\
\{xiaoyu.cao, jinyuan.jia, neil.gong\}@duke.edu}

\maketitle

%%%%%%%%%%%%%%%%%%%%%%%%%%%%%%%%%%%%%%%%%%%%%%%%%%%%%%%

\begin{abstract}
Local Differential Privacy (LDP) protocols enable an untrusted data collector to perform privacy-preserving data analytics. In particular, each user locally perturbs its data to preserve privacy before sending it to the data collector, who aggregates the perturbed data to obtain statistics of interest. In the past several years,  researchers from multiple communities--such as security, database, and theoretical computer science-- have proposed many LDP protocols. These studies mainly focused on improving the utility of the  LDP protocols. However, the security of LDP protocols is largely unexplored. 

In this work, we aim to bridge this gap. We focus on LDP protocols for \emph{frequency estimation} and \emph{heavy hitter identification}, which are two basic data analytics tasks. Specifically, we show that an attacker can inject fake users into an LDP protocol and the fake users send carefully crafted data to the data collector such that the LDP protocol estimates high frequencies for arbitrary attacker-chosen items or identifies them as heavy hitters. We call our attacks \emph{data poisoning attacks}. We theoretically and/or empirically show the effectiveness of our attacks. We also explore three countermeasures against our attacks. Our experimental results show that they can effectively defend against our attacks in some scenarios but have limited effectiveness in others,  highlighting the needs for new defenses against our attacks.  
\end{abstract}

%%%%%%%%%%%%%%%%%%%%%%%%%%%%%%%%%%%%%%%%%%%%%%%%%%%%%%%

\section{Introduction}

Various data breaches~\cite{equifaxdataleak,capitalonedataleak,canvasdataleak} have highlighted the challenges of relying on a data collector (e.g., Equifax) to protect users' private data. \emph{Local Differential Privacy (LDP)}, a variant of differential privacy~\cite{dwork2006calibrating}, aims to address such challenges. In particular, an LDP protocol encodes and perturbs a user's data to protect privacy before sending it to the data collector, who aggregates the users' perturbed data to obtain statistics of interest. Therefore, even if the data collector is compromised, user privacy is still preserved as the attacker only has access to users' privacy-preserving perturbed data.  Because of the resilience against untrusted data collectors, LDP has attracted increasing attention in both academia and industry. Specifically, many LDP protocols~\cite{duchi2013local,erlingsson2014rappor,kairouz2014extremal,wang2017locally,bassily2015local,kairouz2016discrete,bassily2017practical,wang2019locally,qin2016heavy,wang2018locally,zhang2018calm,jia2019calibrate,wang2019consistent,ren2018textsf,avent2017blender,cormode2018marginal,wang2019answering} have been developed in the past several years. Moreover, some of these protocols have been widely deployed in industry including but not limited to Google, Microsoft, and Apple. {For instance, Google deployed LDP~\cite{erlingsson2014rappor} in the Chrome browser to collect  users' default homepages for Chrome; Microsoft~\cite{ding2017collecting} integrated LDP in Windows 10 to collect application usage statistics; and Apple~\cite{appledf2017} adopted LDP on iOS to identify  popular emojis, which are subsequently recommended to users.}

Since LDP perturbs each user's data, it sacrifices utility of the data analytics results obtained by the data collector. Therefore, existing studies on LDP mainly focused on improving the utility via designing new methods to encode/perturb users' data and aggregate the perturbed data to derive statistical results. However, the security of LDP is largely unexplored.   

In this work, we aim to bridge this gap. In particular, we propose a family of attacks called \emph{data poisoning attacks} to LDP protocols. In our attacks, an attacker injects fake users to an LDP protocol and carefully crafts the data sent from the fake users to the data collector, with the goal to manipulate  the data analytics results as the attacker desires.  Specifically, we focus on LDP protocols for \emph{Frequency Estimation} and \emph{Heavy Hitter Identification}, which are two basic data analytics tasks and are usually the first step towards more advanced tasks. The goal of frequency estimation is to estimate the \emph{fraction of users} (i.e., frequency)  that have a certain item for each of a set of items, while the goal of heavy hitter identification is to only identify the top-$k$ items that are the most frequent among the users without estimating the items' frequencies. Our attacks can increase the estimated frequencies for arbitrary attacker-chosen items (called \emph{target items}) in frequency estimation or promote them to be identified as top-$k$ heavy hitters in heavy hitter identification. Our attacks  result in severe security threats to LDP-based data analytics. 
{For example, an attacker can promote a phishing webpage as a popular default homepage of Chrome; an attacker can increase the estimated popularity of its (malicious) application when LDP is used to estimate application popularity; and an attacker can manipulate the identified and recommended popular emojis, resulting in bad user experience and frustration.} 

The major challenge of data poisoning attacks is that, given a limited number of fake users an attacker can inject, what data the fake users should send to the data collector such that the attack effectiveness is maximized. To address the challenge, we formulate our attacks as an optimization problem, whose  objective function  is to maximize the attack effectiveness and whose solution is the data that fake users should send to the data collector. We call our optimization-based attack \emph{Maximal Gain Attack (MGA)}. To better demonstrate the effectiveness of MGA, we also propose two baseline attacks in which the fake users send randomly crafted data to the data collector.  Then, we apply our MGA and the baseline attacks to three state-of-the-art LDP protocols for frequency estimation (i.e., kRR~\cite{kairouz2014extremal}, OUE~\cite{wang2017locally}, and OLH~\cite{wang2017locally}) and one state-of-the-art LDP protocol for heavy hitter identification (i.e., PEM~\cite{wang2019locally}).  

We theoretically evaluate the effectiveness of our attacks. Specifically,  we derive the \emph{frequency gain} of the target items, which is the difference of the target items' estimated frequencies after and before an attack. Our theoretical analysis shows that our MGA can achieve the largest frequency gain among possible attacks. Our theoretical results also show a fundamental security-privacy tradeoff for LDP protocols: when an LDP protocol provides higher privacy guarantees, the LDP protocol is less secure against our attacks (i.e., the frequency gains are larger). Moreover, we observe that different LDP protocols have different security levels against our attacks. For instance,  OUE and OLH have similar security levels against our attacks, and  kRR is less secure than OUE and OLH when the number of items is larger than a threshold. We also empirically evaluate our attacks for both frequency estimation and heavy hitter identification using a synthetic dataset and two real-world datasets. Our empirical results also show the effectiveness of our attacks. For example, on all the three datasets, our MGA can promote 10 randomly selected target items to be identified as  top-$15$ heavy hitters when the attacker only injects $5\%$ of fake users. 

{We also explore three countermeasures, i.e., \emph{normalization},  \emph{detecting fake users}, and \emph{detecting the target item}, to defend against our attacks.} Specifically, in {normalization}, the data collector normalizes the estimated item frequencies to be a probability distribution, i.e., each estimated item frequency is non-negative and the estimated frequencies of all items sum to 1. Since our attacks craft the data for the fake users via solving an optimization problem, the data from the fake users may follow certain patterns that deviate from genuine users. Therefore, in our second countermeasure, the data collector aims to detect fake users via analyzing the statistical patterns of the data from the users, and the data collector filters the detected fake users before estimating frequencies or identifying heavy hitters.  {The third countermeasure detects the target item without detecting the fake users when there is only one target item.} Our empirical results show that these countermeasures can effectively defend against our attacks in some scenarios. For example, when the attacker has  $10$ target items, normalization can reduce the frequency gain of our MGA to OUE from $1.58$ to $0.46$ and detecting fake users can reduce the frequency gain to be almost $0$ because the data collector can detect almost all fake users. However, our attacks are still effective in other scenarios. For instance,  when the attacker has  $10$ randomly selected target items, our MGA to  OLH  still achieves a frequency gain of 0.43 even if both detecting fake users and normalization are used. Our results highlight the needs for new defenses against our attacks. 

In summary, our contributions are as follows: 

\begin{packeditemize}

\item We perform the first systematic study on \emph{data poisoning attacks} to LDP protocols for frequency estimation and heavy hitter identification. 

\item We show that, both theoretically and/or empirically, our attacks can effectively increase the estimated frequencies of the target items or promote them to be identified as heavy hitters. 

\item We explore three countermeasures to defend against our attacks. Our empirical results highlight the needs for new defenses against our attacks. 

\end{packeditemize}

%%%%%%%%%%%%%%%%%%%%%%%%%%%%%%%%%%%%%%%%%%%%%%%%%%%%%%%

\section{Background and Related Work}

We consider LDP protocols for two basic tasks, i.e., \emph{frequency estimation}~\cite{warner1965randomized,duchi2013local,erlingsson2014rappor,kairouz2014extremal,wang2017locally,bassily2015local,kairouz2016discrete,zhang2018calm,jia2019calibrate,wang2019consistent} and \emph{heavy hitter identification}~\cite{bassily2017practical,wang2019locally,qin2016heavy}. Suppose there are $n$ users. Each user holds one item from a certain domain, e.g., the default homepage of a browser. We denote the domain of the items as $\{1,2,\cdots,d\}$. For conciseness, we simplify  $\{1,2,\cdots,d\}$ as $[d]$. In frequency estimation, the data collector (also called \emph{central server}) aims to estimate the frequency of each item among the $n$ users, while heavy hitter identification aims to identify the top-$k$ items that have the largest frequencies among the $n$ users. Frequency of an item is defined as the fraction of users  who have the item. 

\subsection{Frequency Estimation}
\label{background}

An LDP protocol for frequency estimation consists of three key steps: \emph{encode}, \emph{perturb}, and \emph{aggregate}. The encode step encodes each user's item into some numerical value. We denote the space of encoded values  as $\mathcal{D}$. The perturb step randomly perturbs the value in the space $\mathcal{D}$ and sends the perturbed value to the central server. The central server estimates  item frequencies using the perturbed values from all users in the aggregate step. For simplicity, we denote by $PE(v)$ the perturbed encoded value for an item $v$.  Roughly speaking, a protocol satisfies LDP if any two items are perturbed to the same value with close probabilities. Formally, we have  the following definition:

\begin{restatable}[Local Differential Privacy]{defi}{ldp}\label{defi:ldp}
           A protocol $\mathcal{A}$ satisfies $\epsilon$-local differential privacy ($\epsilon$-LDP) if for any pair of items $v_1, v_2\in [d]$ and any perturbed value ${y}\in \mathcal{D}$, we have $\text{Pr}(PE(v_1)={y}) \le e^\epsilon\text{Pr}(PE(v_2)={y})$,  
	where $\epsilon > 0$ is called privacy budget and $PE(v)$ is the random perturbed encoded value of an item $v$.
\end{restatable}

Moreover, an LDP protocol is called pure LDP if it satisfies the following definition:

\begin{restatable}[Pure LDP~\cite{wang2017locally}]{defi}{pure-ldp}\label{defi:pure-ldp}
           An LDP protocol is pure if there are two probability parameters $0<q<p<1$ such that the following equations hold for any pair of items $v_1, v_2\in[d], v_1\neq v_2$: 
	{\small{\begin{align}
		\text{Pr}(PE(v_1)\in\{{y}|v_1\in{S({y})}\}) &= p\\
		\text{Pr}(PE(v_2)\in\{{y}|v_1\in{S({y})}\}) &= q,
	\end{align}}}%
	where $S({y})$ is  the set of items that ${y}$ supports. 
\end{restatable}

We note that the definition of the support $S({y})$ depends on the LDP protocol. For instance, for some LDP protocols~\cite{duchi2013local,wang2017locally}, the support $S({y})$ of a perturbed value ${y}$ is the set of items whose encoded values could be ${y}$. For a pure LDP protocol, the aggregate step is as follows:
{\small{\begin{align}
\label{aggregate}
	\tilde{f}_v = \frac{\frac{1}{n}\sum\limits_{i=1}^{n} \mathbbm{1}_{{S}({y}_i)}(v) - q}{p-q},
\end{align}}}%
where $\tilde{f}_v$ is the estimated frequency for item $v\in[d]$, ${y}_i$ is the perturbed value from the $i$th user, and $\mathbbm{1}_{{S}({y}_i)}(v)$ is an \emph{characteristic function}, which outputs 1 if and only if ${y}_i$ supports item $v$. Formally, the characteristic function $\mathbbm{1}_{{S}({y}_i)}(v)$ is defined as follows: $\mathbbm{1}_{{S}({y})}(v)$ is 1 if $v\in {S}({y})$ and 0 otherwise. 

Roughly speaking, Equation (\ref{aggregate}) means that the frequency of an item is estimated as the fraction of users whose perturbed values support the item normalized by $p, q,$ and $n$. Pure LDP protocols are unbiased estimators of the item frequencies~\cite{wang2017locally}, i.e., $\mathbb{E}[\tilde{f}_v] = f_v$, where $f_v$ is the true frequency for item $v$. Therefore, we have:
{\small{\begin{align}
\label{unbiasedestimator}
	 \sum\limits_{i=1}^{n}\mathbb{E}[\mathbbm{1}_{S(y_i)}(v)] = n(f_v(p-q) + q).
\end{align}}}

Equation (\ref{unbiasedestimator}) will be useful for the analysis of our attacks. 
Next, we describe three state-of-the-art pure LDP protocols, i.e., kRR~\cite{duchi2013local}, OUE~\cite{wang2017locally}, and OLH~\cite{wang2017locally}. These three protocols are recommended for use in different scenarios. Specifically, kRR achieves the smallest estimation errors when the number of items is small, i.e., $d<3e^\epsilon+2$. When the number of items is large, both OUE and OLH achieve the smallest estimation errors. OUE has a larger communication cost, while OLH has a larger computation cost for the central server. Therefore, when the communication cost is a bottleneck, OLH is recommended, otherwise OUE is recommended. 

\subsubsection{kRR} 

\myparatight{Encode} kRR encodes an item $v$ to itself. Therefore, the encoded space $\mathcal{D}$ for kRR is identical to the domain of items, which is $\mathcal{D} = [d]$.  

\myparatight{Perturb} kRR keeps an encoded item unchanged with a probability $p$ and perturbs it to a different random item $a\in\mathcal{D}$ with probability $q$. Formally, we have:
{\small{\begin{align}
\label{krr-perturb}
\text{Pr}(y=a)=
    \begin{cases}
    \frac{e^{\epsilon}}{d-1+e^{\epsilon}}\triangleq p,& \text{ if }a=v, \\
    \frac{1}{d-1+e^{\epsilon}}\triangleq q,& \text{ otherwise},
 \end{cases}
\end{align}}}%
where $y$ is the random perturbed value sent to the central server when a user's item is $v$.  

\myparatight{Aggregate} The key for aggregation is to derive the support set. A perturbed value ${y}$ only supports itself for kRR. Specifically, we have $S({y}) = \{y\}$. 
Given the support set, we can estimate item frequencies using Equation (\ref{aggregate}).

\subsubsection{OUE} 

\myparatight{Encode} OUE encodes an item $v$ to a $d$-bit binary vector $\bm{e}_v$ whose bits are all zero except the $v$-th bit. The encoded space for OUE is  $\mathcal{D} = \{0, 1\}^d$, where $d$ is the number of items. 

\myparatight{Perturb} OUE perturbs the bits of the encoded binary vector independently. Specifically, for each bit of the encoded binary vector, if it is 1, then it remains 1 with a probability $p$. Otherwise if the bit is 0, it is flipped to 1 
with a probability $q$. Formally, we have:
{\small{\begin{align}\label{eq:perturb_oue}
	    \text{Pr}(y_i=1)=
    \begin{cases}
    \frac{1}{2}\triangleq p, &\text{ if }i=v, \\
    \frac{1}{e^{\epsilon}+1}\triangleq q,& \text{ otherwise},
    \end{cases}
\end{align}}}%
where the vector $\bm{y}=[y_1\ y_2\ \cdots\ y_d]$ is the perturbed value for a user with item $v$. 

\myparatight{Aggregate} A perturbed value $\bm{y}$ supports an item $v$ if and only if the $v$-th bit of $\bm{y}$, denoted as $y_v$, equals to 1. Formally, we have $S(\bm{y}) = \{v|v\in[d]\text{ and } y_v=1\}$. 

\subsubsection{OLH} 

\myparatight{Encode}  OLH  leverages a family of hash functions $\mathbf{H}$, each of which  maps an item $v\in[d]$ to a  value $h\in[d']$, where $d'<d$. In particular, OLH uses $d'=e^\epsilon+1$ as it achieves the best performance~\cite{wang2017locally}. 
An example of  the  hash function family $\mathbf{H}$ could be  xxhash \cite{collet2016xxhash}  with different seeds. Specifically,  a seed is a non-negative integer and each seed represents a different xxhash   hash function.
In the encode step, OLH randomly picks a hash function $H$ from  $\mathbf{H}$. When xxhash is used, randomly picking a hash function is equivalent to randomly selecting a non-negative integer as a seed. Then, OLH computes the hash value of the item $v$ as $h=H(v)$. The tuple $(H,h)$ is the encoded value for the item $v$. The  encoded space for OLH is $\mathcal{D} = \{(H,h)|H\in\mathbf{H} \text{ and } h\in[d']\}$. 

\myparatight{Perturb} OLH only perturbs the hash value $h$ and does not change the hash function $H$. In particular, the hash value stays unchanged with probability $p'$ and switches to a different value in $[d']$ with probability $q'$. Formally, we have:
{\small{\begin{align}\label{eq:perturb_olh}
    \text{Pr}(y=(H, a))=
    \begin{cases}
    \frac{e^\epsilon}{e^\epsilon+d'-1}\triangleq p', &\text{ if }a=H(v), \\
    \frac{1}{e^\epsilon+d'-1} \triangleq q', &\text{ otherwise},
    \end{cases}
\end{align}}}%
where $y$ is the perturbed value sent to the central server from a user with item $v$. Therefore, the overall probability parameters $p$ and $q$ are $p = p' = \frac{e^\epsilon}{e^\epsilon+d'-1}$ and $q = \frac{1}{d'}\cdot p' + (1-\frac{1}{d'})\cdot q' = \frac{1}{d'}	$. 

\myparatight{Aggregate} A perturbed value ${y}=(H, h)$ supports an item $v\in[d]$ if $v$ is hashed  to $h$ by $H$. Formally, we have $S({y}) = \{v|v\in[d]\text{ and }H(v)=h\}$.

\subsection{Heavy Hitter Identification}\label{sec:hh_bg}
The goal of heavy hitter identification~\cite{bassily2017practical,wang2019locally,bassily2015local} is to identify the top-$k$ items that are the most frequent among the $n$ users. A direct and simple solution is to first estimate the frequency of each item using a frequency estimation protocol and then select the $k$ items with the largest frequencies. However, such method is not scalable to a large number of items. In response, a line of works~\cite{bassily2017practical,wang2019locally,bassily2015local} developed protocols to identify heavy hitters without estimating item frequencies. For example, Bassily et al.~\cite{bassily2017practical} and Wang et al.~\cite{wang2019locally} independently developed a similar heavy hitter identification protocol, which divides users into groups and iteratively applies a frequency estimation protocol to identify frequent prefixes within each group. Next, we take the \emph{Prefix Extending Method ({PEM})}~\cite{wang2019locally}, a state-of-the-art heavy hitter identification protocol, as an example to illustrate the process. 

In PEM, each user encodes its item as a $\gamma$-bits binary vector.  Suppose users are evenly divided into $g$ groups. In the $j$th iteration, users in the $j$th group use the OLH protocol to perturb the first $\lambda_j=\ceil*{\log_2 k} + \ceil*{j \cdot \frac{\gamma-\ceil*{\log_2 k}}{g}}$ bits of their binary vectors and send the perturbed bits to the central server, which uses the aggregate step of the OLH protocol to estimate the frequencies of the prefixes that extend the previous top-$k$ prefixes. OLH instead of OUE is used because the number of items corresponding to $\lambda_j$ bits is $2^{\lambda_j}$, which is often large and incurs large communication costs for OUE. Specifically, the central server uses the aggregate step of  OLH  to estimate the frequencies of the $\lambda_j$-bits prefixes in the set $R_{j-1} \times \{0,1\}^{\lambda_j-\lambda_{j-1}}$, 
where $R_{j-1}$ is the set of top-$k$ $\lambda_{j-1}$-bits prefixes identified in the $(j-1)$th iteration and the $\times$ symbol denotes Cartesian product. After estimating the frequencies of these $\lambda_j$-bits prefixes, the central server identifies the top-$k$ most frequent ones, which are denoted as the set $R_{j}$. This process is repeated for the $g$ groups and the  set of top-$k$ prefixes  in the final iteration are identified as the  top-$k$ heavy hitters. 

\subsection{Data Poisoning Attacks}
\vspace{-2mm}

\myparatight{Data poisoning attacks to LDP protocols} A concurrent work~\cite{cheu2019manipulation} studied  \emph{untargeted} attacks to LDP protocols. In particular, they focused on degrading the overall performance of frequency estimation or heavy hitter identification.  For instance, we can represent the estimated frequencies of all items as a vector, where an entry corresponds to an item. They studied how an attack can manipulate the $L_p$-norm distance between such vectors before and after attack. 
In contrast, we study \emph{targeted} attacks that aim  to increase the estimated frequencies of the attacker-chosen target items or promote them to be identified as heavy hitters. We note that the  $L_p$-norm distance between the item frequency vectors is different from the increased estimated frequencies for the target items. For instance, $L_1$-norm distance between the item frequency vectors is a loose upper bound of the increased estimated frequencies for the target items.  

\myparatight{Data poisoning attacks to machine learning} A line of works~\cite{newsome2006paragraph,perdisci2006misleading,nelson2008exploiting,rubinstein2009antidote,huang2011adversarial,biggio2012poisoning,wang2014man,newell2014practicality,mozaffari2014systematic,mei2015using,li2016data,alfeld2016data,munoz2017towards,jagielski2018manipulating,shafahi2018poison,gu2019badnets,liu2017trojaning,fang2018poisoning,yang2017fake,jia2020intrinsic,fang2020influence,fang2020local} studied  data poisoning attacks to  machine learning systems. In particular, the attacker manipulates the training data such that a bad model is learnt, which makes predictions as the attacker desires. For instance, Biggio et al.~\cite{biggio2012poisoning} investigated data poisoning attacks against Support Vector Machines. Jagielski et al.~\cite{jagielski2018manipulating} studied data poisoning attacks to regression models. Shafahi et al.~\cite{shafahi2018poison} proposed  poisoning attacks to neural networks, where the learnt model makes incorrect predictions only for target testing examples. Gu et al.~\cite{gu2019badnets} and Liu et al.~\cite{liu2017trojaning} proposed data poisoning attacks (also called backdoor/trojan attacks) to neural networks, where the learnt model  predicts an attacker-chosen label for testing examples with a certain trigger. 
 Data poisoning attacks were also proposed to spam filters~\cite{nelson2008exploiting}, recommender systems~\cite{li2016data,yang2017fake,fang2018poisoning,fang2020influence},  graph-based methods~\cite{wang2019attacking}, etc.. Our data poisoning attacks are different from these attacks because how LDP protocols aggregate the users' data to estimate frequencies or identify heavy hitters is  substantially different from how a machine learning system aggregates training data to derive a model. 

%%%%%%%%%%%%%%%%%%%%%%%%%%%%%%%%%%%%%%%%%%%%%%%%%%%%%%%

\section{Attacking Frequency Estimation}
\subsection{Threat Model}\label{sec:threat_model}
We characterize our threat model with respect to an attacker's capability, background knowledge, and goal. 

\myparatight{Attacker's capability and background knowledge} We assume an attacker can inject some fake users into an LDP protocol. These fake users can send arbitrary data in the encoded space to the central server. Specifically, we assume $n$ genuine users and the attacker injects $m$ fake users to the system. Therefore, the total number of users becomes $n+m$. We note that it is a practical threat model to assume that an attacker can inject fake users.{In particular, previous measurement study~\cite{thomas2013trafficking} showed that attackers can easily have access to a large number of fake/compromised accounts in various web services such as Twitter, Google, and Hotmail. Moreover, an attacker can buy fake/compromised accounts for these web services from merchants in the underground market with cheap prices. For instance, a Hotmail account costs \$0.004 -- 0.03; and a phone verified Google account costs \$0.03 -- 0.50 depending on the merchants.}

 Since an LDP protocol executes the encode and perturb steps locally on users' side, the attacker has access to the implementation of these steps. Therefore,  the attacker knows various parameters of the LDP protocol. In particular, the attacker knows the domain size $d$, the encoded space $\mathcal{D}$, and the support set $S(y)$ for each perturbed value $y\in\mathcal{D}$. 
  
\myparatight{Attacker's goal} We consider the attacker's goal is to {promote} some target items, i.e., increase the estimated frequencies of the target items. 
For example, a company may be interested in making its products  more popular. 
Formally, we denote by ${T}=\{t_1, t_2, \cdots, t_r\}$  the set of $r$ target items. 
To increase the estimated frequencies of the target items, 
the attacker carefully crafts the perturbed values sent from the fake users to the central server. We denote by $\mathbf{Y}$ the set of crafted perturbed values for the fake users, where an entry $y_i$ of $\mathbf{Y}$ is the crafted perturbed value for a fake user. The perturbed value $y_i$ could be a number (e.g., for kRR protocol), a binary vector (e.g., for OUE), and a tuple (e.g., for OLH).   

Suppose $\tilde{f}_{t,b}$ and $\tilde{f}_{t,a}$ are the  frequencies estimated by the LDP protocol for a target item $t$ before and after attack, respectively. We define the \emph{frequency gain} $\Delta\tilde{f}_t$ for a target item $t$ as $\Delta\tilde{f}_t =\tilde{f}_{t,a} - \tilde{f}_{t,b}, \forall t \in T$. A larger frequency gain $\Delta\tilde{f}_t$ implies a more successful attack. Note that an LDP protocol perturbs the value on each genuine user randomly. Therefore, the frequency gain $\Delta\tilde{f}_t$ is random for a given set of crafted perturbed values $\mathbf{Y}$ for the fake users. Thus, we define the attacker's \emph{overall gain} $G$  using the sum of the expected frequency gains for the target items, i.e., $G(\mathbf{Y}) = \sum_{t\in T}\mathbb{E}[\Delta\tilde{f}_t]$, {where $\Delta\tilde{f}_t$ implicitly depends on $\mathbf{Y}$}. Therefore, an attacker's goal is to craft the perturbed values $\mathbf{Y}$ to maximize the overall gain. Formally, the attacker aims to solve the following optimization problem:
{\small{\begin{align}\label{eq:objective}
	\max_{\mathbf{Y}} \quad G(\mathbf{Y}).
\end{align}}}%

We note that, to incorporate the different priorities of the target items, an attacker could also assign different weights to the expected frequency gains $\mathbb{E}[\Delta\tilde{f}_t]$ of different target items when calculating the overall gain. Our attacks are also applicable to such scenarios. However, for simplicity, we assume the target items have the same priority. 

\subsection{Three Attacks}

We propose three attacks:  \emph{Random perturbed-value attack (RPA)}, \emph{random item attack (RIA)}, and \emph{Maximal gain attack (MGA)}. RPA selects a perturbed value from the encoded space of the LDP protocol uniformly at random for each fake user and sends it to the server. RPA does not consider any information about the target items. RIA  selects a  target item from the set of target items uniformly at random for each fake user and uses the LDP protocol to encode and perturb the item. MGA  crafts the perturbed value for each fake user to maximize the overall gain $G$ via solving the optimization problem in  Equation (\ref{eq:objective}).  RPA and RIA are two baseline attacks, which are designed to better demonstrate the effectiveness of MGA. 

\myparatight{Random perturbed-value attack (RPA)} 
For each fake user, RPA selects a value from the encoded space of the LDP protocol uniformly at random and sends it to the server.

\myparatight{Random item attack (RIA)} Unlike RPA, RIA considers information about the target items. In particular, RIA  randomly selects a  target item from the set of target items for each fake user. Then, the LDP protocol is applied to encode and perturb the item. Finally, the perturbed value is sent to the server. 

\myparatight{Maximal gain attack (MGA)} The idea behind this attack is to craft the perturbed values for the fake users via solving the optimization problem in Equation (\ref{eq:objective}). Specifically, according to Equation (\ref{aggregate}), the  frequency gain $\Delta\tilde{f}_t$ for a target item $t$ is:
{\small{\begin{align}
    \Delta\tilde{f}_t &= \frac{\frac{1}{n+m}\sum\limits_{i=1}^{n+m} \mathbbm{1}_{S({y}_i)}(t) - q}{p-q} 
    \ - \frac{\frac{1}{n}\sum\limits_{i=1}^{n} \mathbbm{1}_{S({y}_i)}(t) - q}{p-q} \\
    \label{frequencydifference}
    &=\frac{\sum\limits_{i=n+1}^{n+m} \mathbbm{1}_{S({y}_i)}(t)}{(n+m)(p-q)}  -  \frac{m\sum\limits_{i=1}^{n}\mathbbm{1}_{S({y}_i)}(t)}{n(n+m)(p-q)},
\end{align}}}%
where $y_i$ is the perturbed value sent from user $i$ to the server. The first term in Equation (\ref{frequencydifference}) only depends on fake users, while the second term only depends on genuine users. Moreover, the expected frequency gain  for a target item $t$ is:
{\small{\begin{align}
\label{expecteddifference}
    \mathbb{E}[\Delta\tilde{f}_t] =\frac{\sum\limits_{i=n+1}^{n+m} \mathbb{E}[\mathbbm{1}_{S({y}_i)}(t)]}{(n+m)(p-q)} - \frac{m\sum\limits_{i=1}^{n}\mathbb{E}[\mathbbm{1}_{S({y}_i)}(t)]}{n(n+m)(p-q)}, 
\end{align}}}%
where we denote the second term as a constant $c_t$ for simplicity. Moreover, based on Equation (\ref{unbiasedestimator}), we have:
{\small{\begin{align}
 c_t=\frac{m(f_t(p-q)+q)}{(n+m)(p-q)},
 \end{align}}}%
where $f_t$ is the true frequency of $t$ among the $n$ genuine users. Furthermore, we have the overall gain as follows:
{\small{\begin{align}
\label{overallgaindeter}
    G =\frac{\sum\limits_{i=n+1}^{n+m} \sum\limits_{t\in T}\mathbb{E}[\mathbbm{1}_{S({y}_i)}(t)]}{(n+m)(p-q)} - c, 
\end{align}}}%
where $c=\sum_{t\in T}c_t=\frac{m(f_T(p-q)+rq)}{(n+m)(p-q)}$, where $f_T=\sum_{t\in T}f_t$. $c$ does not depend on the perturbed values sent from the fake users to the central server.  In RPA and RIA,  the crafted  perturbed values for the fake users are random. Therefore,  the expectation of the characteristic function $\mathbb{E}[\mathbbm{1}_{S({y}_i)}(t)]$ and the overall gain depend on such randomness. However, MGA uses the optimal perturbed values for fake users, and the characteristic function $\mathbbm{1}_{S({y}_i)}(t)$ becomes deterministic. Therefore, for MGA, we can drop the expectation $\mathbb{E}$ in Equation (\ref{overallgaindeter}), and then we can transform the optimization problem in Equation (\ref{eq:objective}) as follows:
{\small{\begin{align}
\label{transformed}
	\mathbf{Y}^*=\argmax_{\mathbf{Y}} G(\mathbf{Y})
	=\argmax_{\mathbf{Y}} \quad \sum\limits_{i=n+1}^{n+m} \sum\limits_{t\in T}\mathbbm{1}_{S({y}_i)}(t),
\end{align}}}%
where we remove the constants $c$ and $(n+m)(p-q)$ in the optimization problem. 
Note that the above optimization problem only depends on the perturbed values of the fake users, and the perturbed values ${y}_i$ for the fake users are independent from each other. Therefore, we can solve the optimization problem  independently for each fake user. Formally, for each fake user, we craft its perturbed value $y^*$ via solving the following optimization problem:
 {\small{\begin{align}\label{eq:explicit_obj}
	{y}^* = \argmax\limits_{{y}\in\mathcal{D}} \sum\limits_{t\in T}\mathbbm{1}_{S({y})}(t).
\end{align}}}%

We note that, for each fake user, we obtain its perturbed value via solving the same above optimization problem. However, as we will show in the next sections, the optimization problem has many optimal solutions. Therefore, we randomly pick an optimal solution for a fake user. 

Next, we  discuss how to apply these three attacks  to state-of-the-art LDP protocols including kRR, OUE, and OLH, as well as analyzing their overall gains. 

\subsection{Attacking kRR}

\noindent
{\bf Random perturbed-value attack (RPA):} For each fake user, RPA randomly selects a perturbed value $y_i$ from the encoded space, i.e., $[d]$, and sends it to the server. 
We can calculate the expectation of the characteristic function for $t\in T$ as follows:
{\small{\begin{align}
\mathbb{E}[\mathbbm{1}_{S({y}_i)}(t)]&=\text{Pr}(\mathbbm{1}_{S({y}_i)}(t)=1)\\
&=\text{Pr}(t\in S(y_i))=\text{Pr}(y_i=t) \\
&= \frac{1}{d}
\end{align}}}%
Therefore, according to Equation (\ref{overallgaindeter}),   the overall  gain is $G = \frac{rm}{d(n+m)(p-q)}-c$.

\myparatight{Random item attack (RIA)} For each fake user, RIA randomly selects an item $t_i$ from the set of target items $T$, perturbs the item  following the rule in Equation (\ref{krr-perturb}), and sends the perturbed item $y_i$ to the server. 
First, we can calculate the expectation of the characteristic function as follows:
{\small{\begin{align}
\mathbb{E}[\mathbbm{1}_{S({y}_i)}(t)]&=\text{Pr}(y_i=t) \\
&=\text{Pr}(t_i=t) \text{Pr}(y_i=t|t_i=t) \nonumber \\
&+ \text{Pr}(t_i\neq t) \text{Pr}(y_i=t|t_i\neq t) \\
&= \frac{1}{r}\cdot p + (1-\frac{1}{r}) q,
\end{align}}}%
where $r=|T|$ is the number of  target items. 
According to Equation (\ref{overallgaindeter}), we can obtain the overall gain as $G =\frac{(p+(r-1)q)m}{(n+m)(p-q)}-c$.

\begin{table*}[!ht]\renewcommand{\arraystretch}{1.3}
\centering
    \begin{tabular}{|c|c|c|c|}	
	\hline
         	& \small{kRR} & \small{OUE} &\small{OLH} \\
	\hline
	\small{Random perturbed-value attack (RPA)} & \small{${\beta(\frac{r}{d}-f_T)}$} & \small{$\beta(r-f_T)$} & \small{$-\beta f_T$}\\
	\hline
         	\small{Random item attack (RIA)} & \small{$\beta(1-f_T)$} & \small{$\beta(1-f_T)$} & \small{$\beta(1-f_T)$}\\
	\hline
	\small{Maximal gain attack (MGA)} & \small{$\beta(1-f_T)+\frac{\beta (d-r)}{e^\epsilon-1}$} & \small{$\beta(2r-f_T)+\frac{2\beta r}{e^\epsilon-1}$} & \small{$\beta(2r-f_T)+\frac{2\beta r}{e^\epsilon-1}$}\\
	\hline
	\small{{Standard deviation of estimation}} & \small{{$\frac{r\sqrt{d-2+e^\epsilon}}{(e^\epsilon-1)\sqrt{n}}$}} & \small{{$\frac{2re^{\epsilon/2}}{(e^\epsilon-1)\sqrt{n}}$}} & \small{{$\frac{2re^{\epsilon/2}}{(e^\epsilon-1)\sqrt{n}}$}}\\
	\hline
    \end{tabular}
    \caption{Overall gains of the three attacks for kRR, OUE, and OLH.  $n$ is the number of genuine users, $\beta=\frac{m}{n+m}$ is the fraction of fake users among all users, $d$ is the number of items, $r$ is the number of  target items, $f_T=\sum_{t\in T}f_t$ is the sum of true frequencies of the target items among the genuine users, $\epsilon$ is the privacy budget, and $e$ is the  base of the natural logarithm. {To understand the significance of the overall gains, we also include the standard deviations of the estimated total  frequencies of the target items among the $n$ genuine users~\cite{wang2017locally} in the table.} }
    \label{tab:exp_gain}
\end{table*}

\myparatight{Maximal gain attack (MGA)}
For each fake user, MGA crafts its perturbed value by solving the optimization problem in Equation (\ref{eq:explicit_obj}). 
For the kRR protocol, we have $\sum_{t\in {T}}\mathbbm{1}_{S({y})}(t) \le 1$ and $\sum_{t\in {T}}\mathbbm{1}_{S({y})}(t) = 1$ when $y$ is a  target item in $T$. Therefore, MGA picks any  target item for each fake user. Moreover, according to Equation (\ref{overallgaindeter}), the overall gain is $G = \frac{m}{(n+m)(p-q)}-c$.

\subsection{Attacking OUE}
\label{attackOUE}
\noindent
{\bf Random perturbed-value attack (RPA):} 
For each fake user, RPA  selects a $d$-bits binary vector $\bm{y}_i$ from the encoded space $\{0,1\}^d$ uniformly at random as its perturbed vector and sends it to the server. We denote by $y_{i,j}$  the $j$-th bit of the perturbed vector $\bm{y}_i$. Therefore, for each target item $t\in T$, we have $\mathbb{E}[\mathbbm{1}_{S(\bm{y}_i)}(t)]=\text{Pr}(y_{i,t}=1) = \frac{1}{2}$. 
According to Equation (\ref{overallgaindeter}),  we can obtain the overall gain as $G = \frac{rm}{2(n+m)(p-q)} - c$. 

\myparatight{Random item attack (RIA)} For each fake user, RIA randomly selects a  target item $t_i\in T$,  encodes it to a $d$-bits binary vector $\mathbf{e}_i$ whose bits are all zeros except the $t_i$-th bit,  randomly perturbs $\mathbf{e}_i$ following Equation (\ref{eq:perturb_oue}), and sends the perturbed vector $\bm{y}_i$ to the server.  For a target item $t\in T$,  we can calculate the  expected value of the characteristic function as follows:
{\small{\begin{align}
\mathbb{E}[\mathbbm{1}_{S(\bm{y}_i)}(t)]&=\text{Pr}(y_{i,t}=1) \\
&=\text{Pr}(t_i=t) \text{Pr}(y_{i, t}=1|t_i=t) \nonumber \\
&+ \text{Pr}(t_i\neq t) \text{Pr}(y_{i, t}=1|t_i\neq t) \\
&= \frac{1}{r}\cdot p + (1-\frac{1}{r}) \cdot q,
\end{align}}}%
where $p$ and $q$ are defined in Equation (\ref{eq:perturb_oue}).  
Therefore, the overall gain  is $G =\frac{(p+(r-1)q)m}{(n+m)(p-q)}-c$. 

\myparatight{Maximal gain attack (MGA)} For each fake user, MGA chooses a perturbed vector $\bm{y}_i$ that is a solution of the optimization problem defined in Equation (\ref{eq:explicit_obj}). For OUE, we have 
$\sum_{t\in T}\mathbbm{1}_{S(\bm{y}_i)}(t) \le r$ and $\sum_{t\in T}\mathbbm{1}_{S(\bm{y}_i)}(t) = r$ is achieved when $\mathbbm{1}_{S(\bm{y}_i)}(t)=1, \forall t\in T$.  
Thus, for each fake user, MGA  initializes a perturbed vector $\bm{y}_i$ as a binary vector of all 0's and  sets $y_{i,t}=1$ for all $t\in T$. However, if all fake users send the same perturbed binary vector to the server, the server can easily detect the fake users. 
{For instance,  there is only one entry in the perturbed binary vector that has value 1 when we only have 1 target item; and  the server could detect a vector with only a single 1 to be from a fake user, because it is statistically unlikely for a genuine user to send such a vector.} Therefore, MGA  also randomly samples $l$ non-target bits of the perturbed vector $\bm{y}_i$ and sets them to 1. Specifically, we set $l$ such that the number of 1's in the binary vector is the expected number of 1's in  the perturbed binary vector of a genuine user. Since the perturbed binary vector of a genuine user has $p+(d-1)q$ 1's on average, we set $l=\floor{p+(d-1)q-r}$. Note that $r$ is usually much smaller than $d$, so $l$ is a non-negative value. The final binary vector is sent to the server. According to Equation (\ref{overallgaindeter}), the overall gain is $G = \frac{rm}{(n+m)(p-q)}-c$. 

\subsection{Attacking OLH}

\noindent
{\bf Random perturbed-value attack (RPA):} 
For each fake user, RPA randomly selects a hash function  $H_i\in\mathbf{H}$ and a hash value $a_i\in [d']$, and sends the tuple $y_i=(H_i, a_i)$ to the server. For each $t\in T$, we have $\mathbb{E}[\mathbbm{1}_{S(\bm{y}_i)}(t)]=\text{Pr}(H_i(t)=a_i)= \frac{1}{d'}$.  
Therefore, we can obtain the overall gain as $G = \frac{rm}{d'(n+m)(p-q)} - c$. 

\myparatight{Random item attack (RIA)} For each fake user, RIA randomly selects a  target item $t_i$,  randomly selects a hash function $H_i\in\mathbf{H}$, and calculates the hash value $h_i=H_i(t_i)$. The tuple $(H_i, h_i)$ is then perturbed as $(H_i, a_i)$ according to Equation (\ref{eq:perturb_olh}). $(H_i, a_i)$ is the perturbed value, i.e., $y_i=(H_i, a_i)$. We assume the hash function $H_i$ maps any item in $[d]$ to a value in $[d']$ uniformly at random. For a target item $t\in T$, we can calculate the expectation of the characteristic function as follows:
{\small{\begin{align}
\mathbb{E}[\mathbbm{1}_{S({y}_i)}(t)]&=\text{Pr}(H_i(t)=a_i) \\
&= \text{Pr}(t_i=t)\text{Pr}(H_i(t)=a_i|t_i=t)\nonumber\\
&+ \text{Pr}(t_i \neq t)\text{Pr}(H_i(t)=a_i|t_i\neq t)\\
&= \frac{1}{r}\cdot p + (1-\frac{1}{r})\cdot q.
\end{align}}}%
Thus, the overall gain for RIA is $G=\frac{[p+(r-1)q]m}{(n+m)(p-q)} - c$. 

\myparatight{Maximal gain attack (MGA)} 
 For each fake user, MGA chooses a perturbed value ${y}_i=(H_i, a_i)$ that is a solution of the optimization problem defined in Equation (\ref{eq:explicit_obj}). For OLH, we have 
$\sum_{t\in T}\mathbbm{1}_{S({y}_i)}(t) \le r$ and $\sum_{t\in T}\mathbbm{1}_{S({y}_i)}(t) = r$ is achieved when  the hash function $H_i$ maps all items in $T$ to $a_i$, i.e., $H_i(t)=a_i, \forall t\in T$. Thus, for each fake user, MGA  searches for a hash function $H_i$ in $\mathbf{H}$ such that $H_i(t)=a_i, \forall t\in T$ holds. 
Therefore, according to Equation (\ref{overallgaindeter}), the overall gain is $G = \frac{rm}{(n+m)(p-q)}-c$. Note that we may not be able to find such a hash function in practice. In our experiments, for each fake user, we randomly sample 1,000 hash functions and use the one that hashes the most target items to the same value. 

\subsection{Theoretical Analysis}

\label{sec:theretic_analysis}

Table~\ref{tab:exp_gain} summarizes the overall gains of the three attacks for  kRR, OUE, and OLH, where we have replaced the parameters $p$ and $q$ for each LDP protocol according to Section~\ref{background}. Next, we compare the three attacks, discuss a fundamental security-privacy tradeoff, and compare the three LDP protocols with respect to their security against our data poisoning attacks. 

\myparatight{Comparing the three attacks} All three attacks achieve larger overall gains when  the target items' true frequencies are smaller (i.e., $f_T$ is smaller). MGA achieves the largest overall gain among the three attacks. In fact, given an LDP protocol, a set of target items and fake users,  MGA achieves the largest overall gain among all possible attacks. This is because MGA crafts the perturbed values for the fake users such that the overall gain is maximized. RIA achieves larger overall gains than RPA for kRR and OLH, while RPA achieves a larger overall gain than RIA for OUE.

{Table~\ref{tab:exp_gain} also includes the standard deviations of the estimated total frequencies of the target items among the $n$ genuine users. Due to the $\sqrt{n}$ term in the denominators, the standard deviations are much smaller than the overall gains of our MGA attacks. For instance, on the Zipf dataset in our experiments with the default parameter settings, the overall gains of  MGA  are 1600, 82, and 82 times larger than the standard deviations for kRR, OUE, and OLH, respectively.} 

\myparatight{Fundamental security-privacy tradeoffs} The security of an LDP protocol is determined by the strongest attack (i.e., MGA) to it. Intuitively, when the privacy budget $\epsilon$ is smaller (i.e., stronger privacy), genuine users add larger noise to their data.  However, the perturbed values that MGA crafts for the fake users do not depend on the privacy budget. As a result, the fake users contribute more towards the estimated item frequencies, making the overall gain larger. In other words, we have a fundamental security-privacy tradeoff. Formally, the following theorem shows such tradeoffs. 
        \begin{restatable}[Security-Privacy Tradeoff]{thm}{tradeoff}\label{thm:tradeoff}
            For any of the three LDP protocols kRR, OUE, and OLH, when the  privacy budget $\epsilon$ is smaller (i.e., stronger privacy), MGA achieves a larger overall gain $G$ (i.e., weaker security).   
        \end{restatable}
        \begin{proof}
\vspace{-1mm}
        Table~\ref{tab:exp_gain} shows that  $\epsilon$ is in the denominator of the overall gains for MGA. Therefore,  the overall gains of MGA increase as $\epsilon$ decreases. 
\vspace{-1mm}
        \end{proof}

\myparatight{Comparing the security of the three LDP protocols} Table~\ref{tab:exp_gain} shows that, when MGA is used, OUE and OLH achieve  the same overall gain. 
Therefore, OUE and OLH have  the same level of security against data poisoning attacks. The following theorem shows that OUE and OLH are more secure than kRR when the number of items is larger than a threshold. 

    \begin{restatable}{thm}{comp_protocols}\label{thm:comp_protocols}
        Suppose MGA is used. OUE and OLH are more secure than kRR when the number of items is larger than some threshold, i.e., $d > (2r-1)(e^\epsilon-1) + 3r$.  
    \end{restatable}

\begin{proof}
See Appendix \ref{proof:thm2}.
\end{proof}
    
%%%%%%%%%%%%%%%%%%%%%%%%%%%%%%%%%%%%%%%%%%%%%%%%%%%%%%%

\section{Attacking Heavy Hitter Identification}
\subsection{Threat model}

\myparatight{Attacker's capability and background knowledge} We make the same assumption on the attacker's capability and background knowledge as in attacking frequency estimation, i.e., the attacker can inject fake users into the protocol and send arbitrary data to the central server.  
 
\myparatight{Attacker's goal} We consider the attacker's goal is to {promote} some target items, i.e., manipulate the heavy hitter identification protocol to recognize the target items as top-$k$ heavy hitters. Formally, we denote by ${T}=\{t_1, t_2, \cdots, t_r\}$  the set of $r$ target items, which are not among the true top-$k$ heavy hitters. We define \emph{success rate} of an attack as the fraction of target items that are promoted to be top-$k$ heavy hitters by the attack. An attacker's goal is to achieve a high success rate. 

\subsection{Attacks}
State-of-the-art heavy hitter identification protocols iteratively apply frequency estimation protocols.  Therefore, we apply the three attacks for frequency estimation to heavy hitter identification. Next, we use PEM as an example to illustrate how to attack heavy hitter identification protocols. 

In PEM, each item is encoded by a $\gamma$-bits binary vector and users are randomly divided into $g$ groups. On average, each group contains a fraction of $\frac{m}{n+m}$ fake users. In the $j$th iteration, PEM uses OLH to perturb the first $\lambda_j$ bits of the binary vectors for users in the $j$th group and sends them to the central server. An attacker uses the RPA, RIA, or MGA to craft the data sent from the fake users to the central server by treating the first $\lambda_j$ bits of the binary vectors corresponding to the target items as the ``target items'' in the $j$th iteration. Such attacks can increase the likelihood that the first $\lambda_j$ bits of  the target items are identified as the top-$k$ prefixes in the $j$th iteration, which in turn makes it more likely to promote the target items as top-$k$ heavy hitters. 

%%%%%%%%%%%%%%%%%%%%%%%%%%%%%%%%%%%%%%%%%%%%%%%%%%%%%%%

\section{Evaluation}

\subsection{Experimental Setup}
\label{exp:setup}

\begin{figure*}[!t]
\vspace{-4mm}
	 \centering
\subfloat{\includegraphics[width=0.2\textwidth]{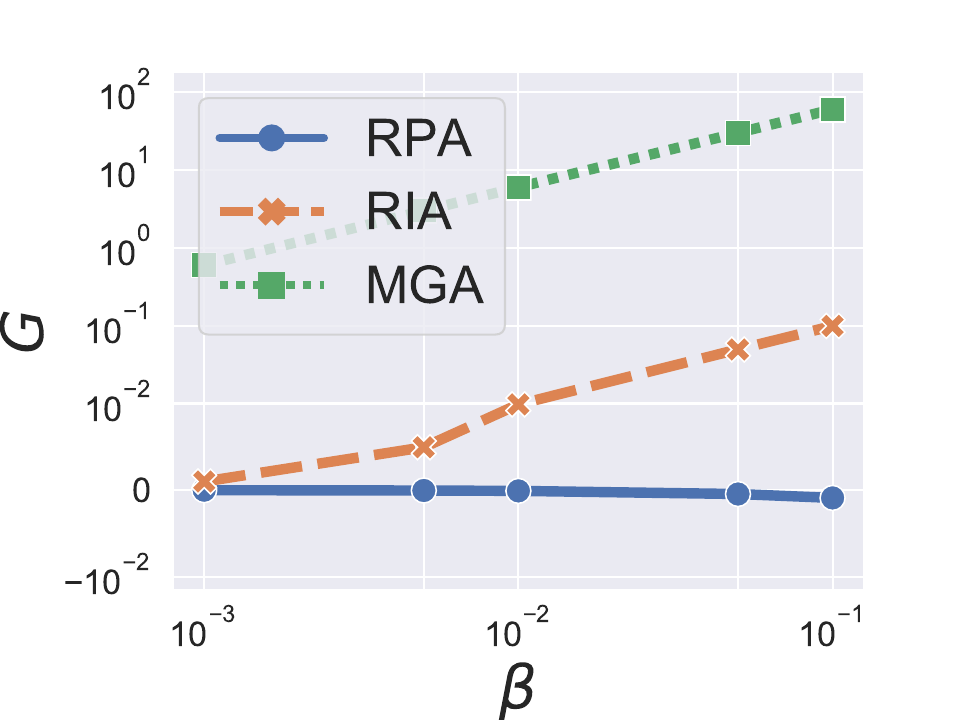}}
\subfloat{\includegraphics[width=0.2\textwidth]{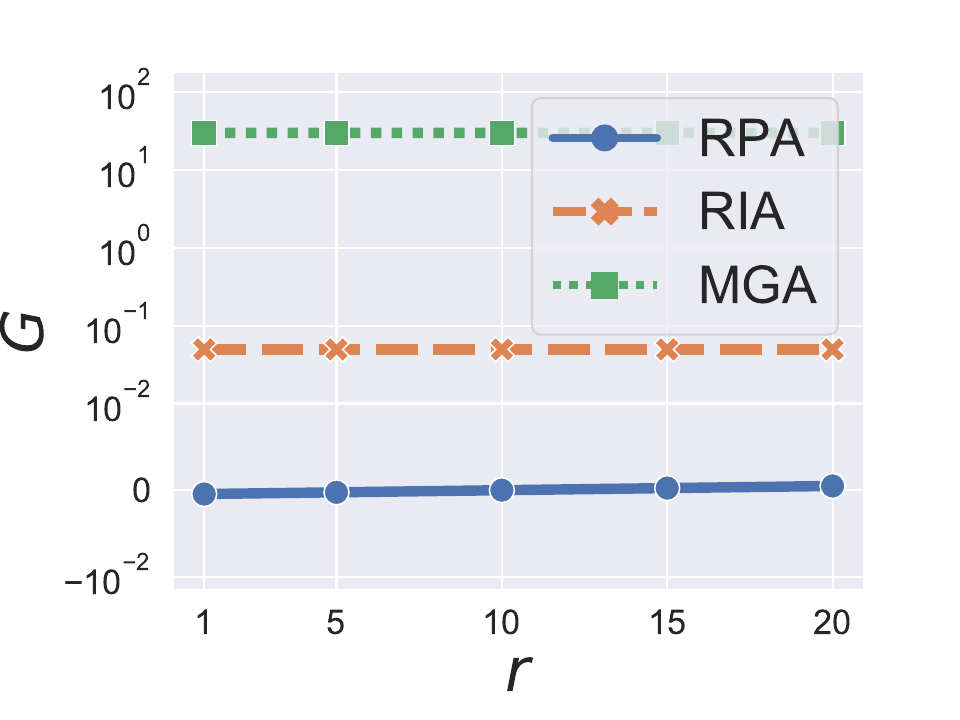}}
\subfloat{\includegraphics[width=0.2\textwidth]{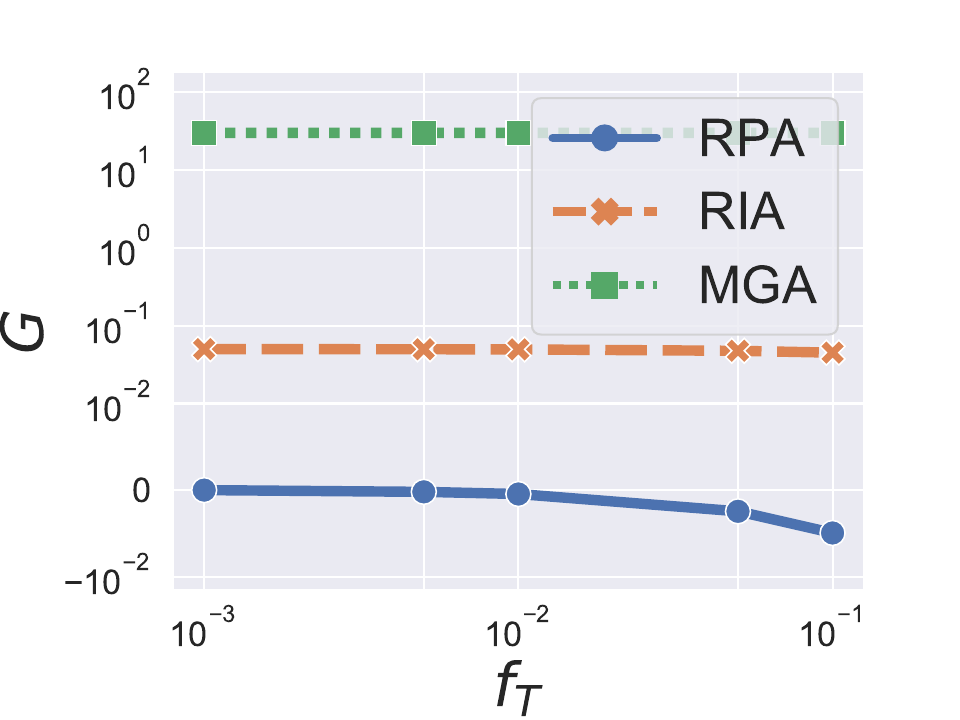}}
\subfloat{\includegraphics[width=0.2\textwidth]{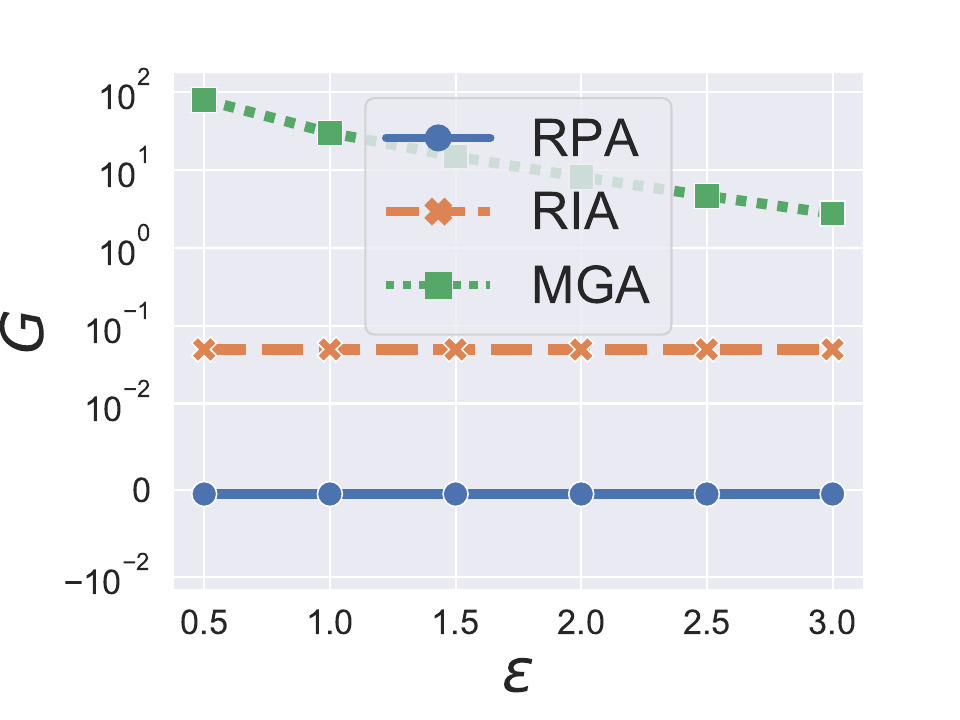}}
\subfloat{\includegraphics[width=0.2\textwidth]{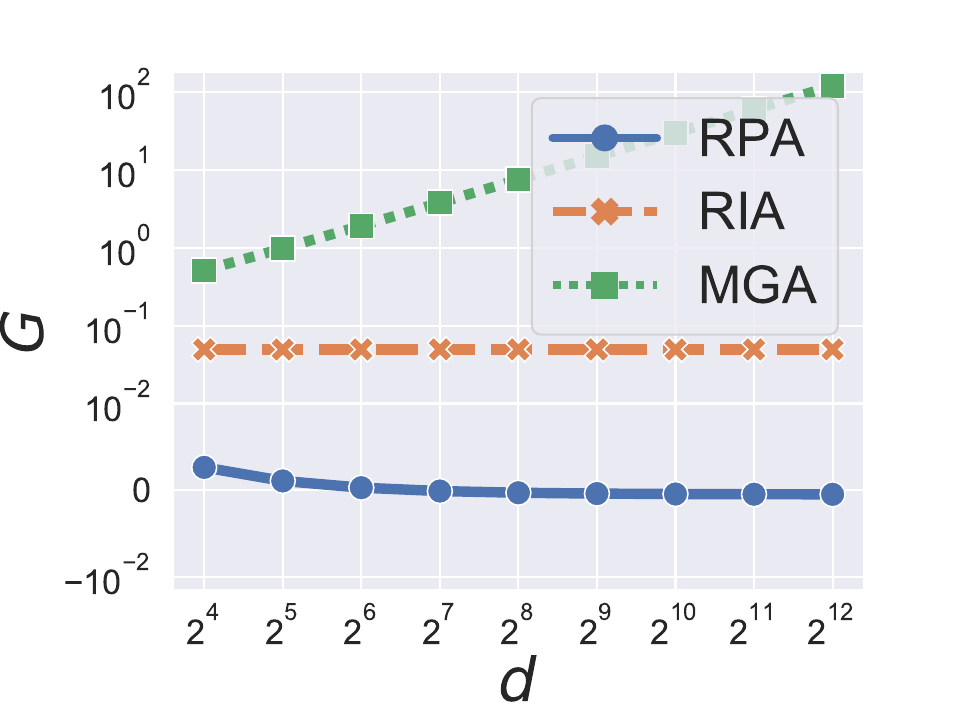}}
\vspace{-4mm}
\setcounter{subfigure}{0}
\subfloat{\includegraphics[width=0.2\textwidth]{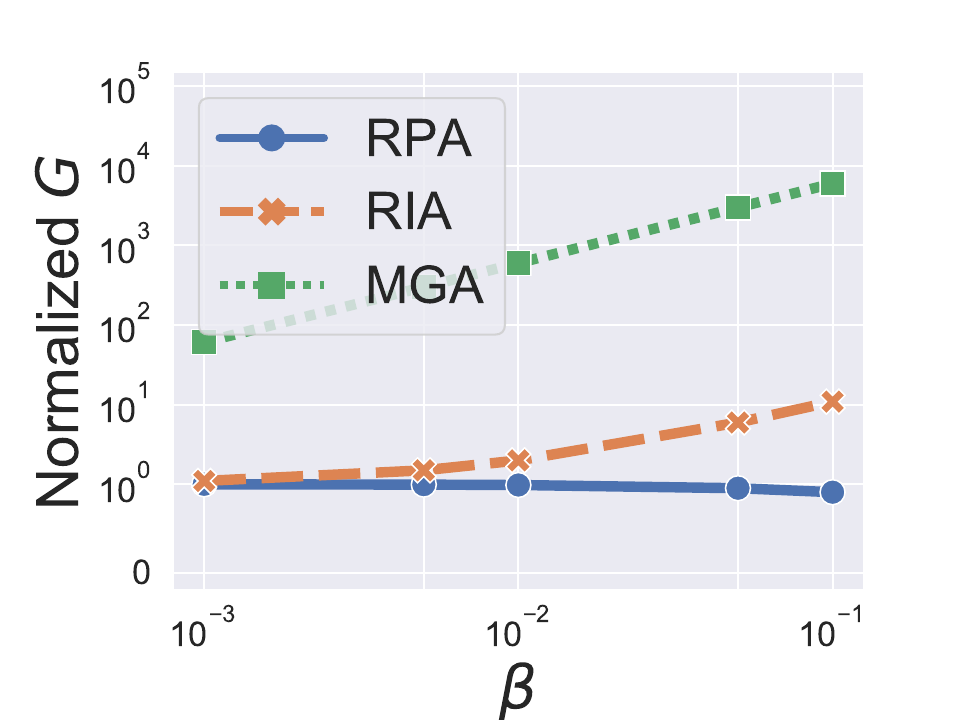}}
\subfloat{\includegraphics[width=0.2\textwidth]{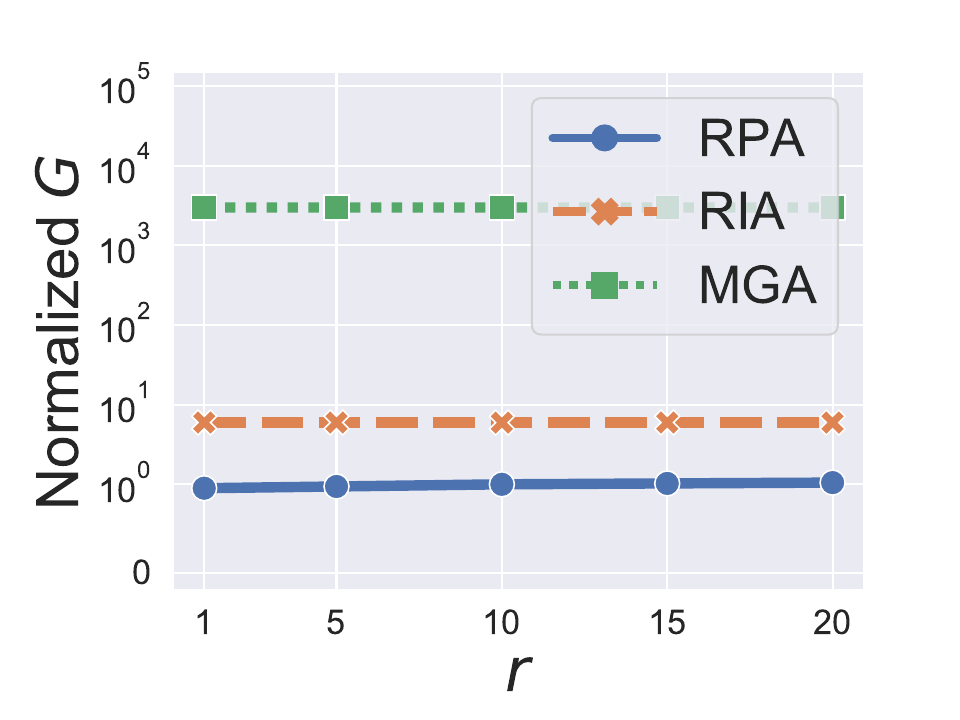}}
\subfloat{\includegraphics[width=0.2\textwidth]{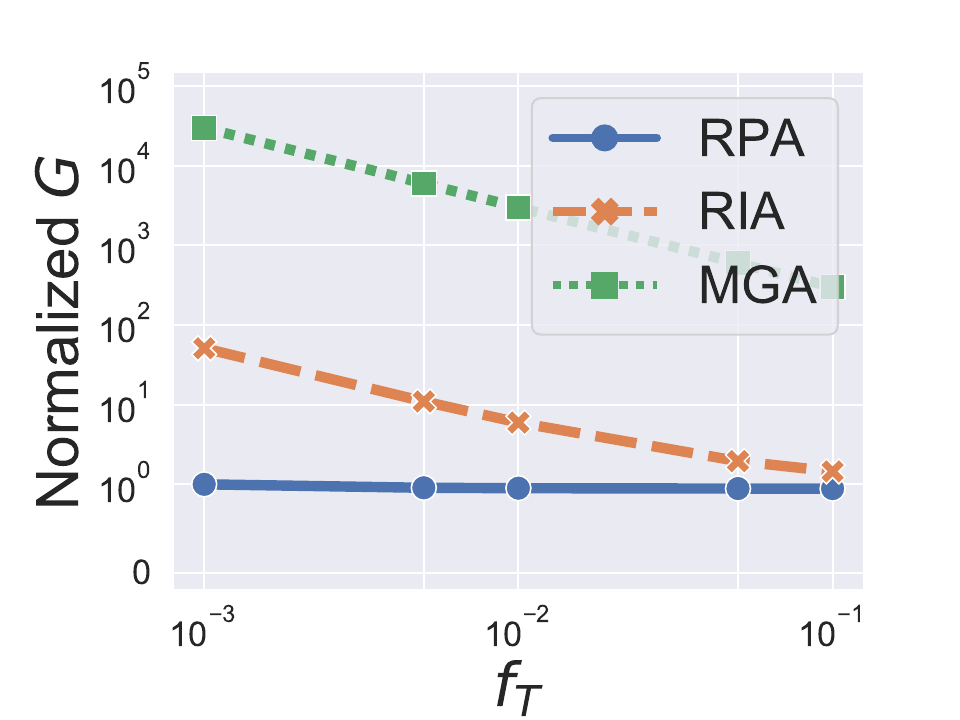}}
\subfloat{\includegraphics[width=0.2\textwidth]{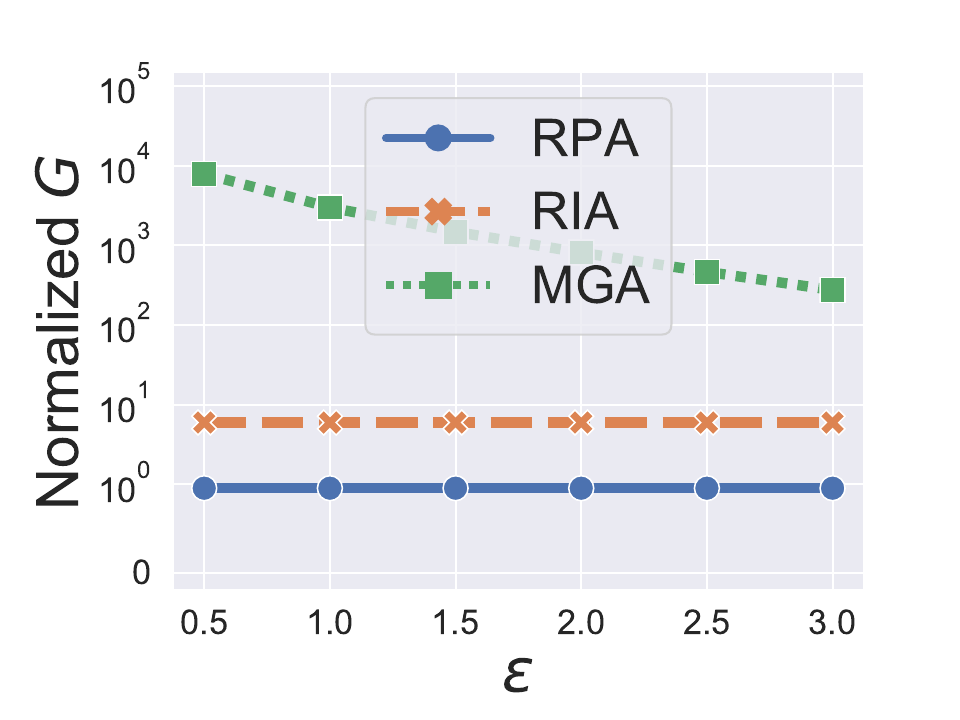}}
\subfloat{\includegraphics[width=0.2\textwidth]{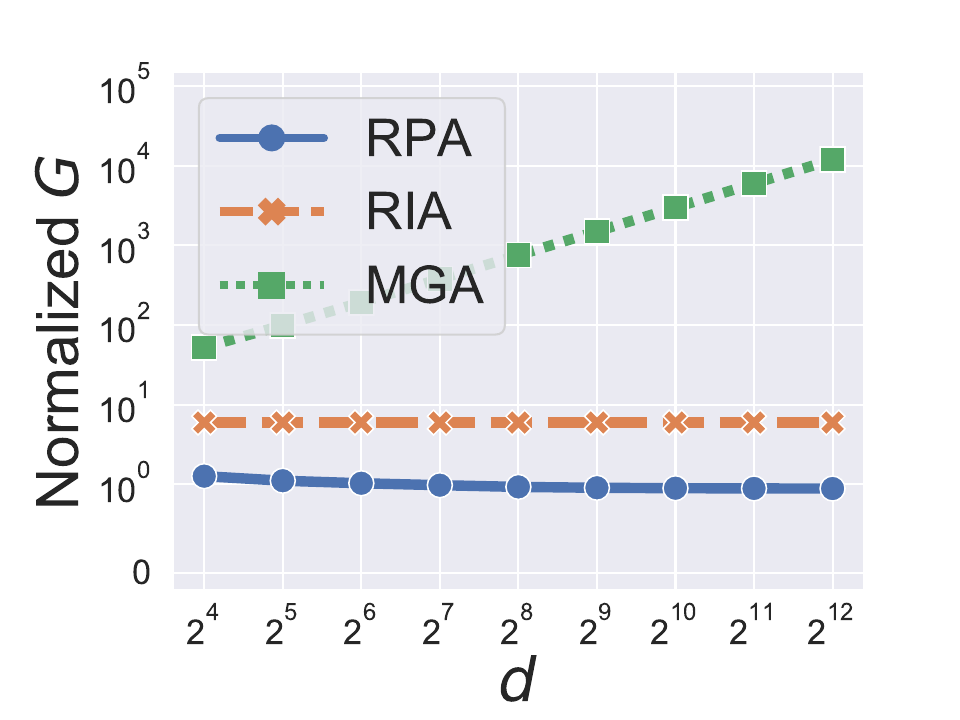}\label{fig:krr_d}}
\vspace{-3mm}
	 \caption{{Impact of different parameters on the overall gains (first row) and normalized overall gains (second row) of the three attacks for kRR.}}
	\label{prameterimpact_krr}
\vspace{-5mm}
\end{figure*}

\begin{figure*}[!t]
	 \centering
\subfloat{\includegraphics[width=0.2\textwidth]{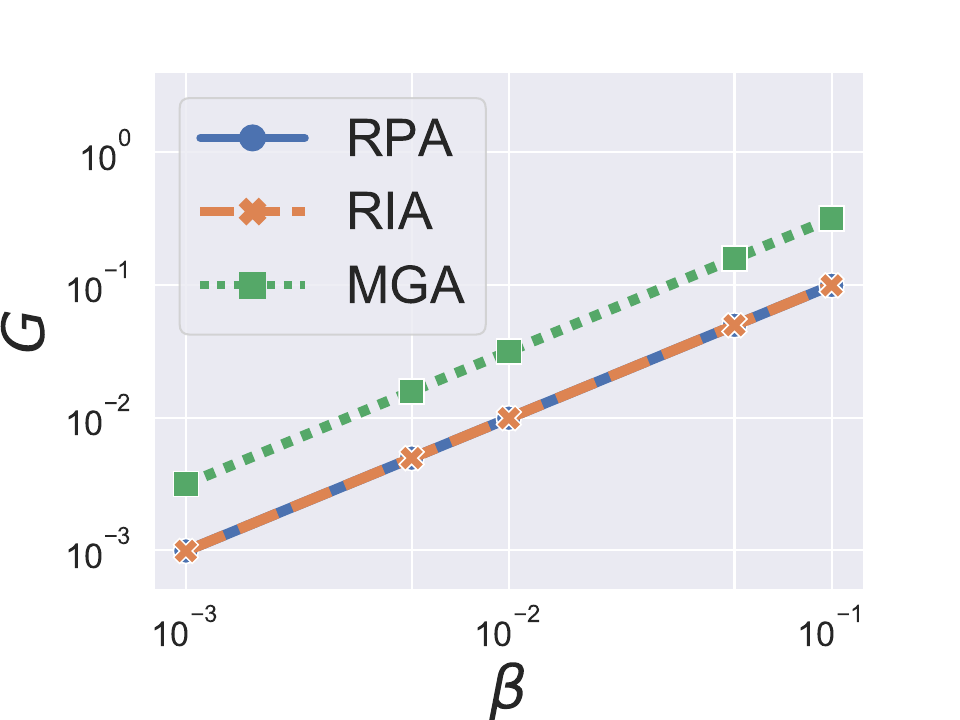}}
\subfloat{\includegraphics[width=0.2\textwidth]{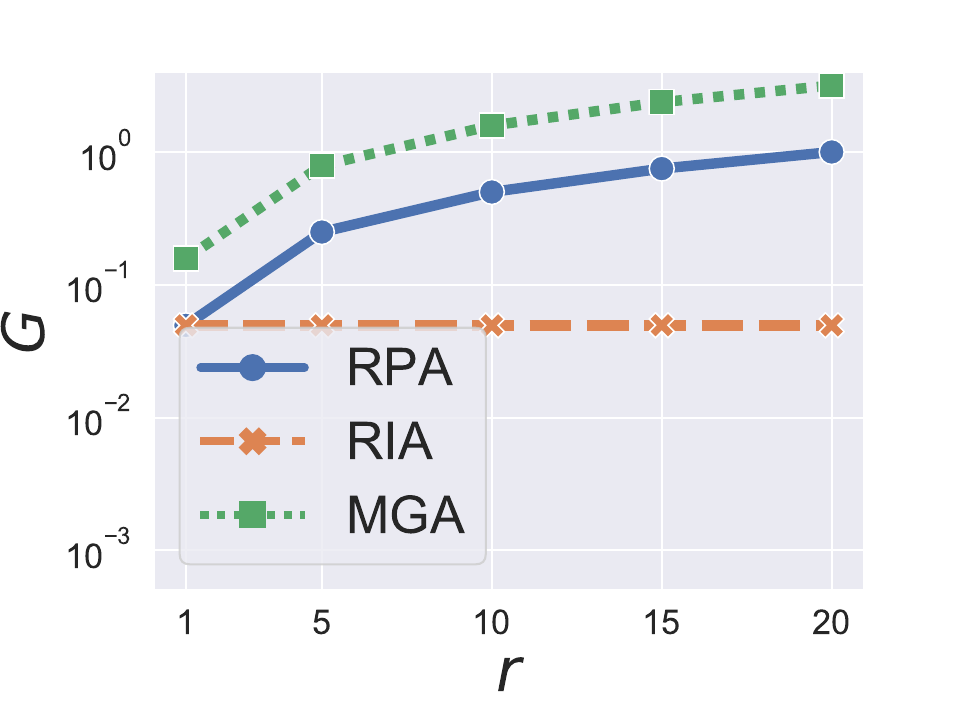}}
\subfloat{\includegraphics[width=0.2\textwidth]{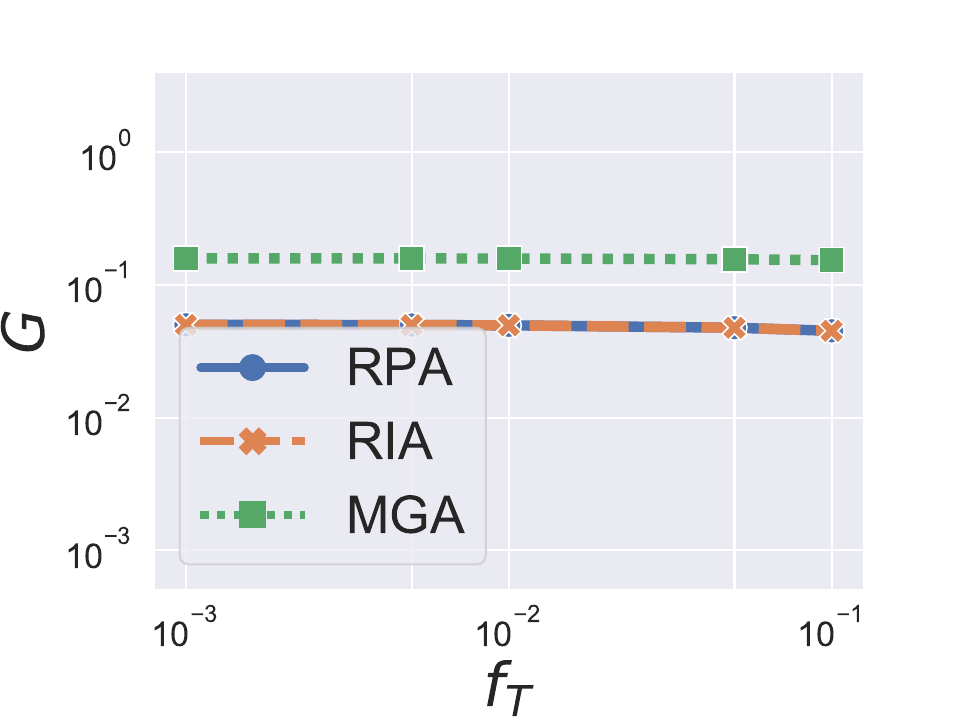}}
\subfloat{\includegraphics[width=0.2\textwidth]{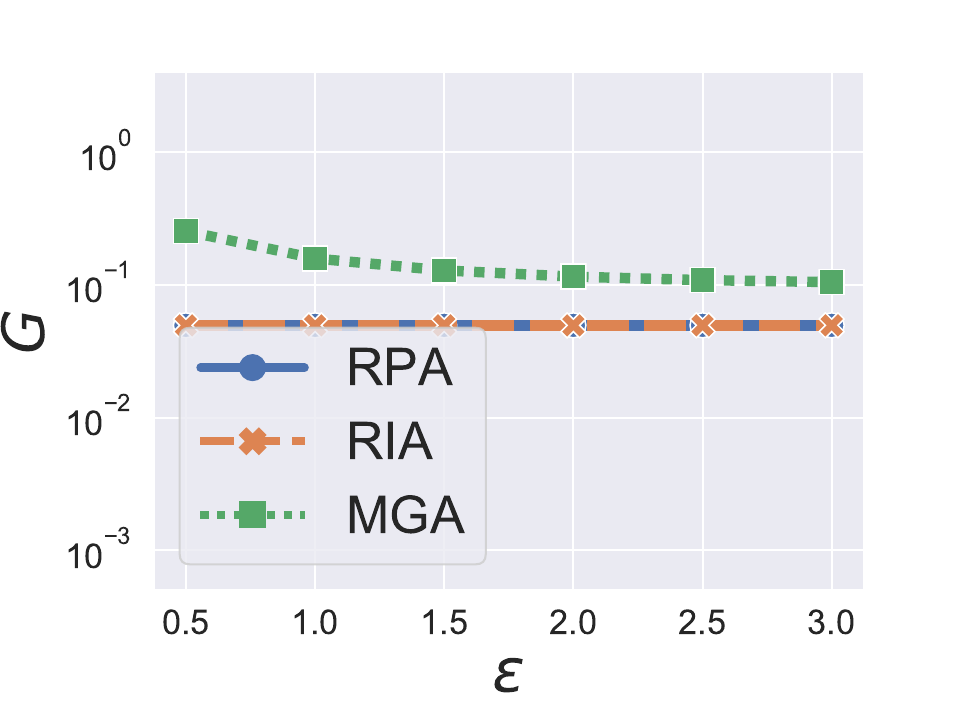}}
\subfloat{\includegraphics[width=0.2\textwidth]{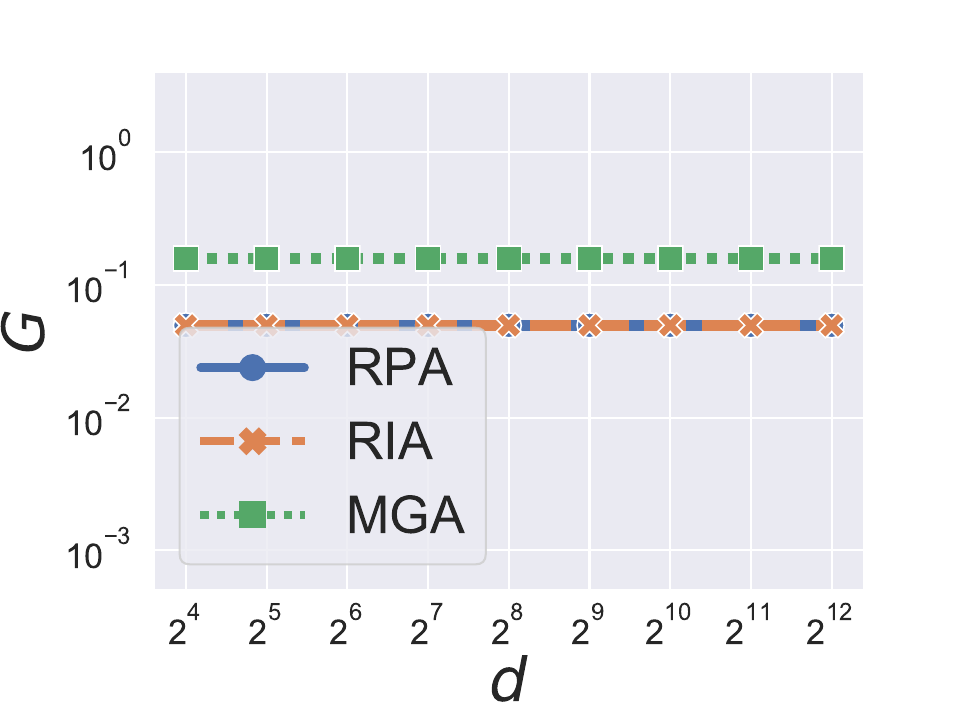}}
\vspace{-4mm}
\setcounter{subfigure}{0}
\subfloat{\includegraphics[width=0.2\textwidth]{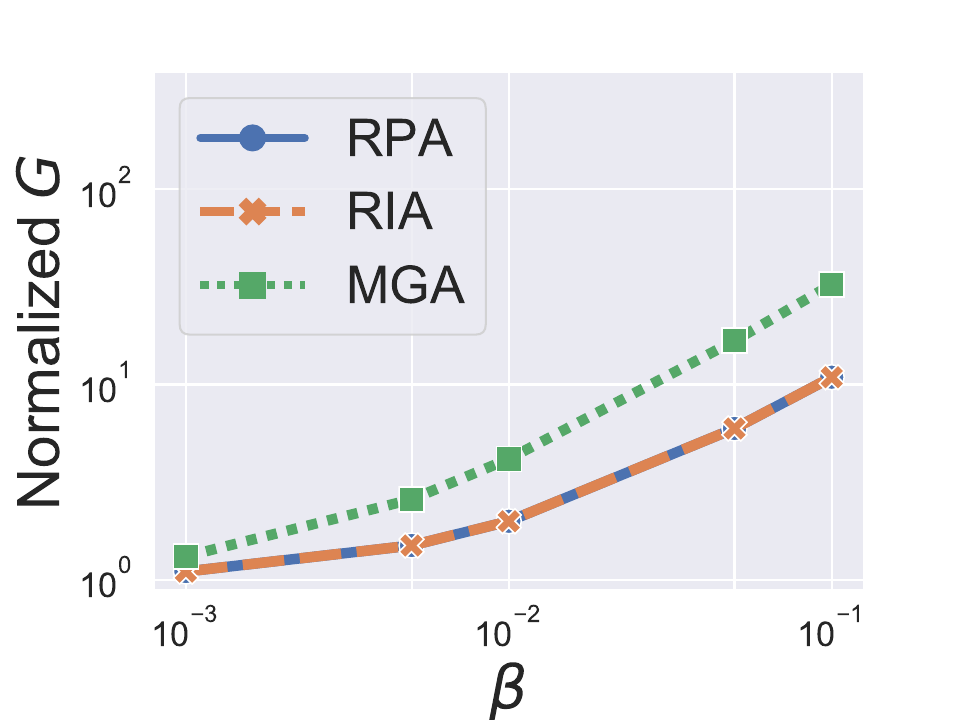}}
\subfloat{\includegraphics[width=0.2\textwidth]{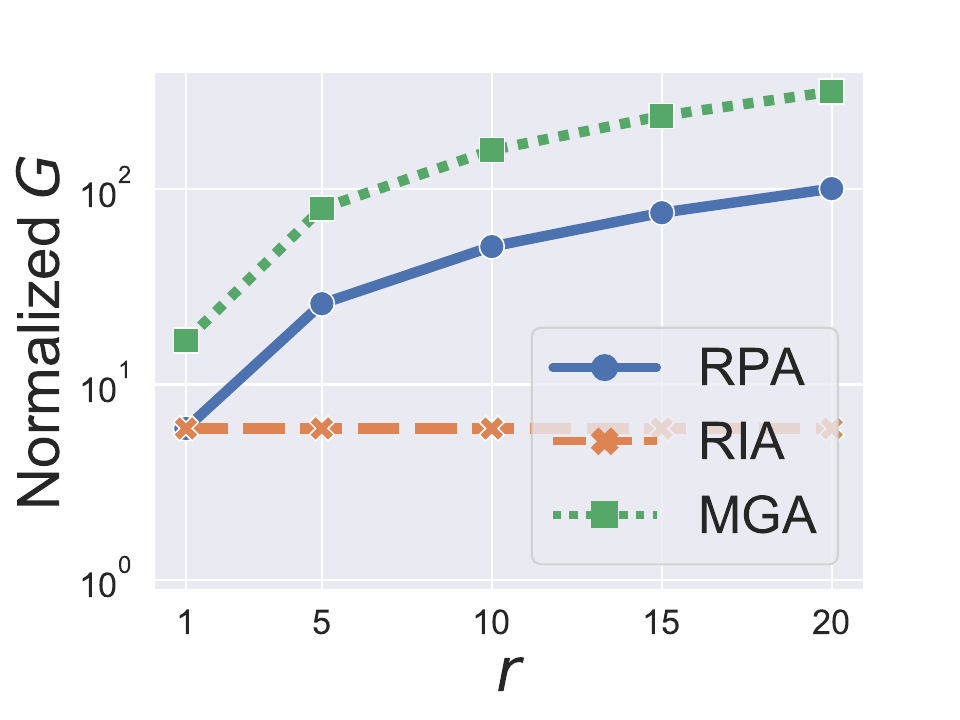}}
\subfloat{\includegraphics[width=0.2\textwidth]{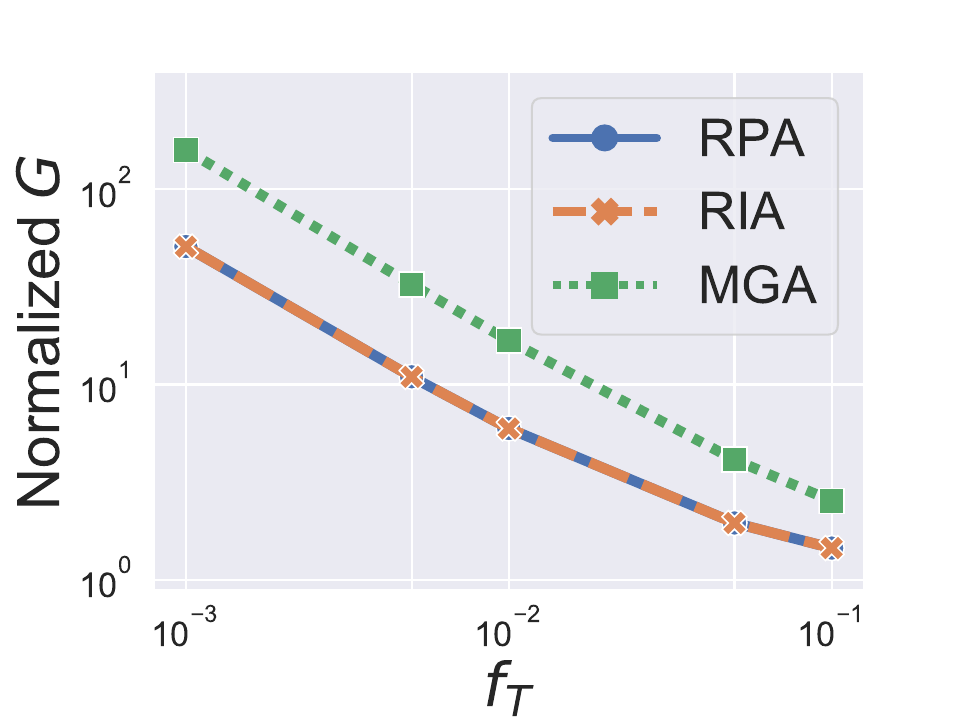}}
\subfloat{\includegraphics[width=0.2\textwidth]{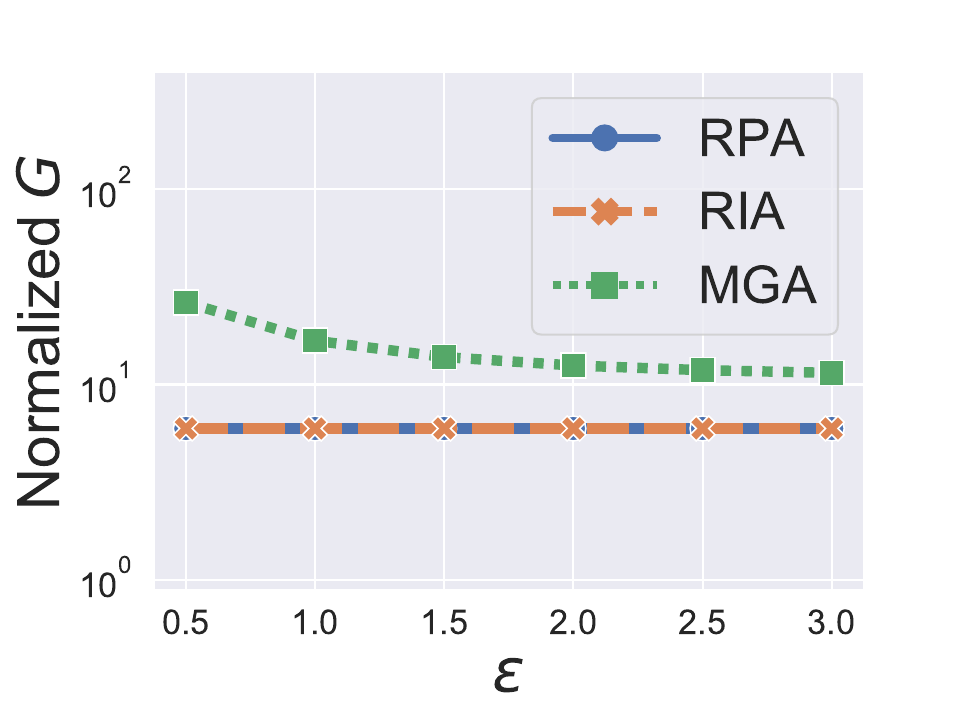}}
\subfloat{\includegraphics[width=0.2\textwidth]{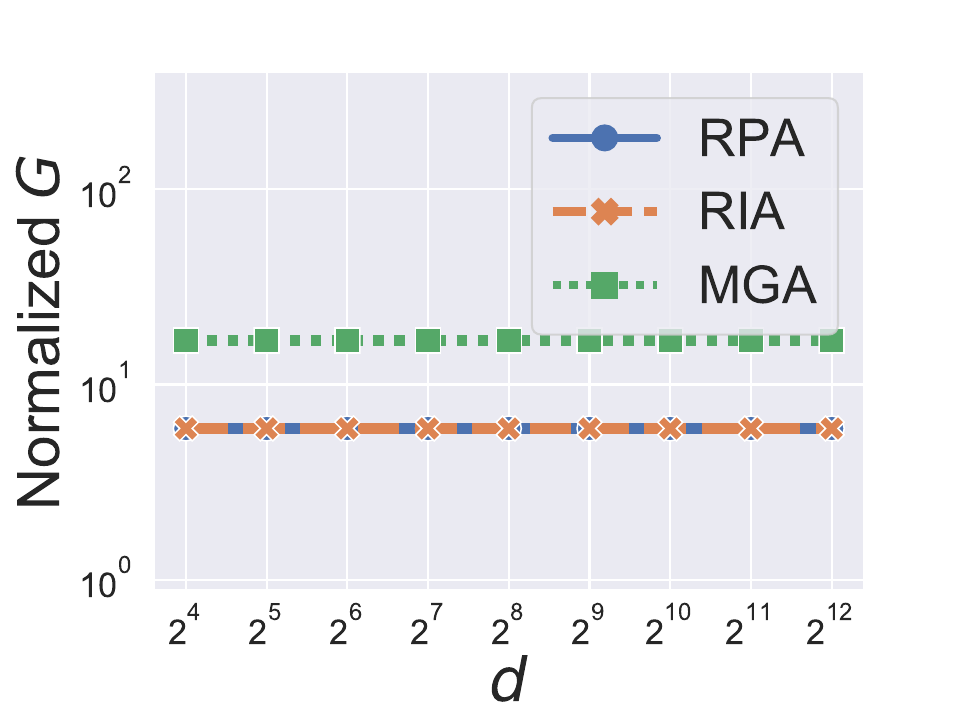}\label{fig:oue_d}}
\vspace{-3mm}
	 \caption{{Impact of different parameters on the overall gains (first row) and normalized overall gains (second row) of the three attacks for OUE.}}
	\label{prameterimpact_oue}
\vspace{-5mm}
\end{figure*}

\begin{figure*}[!t]
	 \centering
\subfloat{\includegraphics[width=0.2\textwidth]{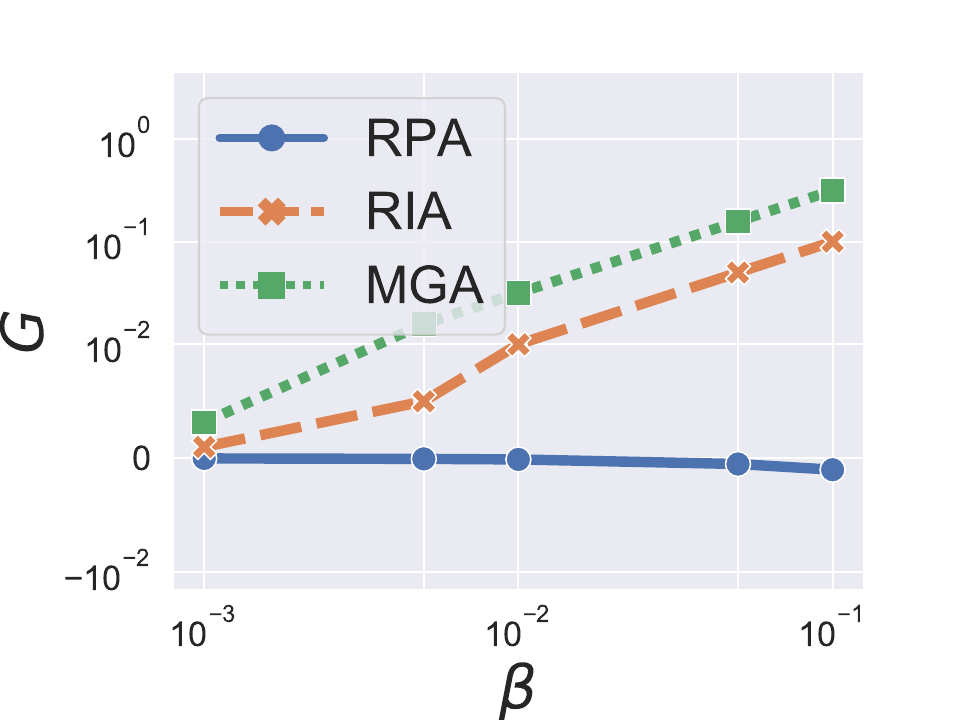}}
\subfloat{\includegraphics[width=0.2\textwidth]{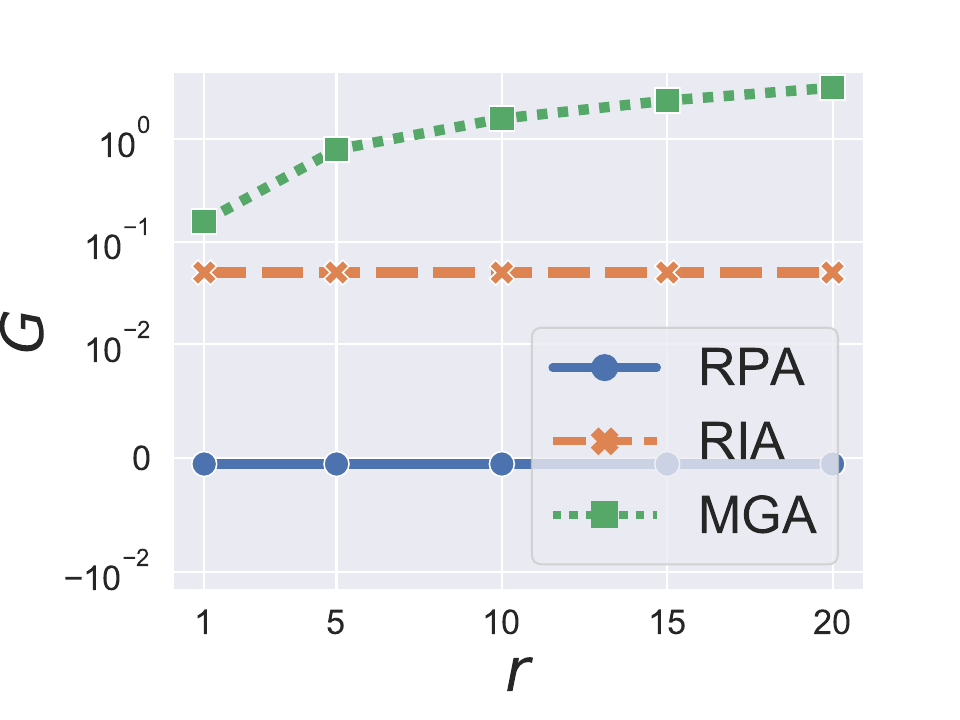}}
\subfloat{\includegraphics[width=0.2\textwidth]{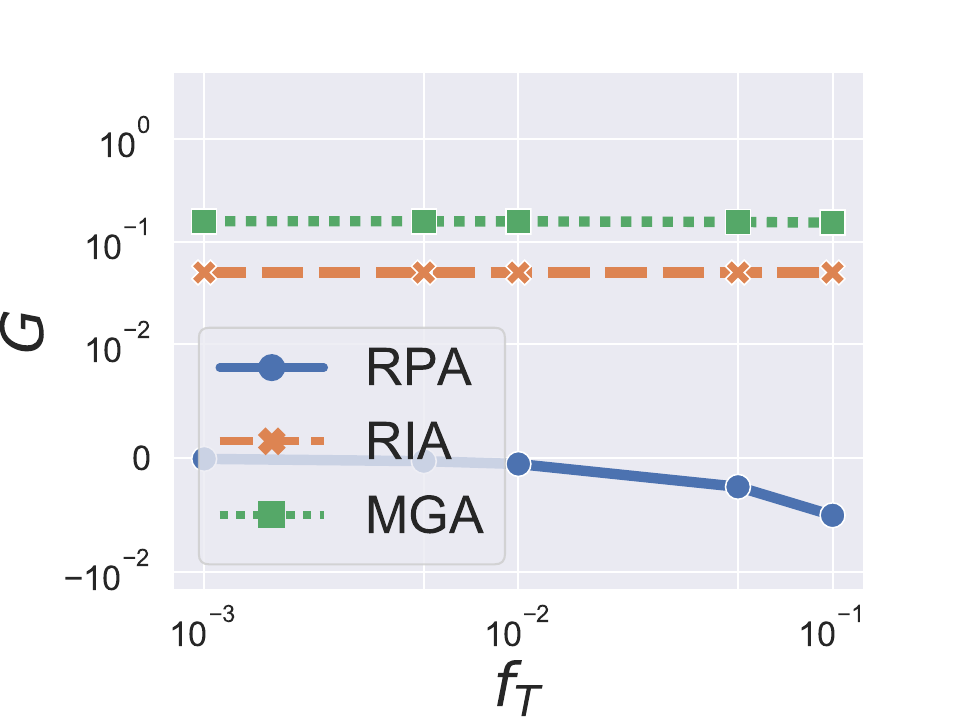}}
\subfloat{\includegraphics[width=0.2\textwidth]{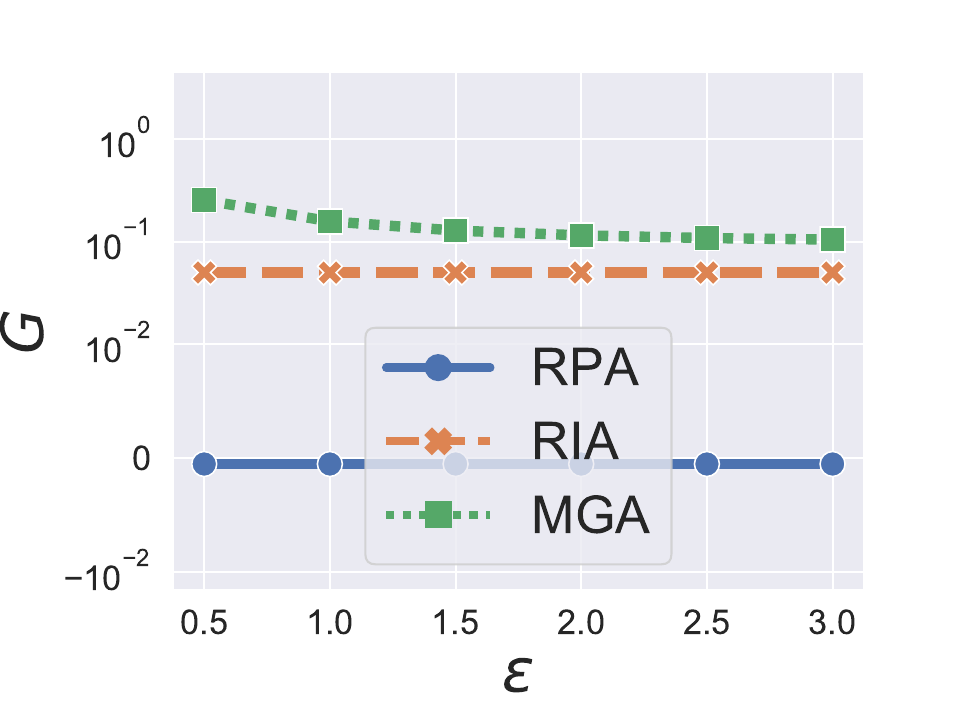}}
\subfloat{\includegraphics[width=0.2\textwidth]{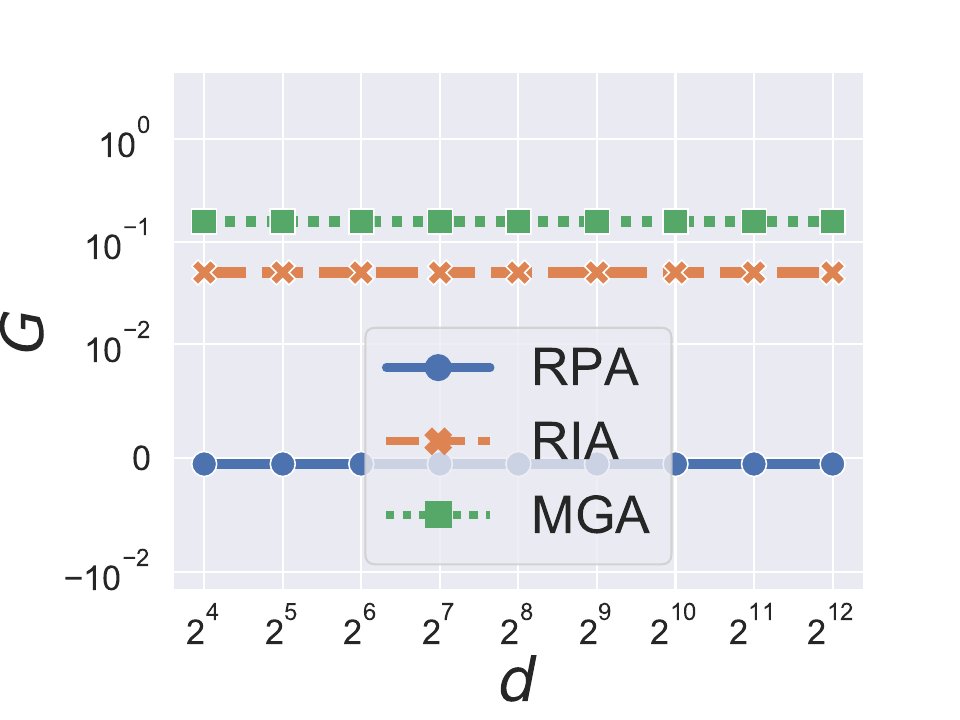}}
\vspace{-4mm}
\setcounter{subfigure}{0}
\subfloat{\includegraphics[width=0.2\textwidth]{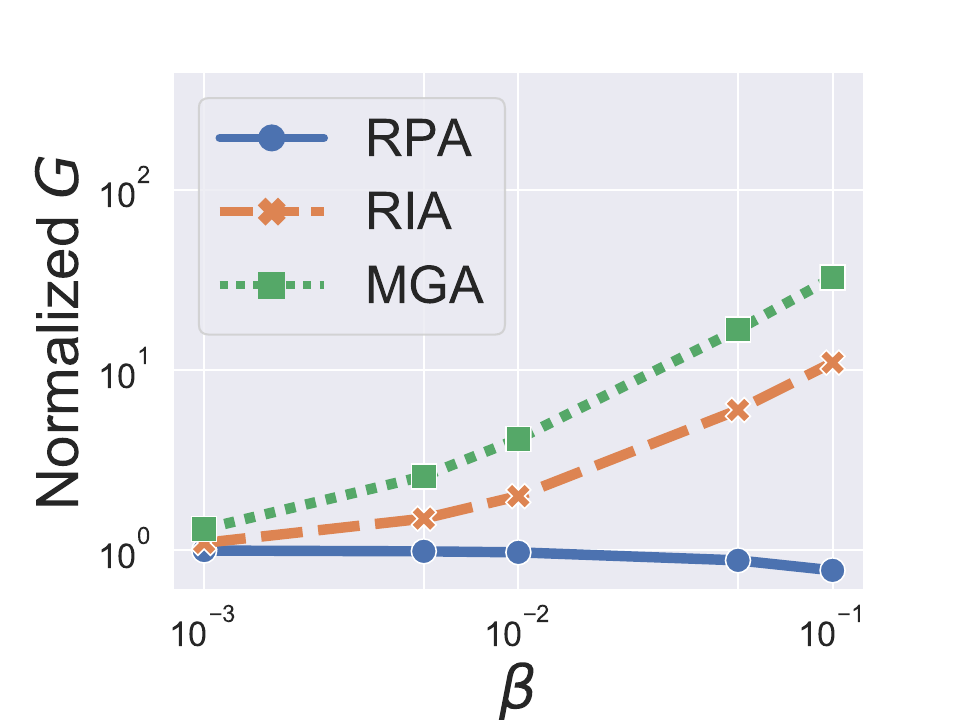}}
\subfloat{\includegraphics[width=0.2\textwidth]{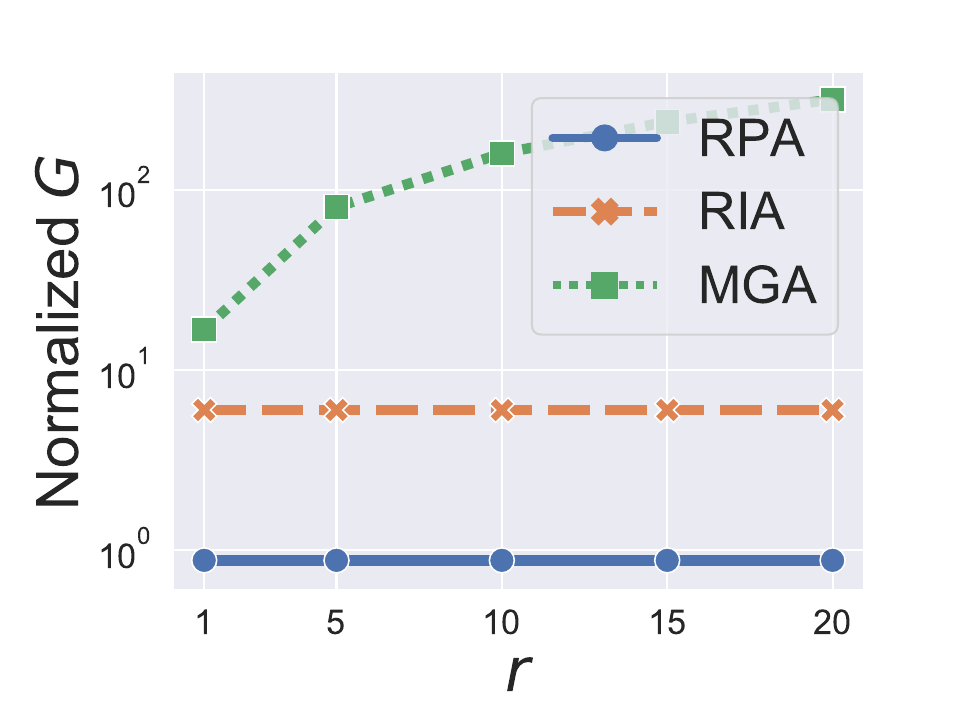}}
\subfloat{\includegraphics[width=0.2\textwidth]{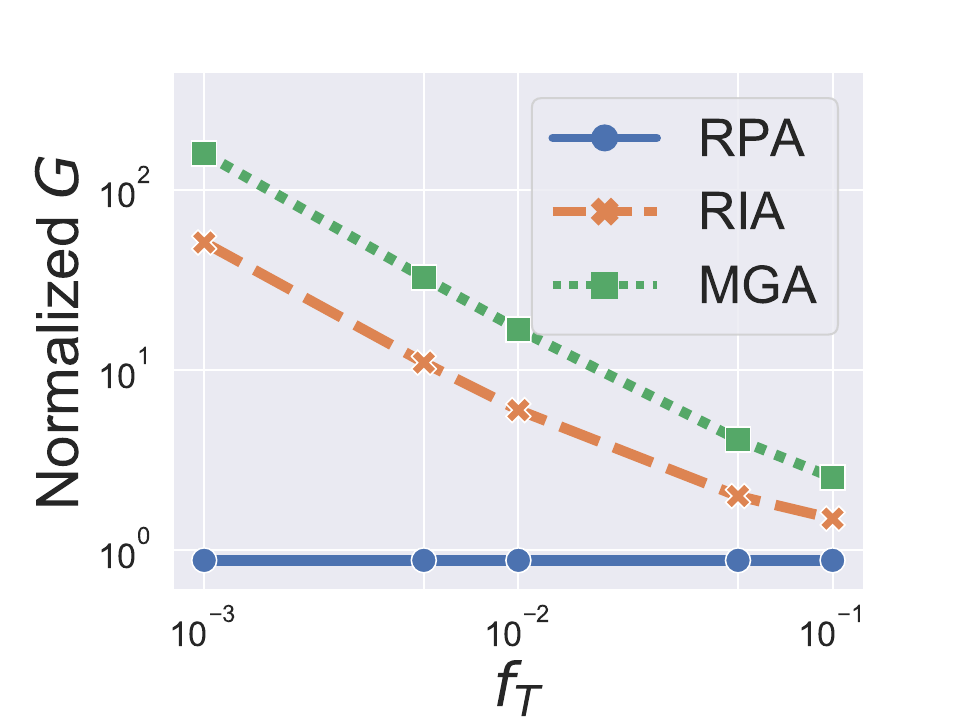}}
\subfloat{\includegraphics[width=0.2\textwidth]{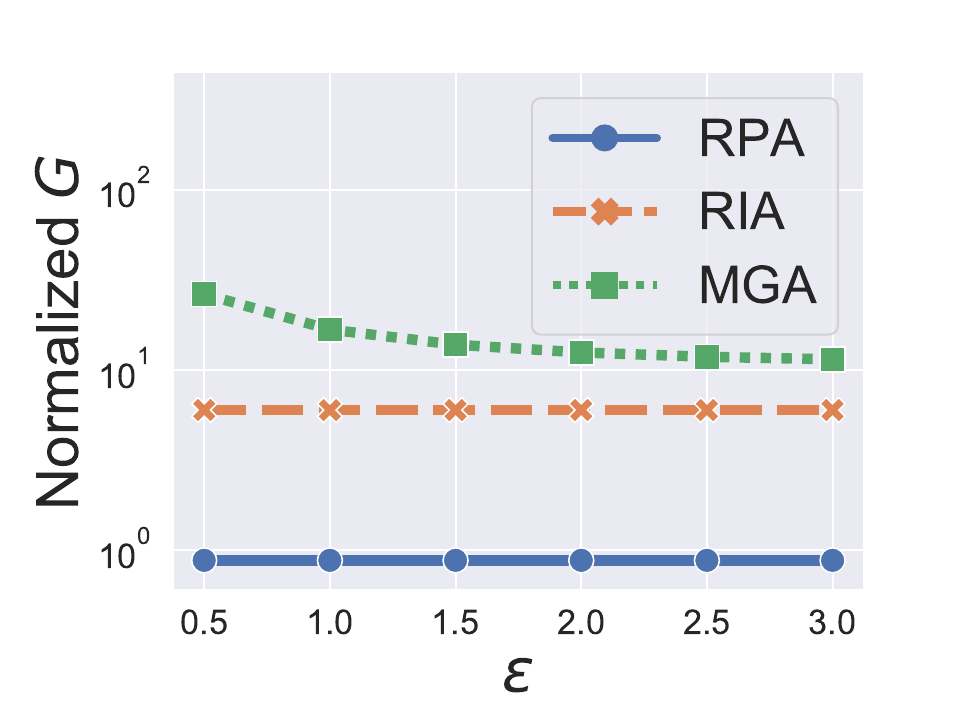}}
\subfloat{\includegraphics[width=0.2\textwidth]{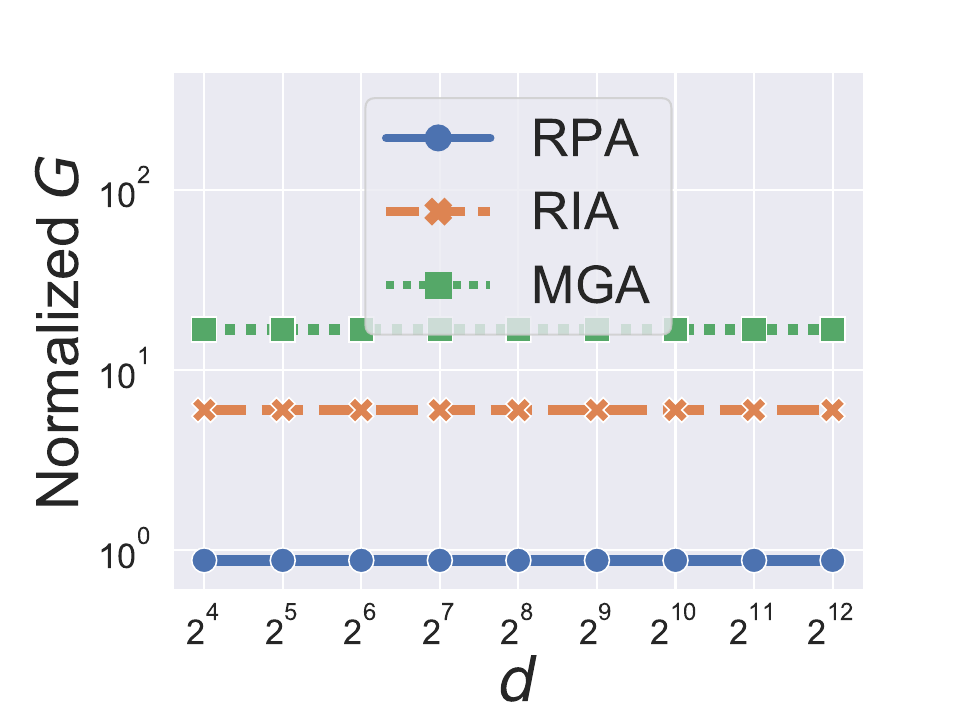}\label{fig:olh_d}}
\vspace{-3mm}
	 \caption{{Impact of different parameters on the overall gains (first row) and normalized overall gains (second row) of the three attacks for OLH.}}
	\label{prameterimpact_olh}
\vspace{-3mm}
\end{figure*}

\myparatight{Datasets}
We evaluate our attacks on three datasets, including a synthetic dataset and two real-world datasets, i.e., Fire \cite{fire}  and IPUMS \cite{ipums}.
\begin{packeditemize}
    \item \myparatight{Zipf} Following previous work on LDP protocols, we generate random data following the Zipf's distribution. In particular, we use the same parameter in the Zipf's distribution  as in \cite{wang2017locally}. By default, we synthesize a dataset with 1,024 items and 1,000,000 users. 
    \item \myparatight{Fire \cite{fire}} The Fire dataset was collected by the San Francisco Fire Department, recording information about calls for service. We filter the records by call type and use the data of type ``Alarms''. We treat the unit ID as the item that each user holds, which results in a total of 244 items and 548,868 users. 
    \item \myparatight{IPUMS \cite{ipums}} The IPUMS dataset contains the US census data over the years. We select the latest data  of 2017 and treat the city attribute as the item each user holds, which  results in a total of  102 items and 389,894 users.
\end{packeditemize}

\begin{table}[!t]
    \centering
    \begin{tabular}{|c|c|} \hline
         Parameter& Default setting  \\ \hline
         $\beta$ & 0.05 \\ \hline
          $r$ & 1  \\ \hline
          $f_T$ & 0.01  \\ \hline
	$\epsilon$ & 1\\ \hline
	$k$ & 20 \\ \hline
	$g$ & 10 \\ \hline
	    \end{tabular}
    \caption{Default parameter settings.}
    \label{tab:dft_st}
 \vspace{-4mm}
\end{table}

\myparatight{Parameter setting} For frequency estimation, the overall gains of our attacks may depend on $\beta$ (the fraction of fake users), $r$ and $f_T$ (the number of target items and their true frequencies),  $\epsilon$ (privacy budget), and $d$ (number of items in the domain). For heavy hitter identification, the success rates of our attacks  further depend on $k$ (the number of items identified as heavy hitters) and $g$ (the group size used by the PEM protocol). Table~\ref{tab:dft_st} shows the default settings for these parameters, which we will use in our experiments unless otherwise mentioned. We will study the impact of each parameter, while fixing the remaining parameters to their default settings. Moreover, we use $d'=\ceil{e^{\epsilon}+1}$ in OLH as $d'$ is an integer. 

\begin{figure*}[!t]
	 \centering
\subfloat{\includegraphics[width=0.2\textwidth]{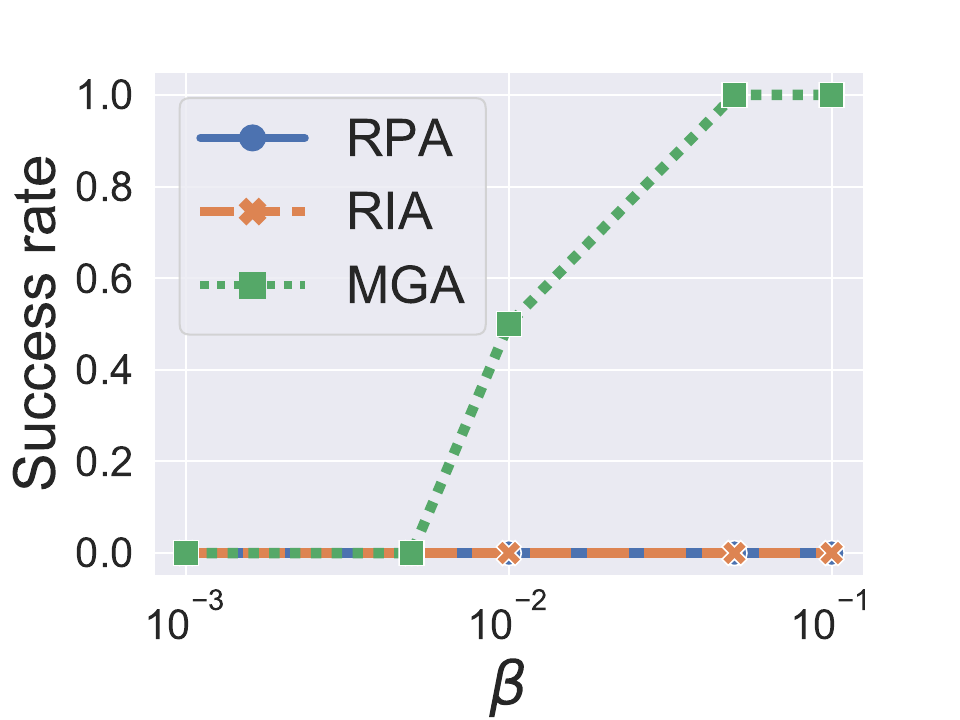}}
\subfloat{\includegraphics[width=0.2\textwidth]{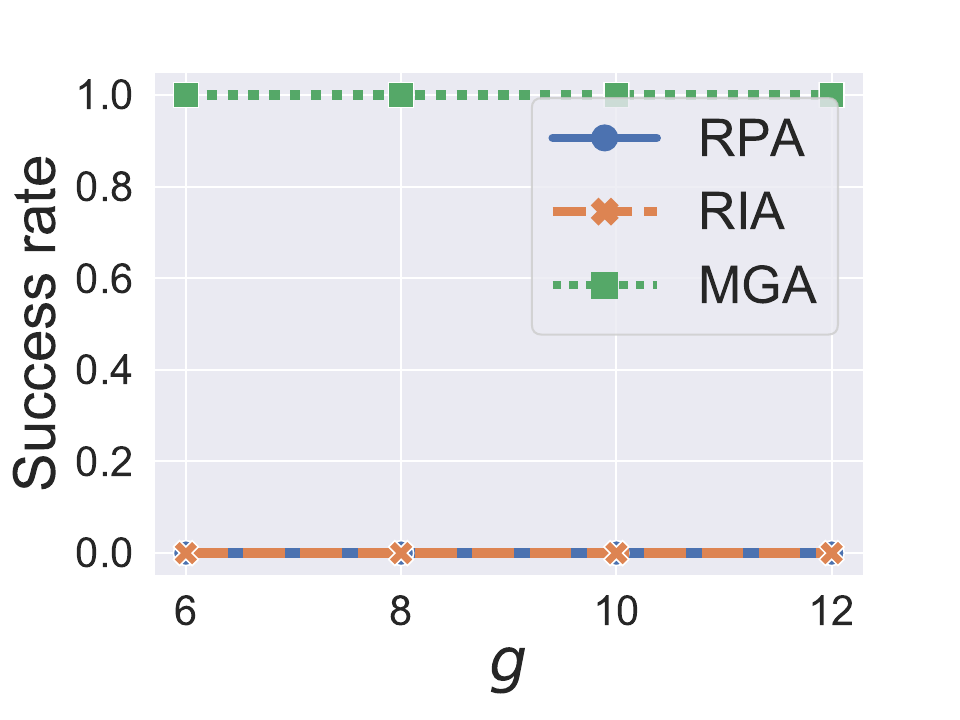}}
\subfloat{\includegraphics[width=0.2\textwidth]{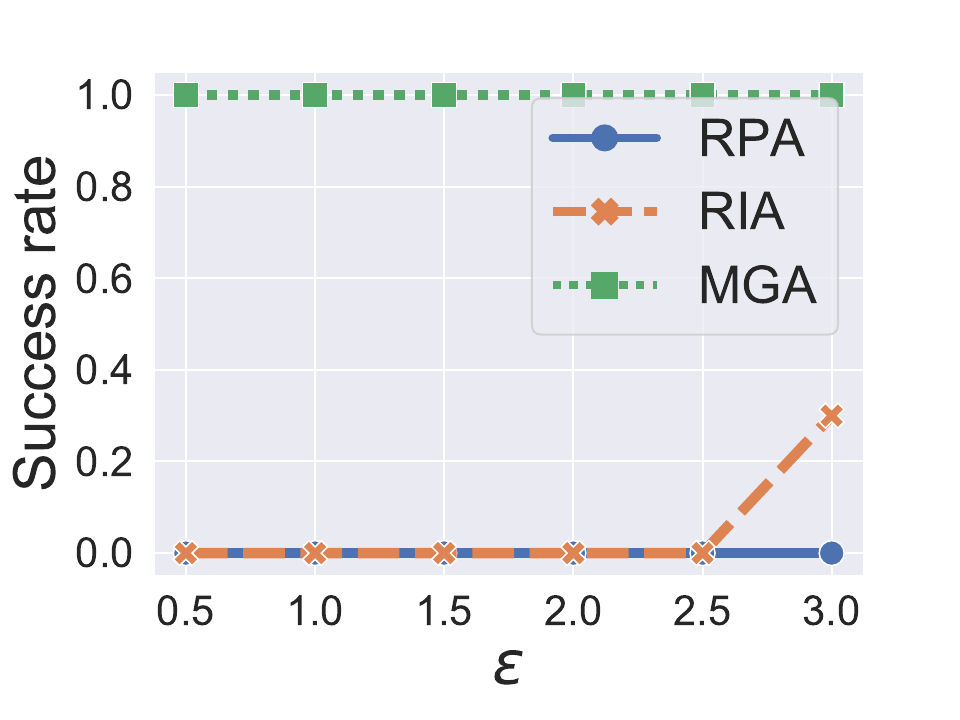}}
\subfloat{\includegraphics[width=0.2\textwidth]{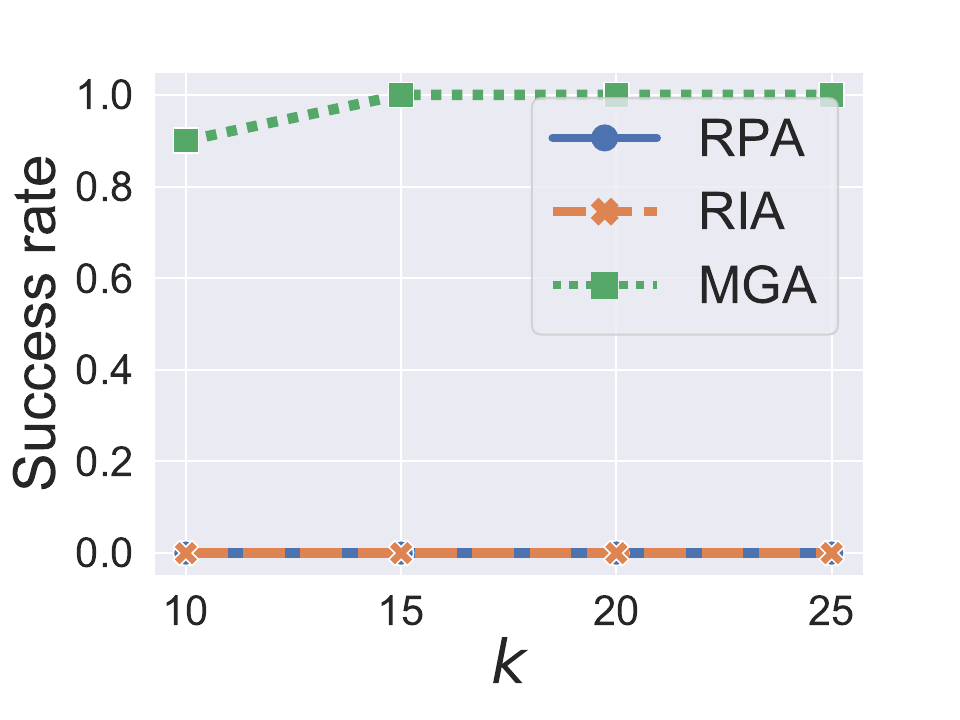}}
\subfloat{\includegraphics[width=0.2\textwidth]{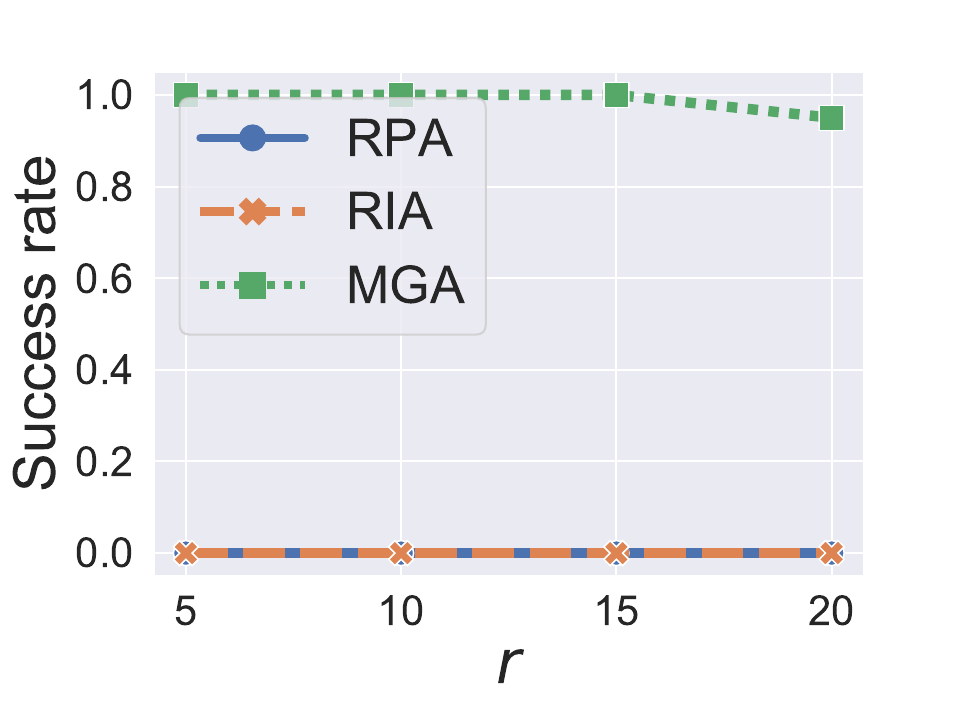}}
\vspace{-4mm}

\subfloat{\includegraphics[width=0.2\textwidth]{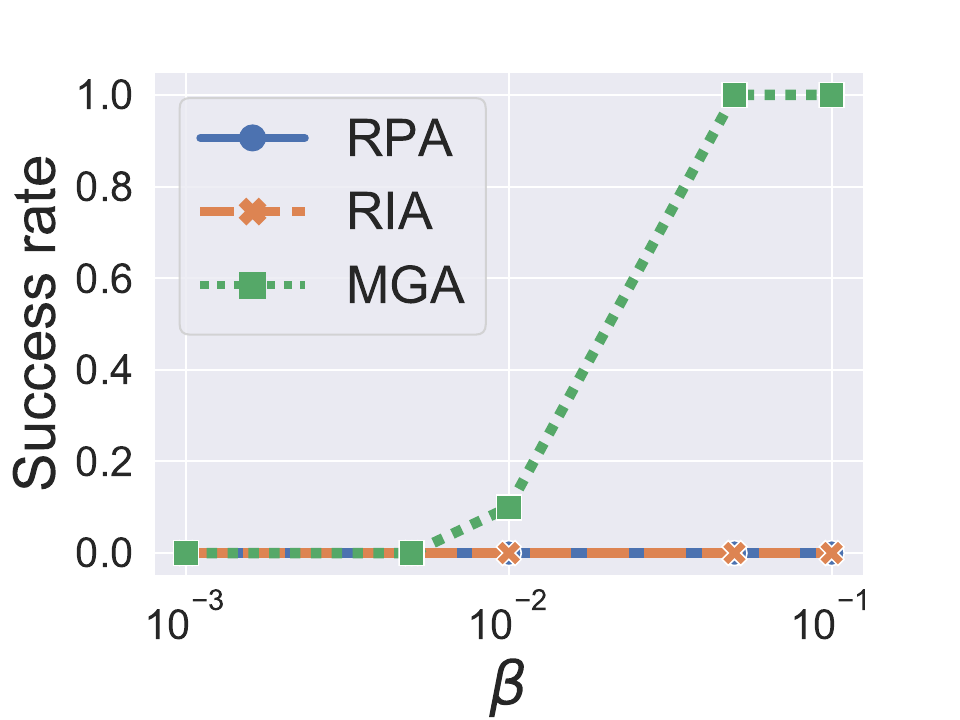}}
\subfloat{\includegraphics[width=0.2\textwidth]{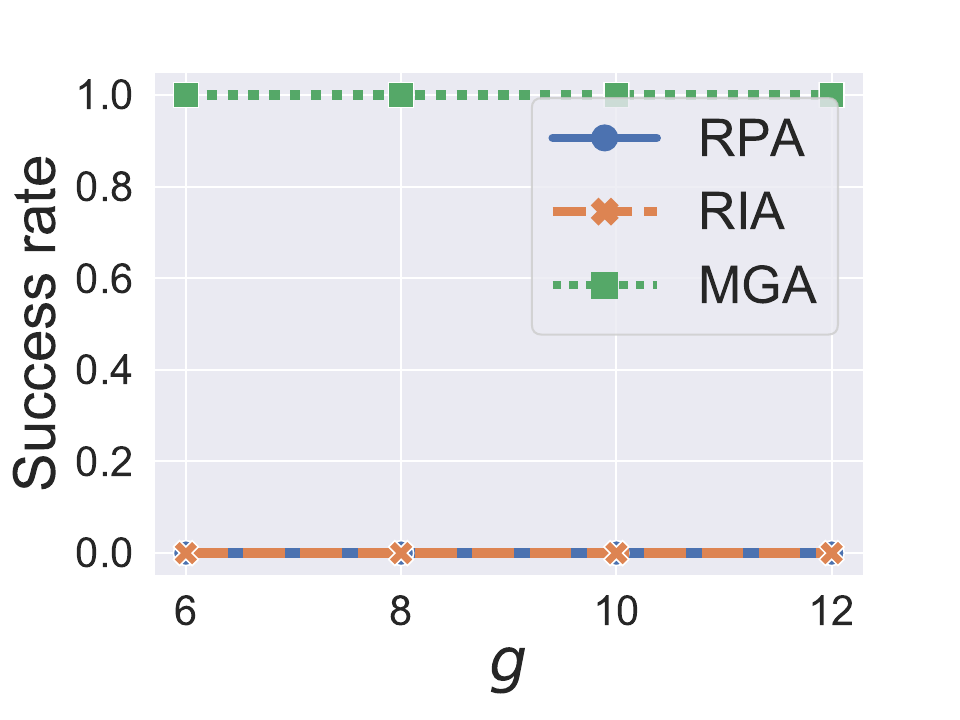}}
\subfloat{\includegraphics[width=0.2\textwidth]{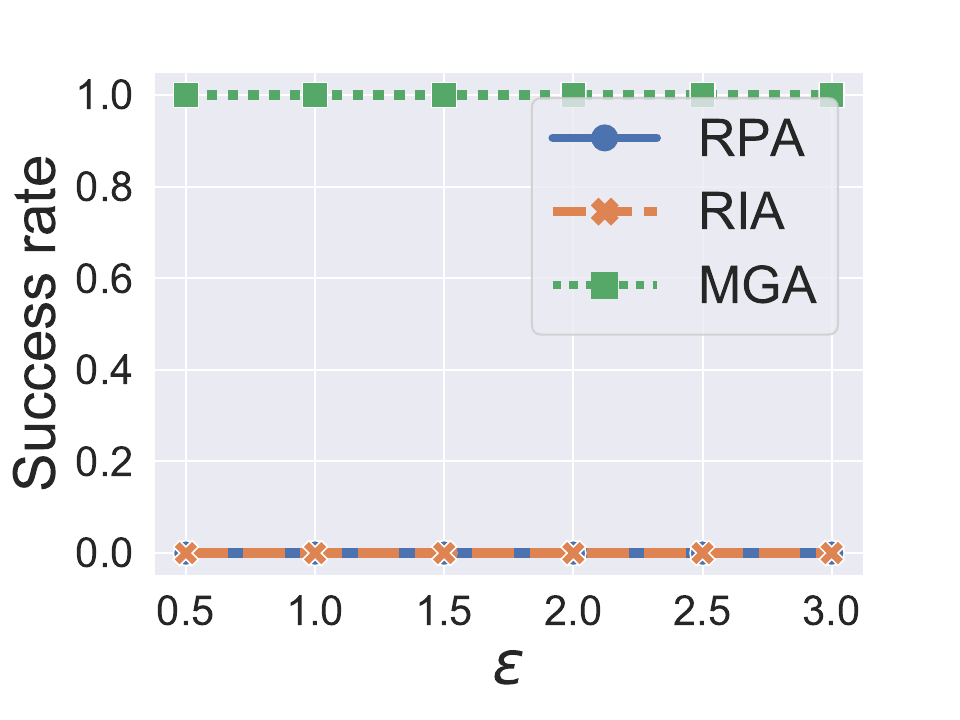}}
\subfloat{\includegraphics[width=0.2\textwidth]{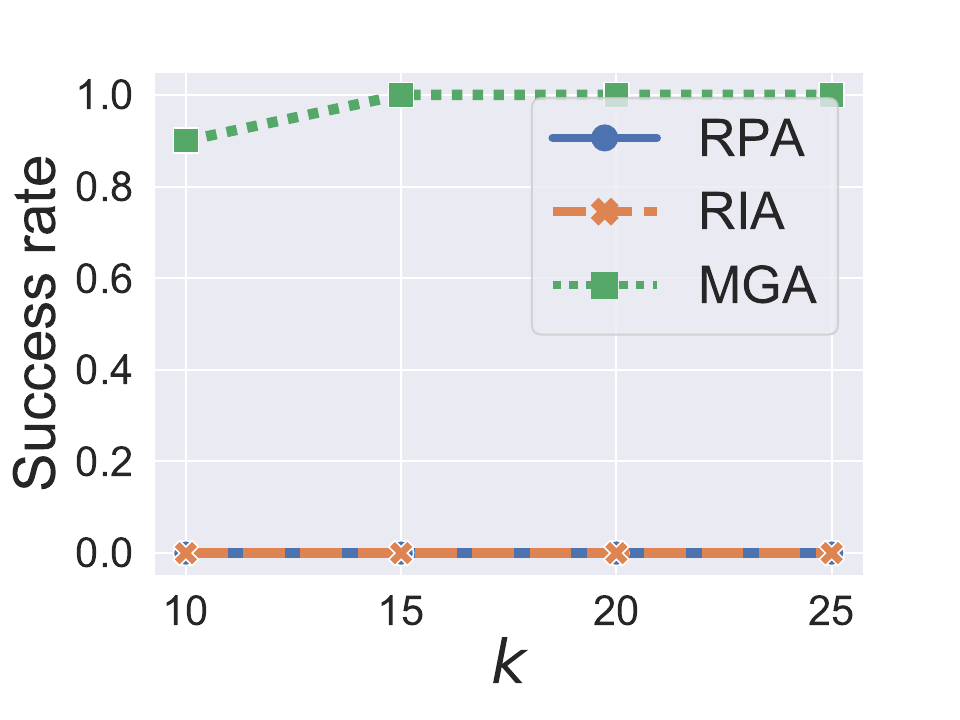}}
\subfloat{\includegraphics[width=0.2\textwidth]{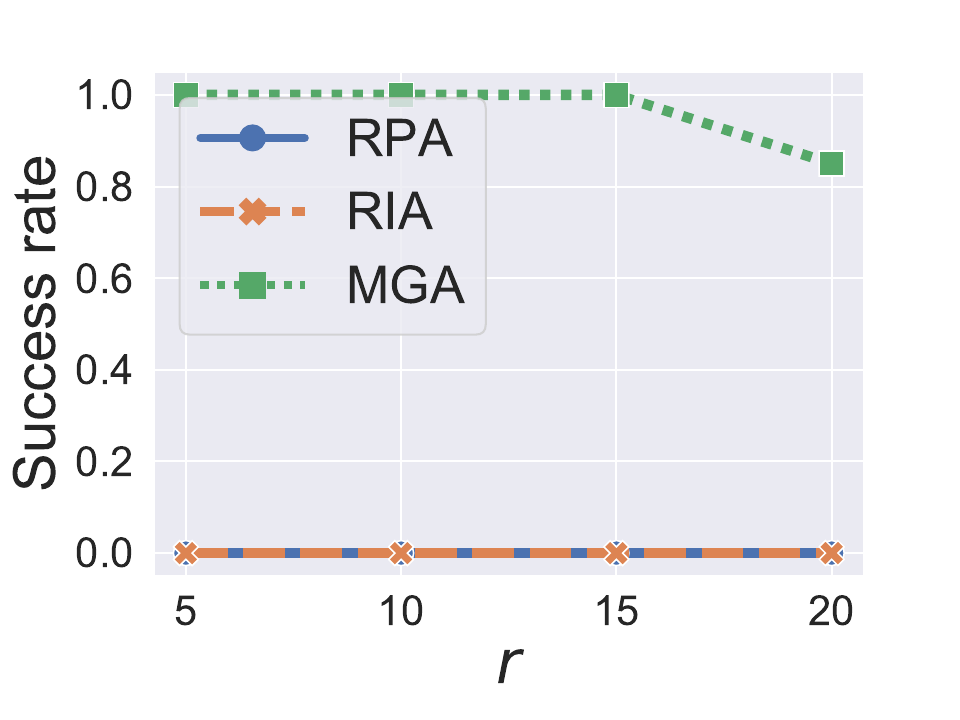}}
\vspace{-4mm}

\subfloat{\includegraphics[width=0.2\textwidth]{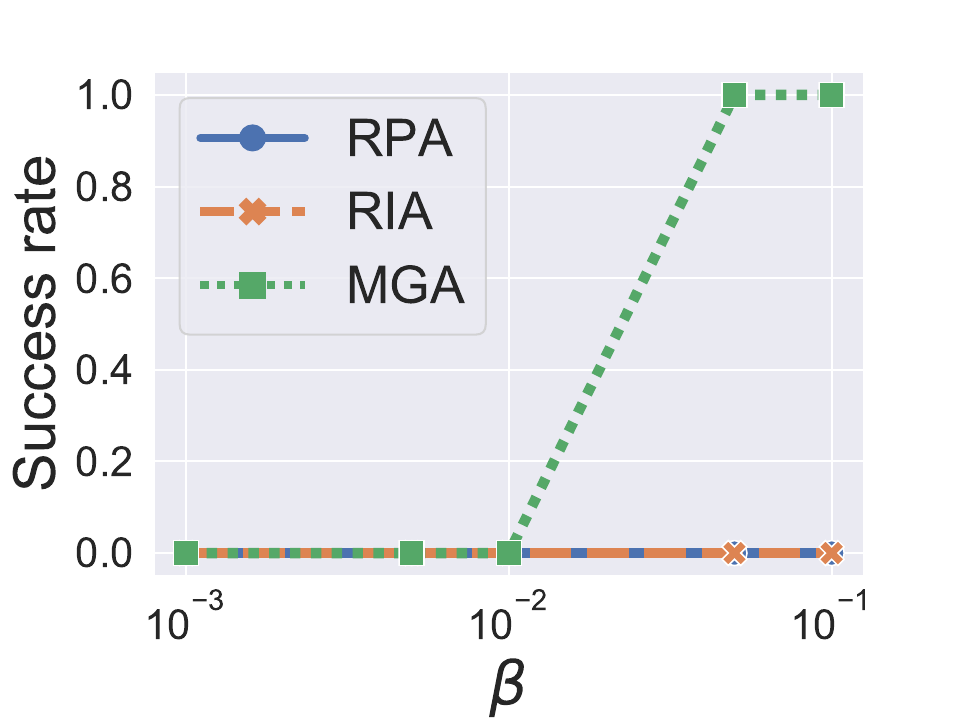}}
\subfloat{\includegraphics[width=0.2\textwidth]{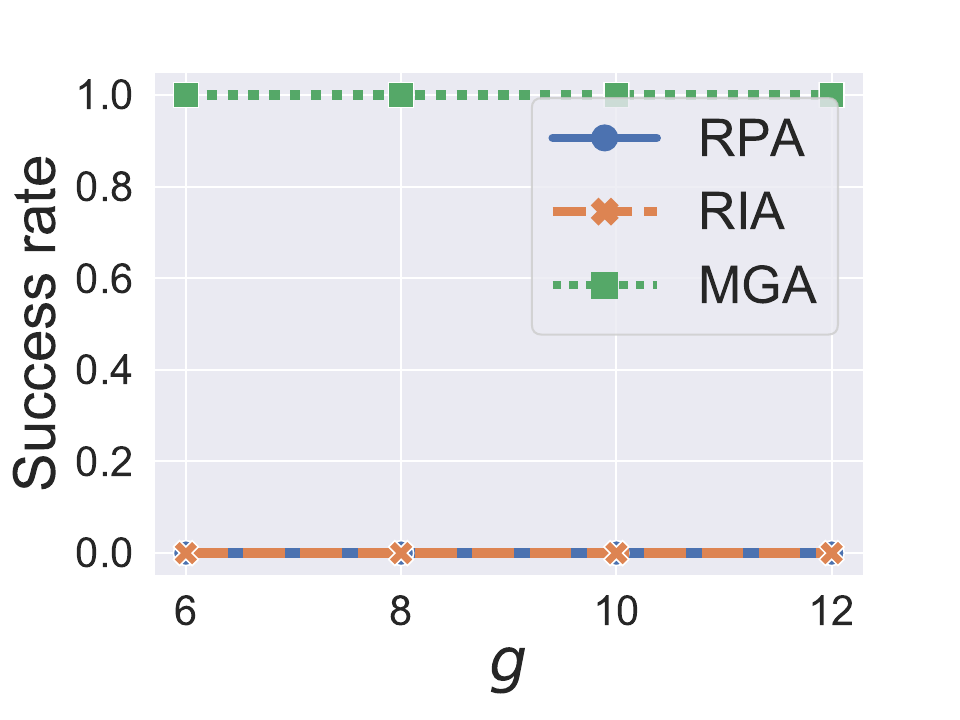}}
\subfloat{\includegraphics[width=0.2\textwidth]{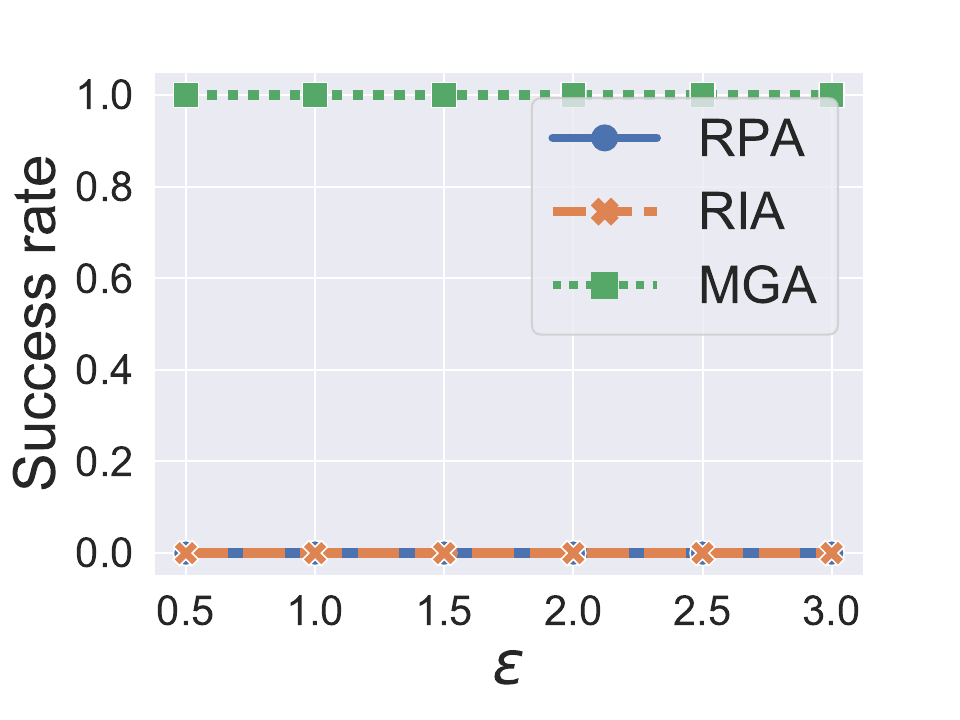}}
\subfloat{\includegraphics[width=0.2\textwidth]{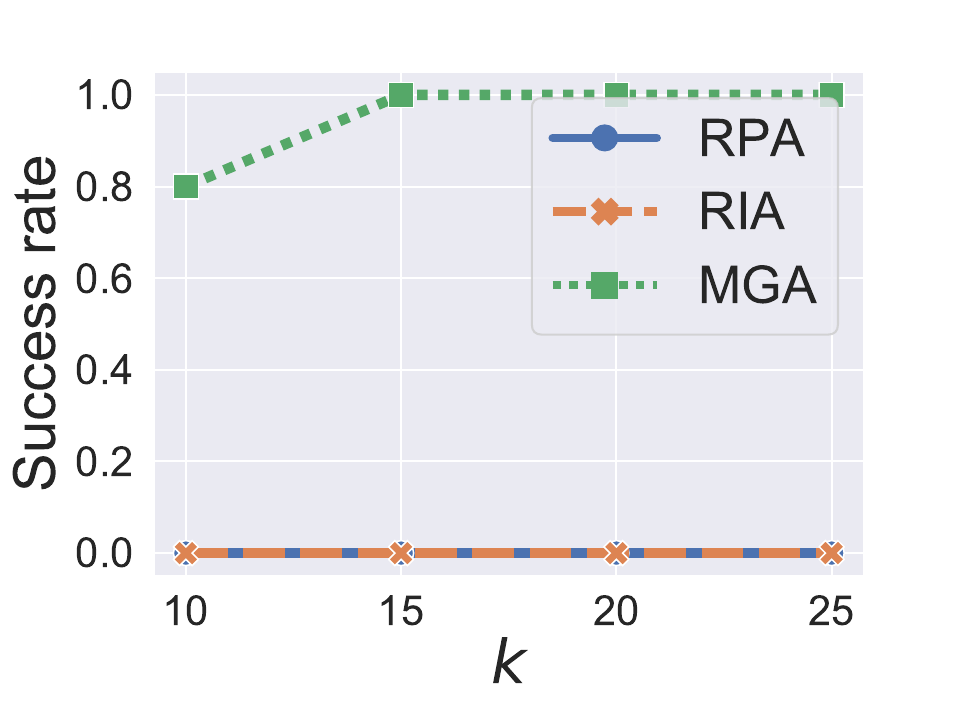}}
\subfloat{\includegraphics[width=0.2\textwidth]{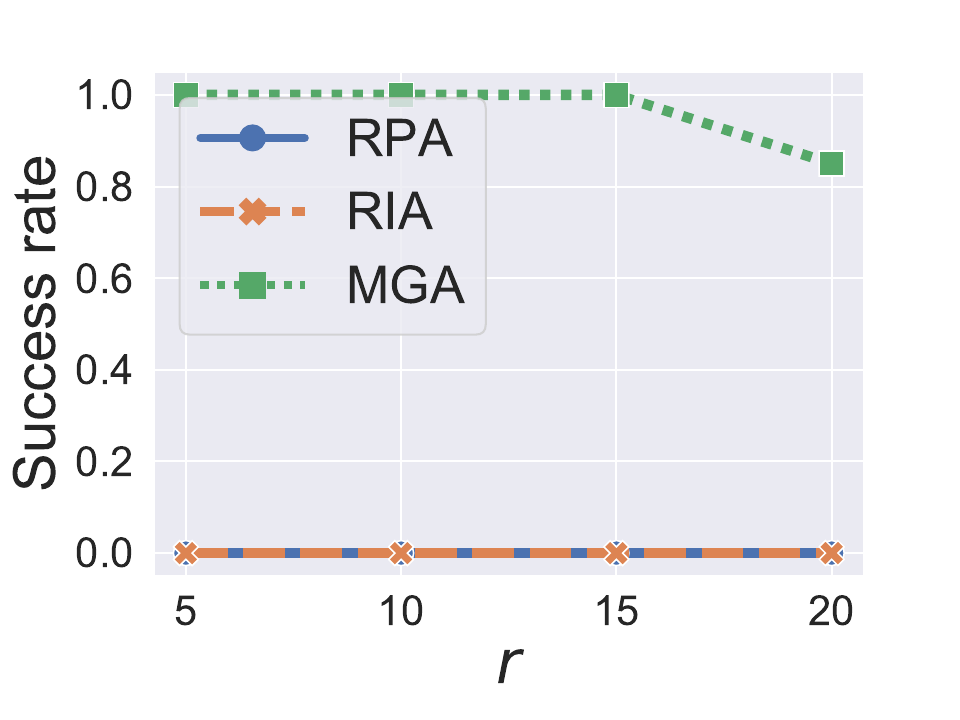}}

\vspace{-2mm}
	 \caption{Impact of different parameters on the success rates of the three attacks for PEM (heavy hitter identification protocol). The first row is on Zipf, the second row is on Fire, and the third row is on IPUMS.}\label{fig:hh_synth}
	 \vspace{-4mm}
\end{figure*}

\subsection{Results for Frequency Estimation}\label{sec:fe_result}
\vspace{-2mm}

\myparatight{Impact of different parameters} Table~\ref{tab:exp_gain} shows the theoretical overall gains of the three attacks for the kRR, OUE, and OLH protocols. We use these theoretical results to study the impact of each parameter. {\Cref{prameterimpact_krr,prameterimpact_oue,prameterimpact_olh} show the impact of different parameters on the overall gains and \emph{normalized overall gains}. A normalized overall gain is the ratio between the total frequencies of the target items after and before an attack, i.e., $(G+f_T)/f_T$, where $f_T$ is the total true frequencies of the target items.} We observe that MGA outperforms RIA, which outperforms RPA or achieves similar {(normalized)} overall gains with RPA. The reason is that MGA is an optimization-based attack, RIA considers information of the target items, and RPA does not consider  information about the target items. Next, we focus our analysis on MGA since it is the strongest attack.  The {(normalized)} overall gains of MGA increase as the attacker injects more fake users, the attacker promotes more target items (except the kRR protocol), or the privacy budget $\epsilon$ becomes smaller (i.e., security-privacy tradeoffs). {The (normalized) overall gain of MGA decreases as the total true frequency of the target items (i.e., $f_T$) increases, though the decrease of the overall gain is marginal.} The {(normalized)} overall gain of MGA increases for kRR but keeps unchanged for OUE and OLH as $d$ increases. {We note that, for a given set of target items (i.e., $f_T$ is given), the trend of normalized overall gain is the same as that of the overall gain with respect to parameters $\beta$, $r$, $\epsilon$, and $d$. Therefore, in the rest of the paper, we focus on overall gain for simplicity.} 

\myparatight{Measuring RIA and MGA for OLH}  The theoretical overall gain of RIA for OLH is derived based on the ``perfect'' hashing assumption, i.e., an item is hashed to a value in the hash domain $[d']$ uniformly at random. Practical hash functions may not satisfy this assumption. Therefore, the theoretical overall gain of RIA for OLH may be inaccurate in practice. We use xxhash~\cite{collet2016xxhash} as hash functions to evaluate the gaps between the theoretical and practical overall gains. In particular, Figure~\ref{verifyRIA} compares the theoretical and practical overall gains of RIA for OLH, where 1 item is randomly selected as target item, $\beta=0.05$, and $\epsilon=1$.   We observe that the theoretical and practical overall gains of RIA for OLH are similar. 

Our theoretical overall gain of MGA for OLH is derived based on the assumption that the attacker can find a hash function that hashes all target items to the same value. In practice, we may not be able to find such hash functions within a given amount of time. Therefore, for each fake user, we randomly sample some xxhash hash functions and use the one that hashes the most target items to the same value. Figure~\ref{verifyMGA_app} compares the theoretical and practical overall gains of MGA for OLH on the IPUMS dataset as we sample more hash functions for each fake user, where we randomly select 5 items as target items, i.e., $r=5$. Our results show that the practical overall gains approach the theoretical ones with several hundreds of randomly sampled hash functions when $r=5$. We have similar observations for the other two datasets and thus we omit their results due to the limited space.   

\begin{figure}[!t]
 \vspace{-4mm}
	 \centering
\subfloat[]{\includegraphics[width=0.25\textwidth]{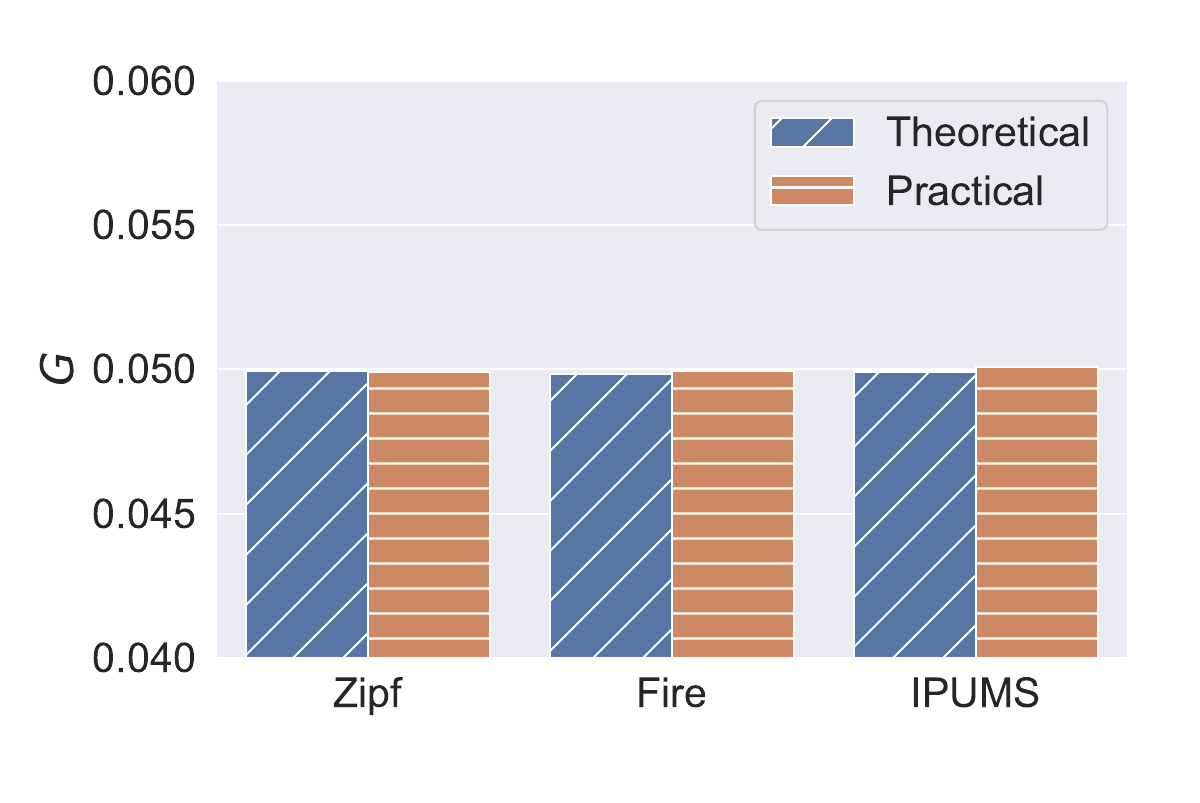}\label{verifyRIA}}
\subfloat[]{\includegraphics[width=0.25\textwidth]{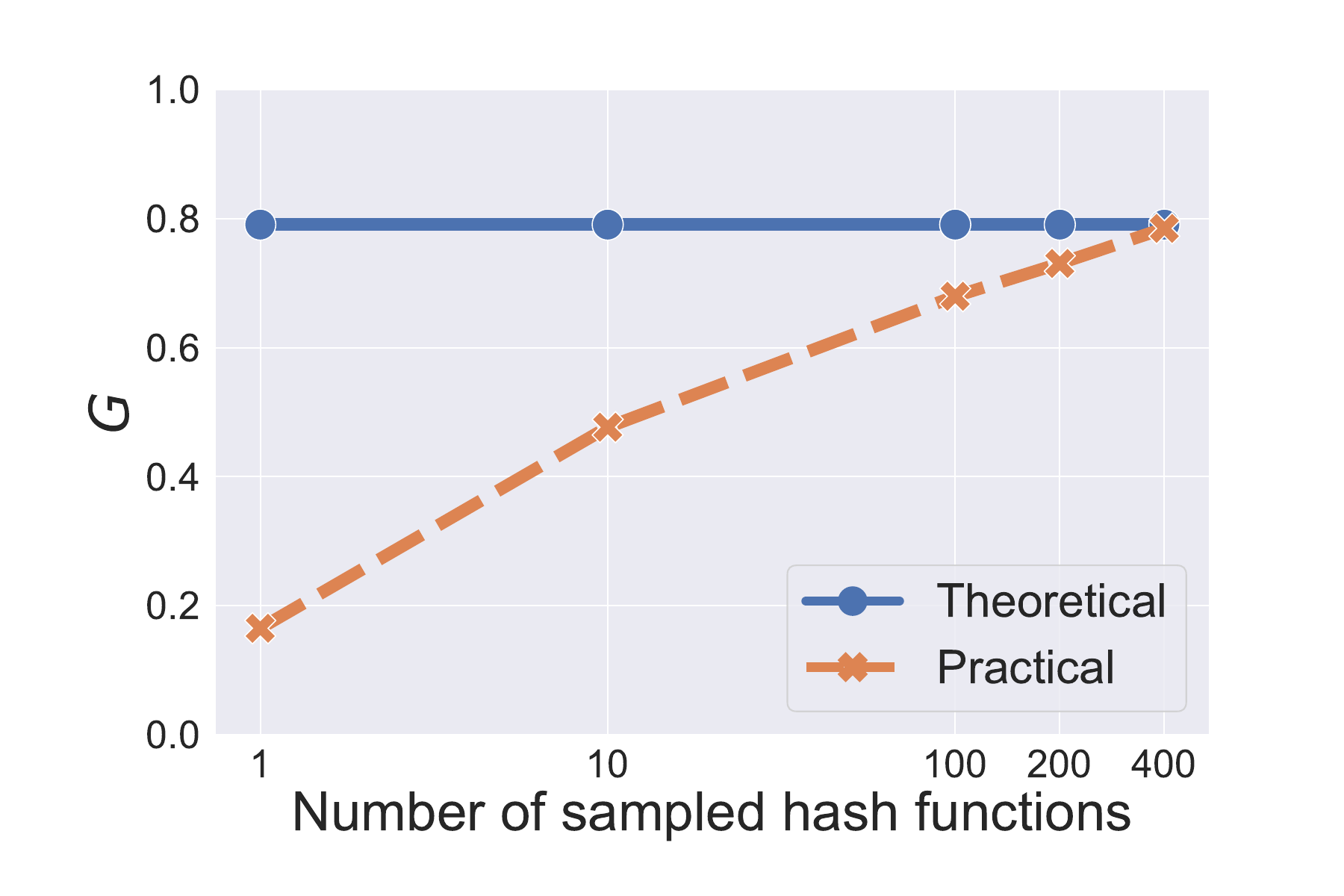}\label{verifyMGA_app}}
 \vspace{-2mm}
	 \caption{(a) Theoretical and practical overall gains of RIA for OLH. (b) Theoretical and practical overall gains of MGA for OLH on the IPUMS dataset as we sample more hash functions for each fake user, where $r=5$. }
 \vspace{-4mm}
\end{figure}

\subsection{Results for Heavy Hitter Identification}

Figure~\ref{fig:hh_synth} shows the empirical results of applying our three attacks, i.e., RPA, RIA and MGA, to PEM  on the Zipf, Fire, and IPUMS datasets, respectively. By default, we randomly select $r=10$ target items that are not identified as top-$k$ heavy hitters by PEM before attack and use the three attacks to promote them. Default values for the other parameters are identical to those in Table \ref{tab:dft_st}. The success rate of an attack is calculated as the fraction of target items that  appear in the estimated top-$k$ heavy hitters. The results show that our MGA attacks can effectively  compromise the PEM protocol. In particular, we observe that MGA only needs about 5\% of fake users to achieve a 100\% success rate when $r=10$ and $k=20$. In fact, with only 5\% of fake users, we can promote 10 target items to be in the top-15 heavy hitters, or promote  15 target items to be in the top-20 heavy hitters. However, RPA and RIA are ineffective. Specifically, even if we inject 10\% of fake users, neither RPA nor RIA can successfully promote even one of the target items to be in the top-$k$ heavy hitters.  Moreover, the number of groups $g$ and the privacy budget $\epsilon$  have negligible impact on the effectiveness of our attacks. 

%%%%%%%%%%%%%%%%%%%%%%%%%%%%%%%%%%%%%%%%%%%%%%%%%%%%%%%

\vspace{-1mm}
\section{Countermeasures}
{We explore three countermeasures. The first countermeasure is to normalize the estimated item frequencies to be a probability distribution,  the second countermeasure is to detect fake users via \emph{frequent itemset mining} of the users' perturbed values and remove the detected fake users before estimating item frequencies, and the third countermeasure is to detect the target item without detecting the fake users when there is only one target item. The three countermeasures are effective in some scenarios. However, our MGA is still effective in other scenarios, highlighting the needs for new defenses against our data poisoning attacks.} 

\vspace{-1mm}
\subsection{Normalization}  
\label{defense:normalization}
The LDP protocols estimate item frequencies using Equation (\ref{aggregate}). Therefore, the estimated item frequencies may not form a probability distribution, i.e., some estimated item frequencies may be negative and they may not sum to 1. For instance, our experimental results in Section~\ref{sec:fe_result} show that the overall gains of MGA may be even larger than 1. Therefore, one natural countermeasure  is to normalize the estimated item frequencies such that each estimated item frequency is non-negative and the estimated item frequencies sum to 1. For instance, one normalization we consider is as follows: the central server first estimates the  frequency $\tilde{f}_v$ for each item $v$ following a LDP protocol (kRR, OUE, or OLH); then the server finds the minimal estimated item frequency $\tilde{f}_{min}$; finally, the server calibrates the estimated frequency for each item $v$ as $\bar{f}_v=\frac{\tilde{f}_v - \tilde{f}_{min}}{\sum_v (\tilde{f}_v - \tilde{f}_{min})}$, where $\bar{f}_v$ is the calibrated frequency. Our overall gain is calculated by the difference between the calibrated frequencies of the target items after and before attack. {We note that there are also other methods to normalize the estimated item frequencies~\cite{jia2019calibrate,wang2019consistent}, which we leave as future work.}   
 Note that the normalization countermeasure is not applicable to heavy hitter identification because normalization does not impact the ranking of items' frequencies. 

\vspace{-1mm}
\subsection{Detecting Fake Users}\label{sec:detection}
RPA and MGA directly craft the perturbed values for fake users, instead of using the LDP protocol to generate the perturbed values from certain items. Therefore, the perturbed values for the fake users may be statistically abnormal. 
We note that it is challenging to detect fake users via statistical analysis of the perturbed values  for the kRR protocol, because the perturbed value of a user is just an item, no matter whether or not the attacker follows the protocol to generate the perturbed value. Therefore, we study detecting fake users in the RPA and MGA attacks for the OUE and OLH protocols. Since PEM iteratively applies OLH, we can also apply  detecting fake users to PEM. 

\myparatight{OUE}
Recall that MGA assigns 1 to all target items and $l$ randomly selected items in the perturbed binary vector for each fake user. Therefore, among the perturbed binary vectors from the fake users, a set of items will always be 1. However, if the perturbed binary vectors follow the OUE protocol, it is unlikely to observe that this set of items are all 1's for a large number of users. Therefore, our idea to detect fake users consists of two steps. In the first step, the server identifies itemsets  that are all 1's in the perturbed binary vectors of a large number of users. In the second step, the server detects fake users if the probability that such large number of users have these itemsets of all 1's is small, when following  OUE.      

{\bf Step I.} In this step, the server identifies {itemsets} that are frequently all 1's among the perturbed binary vectors. Figure~\ref{example_itemset} shows an example  itemset that are all 1's in 3 of the 4 binary vectors. Identifying such  itemsets is also known as \emph{frequent itemset mining}~\cite{agrawal1993mining}.  In our problem, given the perturbed binary vectors from all users,  frequent itemset mining can find the itemsets that are all 1's in at least a certain number of users. Specifically, a frequent itemset mining method produces some tuples $\bm{B}=\{(B, s)|s\geq \tau\}$, where $B$ is an itemset and $s$ is the number of users whose perturbed binary vectors are 1's for all items in $B$. 

\begin{figure}[!t]
	 \centering
\includegraphics[width=0.45\textwidth]{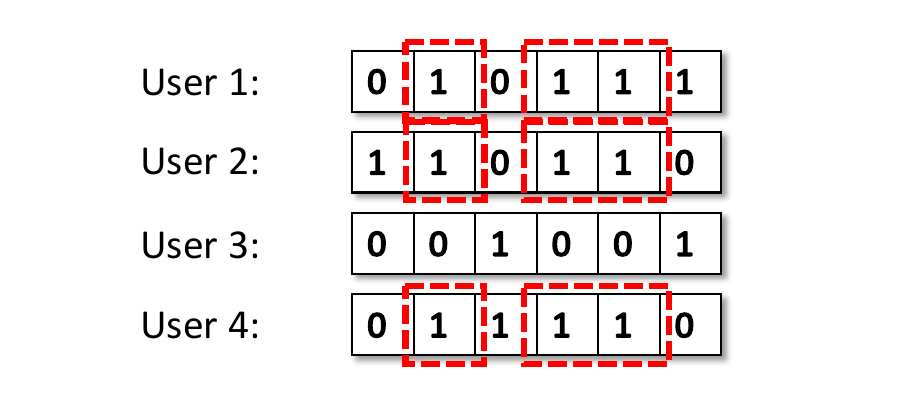}
 \vspace{-6mm}
	 \caption{An example itemset that are all 1's in 3 of the 4 binary vectors. Each column corresponds to an item.}
        \label{example_itemset}
 \vspace{-5mm}
\end{figure}

{\bf Step II.} In this step, we determine whether there are frequent itemsets that are statistically abnormal. Specifically, we predict a tuple $(B, s)\in \bm{B}$ to be abnormal if $s\geq \tau_z$, where $z=|B|$ is the size of the itemset $B$. When an itemset is predicted to be abnormal, we predict the items as the target items and the users whose perturbed binary vectors are 1's for all items in the itemset to be fake. The threshold $\tau_z$ achieves a tradeoff between \emph{false positive rate} and \emph{false negative rate} of detecting fake users. Specifically, when $\tau_z$ is larger, a smaller number of genuine users are predicted as fake (i.e., a smaller false positive rate), while a larger number of fake users are not  detected (i.e., a larger false negative rate). Therefore, a key challenge is how to select the threshold $\tau_z$. We propose to select the threshold such that the false positive rate is at most $\eta$. Specifically, given a threshold $\tau_z>(n+m)pq^{z-1}$, we can derive an upper bound of the false positive rate as $\frac{(n+m)pq^{z-1}(1-pq^{z-1})}{[\tau_z-(n+m)pq^{z-1}]^2}$ (see Appendix~\ref{proof:tau_oue} for details). Therefore, to guarantee that the false positive rate is at most $\eta$ and achieve a small false negative rate, we select the smallest $\tau_z$ that satisfies $\tau_z>(n+m)pq^{z-1}$ and  $\frac{(n+m)pq^{z-1}(1-pq^{z-1})}{[\tau_z-(n+m)pq^{z-1}]^2} \le\eta$. We set $\eta=0.01$ in our experiments. 

\myparatight{OLH}
To attack the OLH protocol, MGA searches a hash function for each fake user that hashes  as many target items to the same value as possible. Suppose we construct a $d$-bit binary vector $\bm{y}$ for each user with a tuple $(H, a)$ such that $y_v=1$ if and only if $H(v)=a$. Then, the target items will be 1's in the binary vectors for a large number of users. Therefore, we can also leverage the method to detect fake users in OLH. Specifically, in Step I, we find frequent itemsets in the constructed binary vectors. In Step II, we predict an itemset $B$ to be abnormal if its number of occurrences $s$ among the $n+m$ binary vectors is larger than a threshold $\tau_z$, where $z=|B|$ is the size of the itemset. Like OUE, we select the threshold $\tau_z$ such that the false positive rate is at most $\eta$. Specifically, we select the smallest $\tau_z$ that satisfies $I(q^{z-1};\tau_z, n+m-\tau_z+1)\le \eta$, where $I$ is the \emph{regularized incomplete beta function}~\cite{abramowitz1965handbook}. $I(q^{z-1};\tau_z, n+m-\tau_z+1)$ is the false positive rate for a given $\tau_z$  (see Appendix \ref{proof:tau_oue} for details).

\myparatight{PEM} The heavy hitter identification protocol PEM iteratively applies OLH to identify heavy hitters. Therefore, we can apply the frequent itemset mining based detection method to detect fake users in PEM. Specifically, in each iteration of PEM, the central server applies the detection method  in OLH to detect fake users in PEM; and the central server removes the predicted fake users before computing the top-$k$ prefixes. 

{
\subsection{Conditional Probability based Detection} 
\label{sec:onetarget}
The frequent itemset mining based detection method above requires at least two target items as it identifies the abnormal frequent itemset as the target items. When there is only one target item, i.e., $r=1$, it fails to detect the target item. Therefore, we discuss another method to detect the target item when $r=1$, which leverages conditional probabilities. Note that this method does not detect fake users. 

\myparatight{OUE} Suppose $\bm{y}$ is a user's perturbed binary vector. With a little abuse of notation, we denote the $j$-th bit of $\bm{y}$ as $y_j$. Given the target item $t$ and a random item $j$, we have the following equations under our MGA attacks to OUE:
{\small
\begin{align}
	\text{Pr}(y_j=y_t=1) &= \text{Pr}(v=t)\cdot\text{Pr}(y_j=y_t=1|v=t)\nonumber\\
	                    &+ \text{Pr}(v=j)\cdot\text{Pr}(y_j=y_t=1|v=j)\nonumber\\
						&+\text{Pr}(v\neq t, j)\cdot\text{Pr}(y_j=y_t=1|v\neq t,j)\nonumber\\
						&+\text{Pr}(\text{fake})\cdot\text{Pr}(y_j=y_t=1|\text{fake})\\
						&= \frac{nf_t}{n+m}\cdot  pq+ \frac{nf_j}{n+m}\cdot  pq\nonumber\\
						&+ \frac{n(1-f_t-f_j)}{n+m}\cdot q^2 + \frac{m}{n+m}\cdot \frac{l}{d-1},\\
	\text{Pr}(y_t=1) &= \text{Pr}(v=t)\cdot\text{Pr}(y_t=1|v=t) \nonumber\\
	               &+ \text{Pr}(v\neq t)\cdot\text{Pr}(y_t=1|v\neq t)\nonumber\\
	               & + \text{Pr}(\text{fake})\cdot\text{Pr}(y_t=1|\text{fake})\\
						&= \frac{nf_t}{n+m}\cdot p + \frac{n(1-f_t)}{n+m}\cdot  q + \frac{m}{n+m}, \\
\text{Pr}(y_j=1|y_t=1)=&\frac{\text{Pr}(y_j=y_t=1)}{\text{Pr}(y_t=1)}\\
	= &q + \frac{f_jq(p-q) + \frac{\beta}{1-\beta}\cdot(\frac{l}{d-1}-q)}{f_tp+(1-f_t)q+\frac{\beta}{1-\beta}}.
	\label{eq:target}
\end{align}}%
Given a non-target item $u\neq j$, we have the following:
{\small{\begin{align}
&\quad\text{Pr}(y_j=y_u=1) \nonumber\\
	 &= \text{Pr}(v=u)\cdot\text{Pr}(y_j=y_u=1|v=u)\nonumber\\
	                    &+ \text{Pr}(v=j)\cdot\text{Pr}(y_j=y_u=1|v=j)\nonumber\\
						&+\text{Pr}(v\neq j, u)\cdot\text{Pr}(y_j=y_u=1|v\neq j,u)\nonumber\\
						&+\text{Pr}(\text{fake})\cdot\text{Pr}(y_j=y_u=1|\text{fake})\\
						&= \frac{nf_u}{n+m}\cdot  pq+ \frac{nf_j}{n+m}\cdot  pq+ \frac{n(1-f_u-f_j)}{n+m}\cdot q^2 \nonumber\\
						&+ \frac{m}{n+m}\cdot \frac{l}{d-1}\cdot\frac{l-1}{d-2}, \\
	&\quad\text{Pr}(y_u=1) \nonumber\\
	&= \text{Pr}(v=u)\cdot\text{Pr}(y_u=1|v=u)\nonumber\\
						&+ \text{Pr}(v\neq u)\cdot\text{Pr}(y_u=1|v\neq u)\nonumber\\
						&+ \text{Pr}(\text{fake})\cdot\text{Pr}(y_u=1|\text{fake})\\
						&= \frac{nf_u}{n+m}\cdot  p + \frac{n(1-f_u)}{n+m}\cdot  q + \frac{m}{n+m}\cdot\frac{l}{d-1}, \\
	&\quad\text{Pr}(y_j=1|y_u=1) \nonumber\\
	&=\frac{\text{Pr}(y_j=y_u=1)}{\text{Pr}(y_u=1)}\\
	&= q + \frac{f_jq(p-q) +\frac{\beta}{1-\beta}\cdot\frac{l}{d-1}\cdot(\frac{l-1}{d-2}-q)}{f_up+(1-f_u)q+\frac{\beta}{1-\beta}\cdot\frac{l}{d-1}}.
	\label{eq:non-target}
\end{align}}}%

\begin{table}[!t]\renewcommand{\arraystretch}{1}
\setlength{\tabcolsep}{4pt}
\centering
\subfloat[$\beta=0.05$\label{tab:beta0.05}]{
    \begin{tabular}{|c|c|c|c|c|c|c|c|c|}	
	\hline
	\small{$f_j$} & \small{$0.01$} & \small{$0.01$} & \small{$0.1$}& \small{$0.1$}& \small{$0.5$} & \small{$0.5$} & \small{$0.9$} & \small{$0.9$}\\
	\hline
         	\small{$f_t$} & \small{$0$} & \small{$0.01$} & \small{$0$} & \small{$0.01$} & \small{$0$} & \small{$0.01$} & \small{$0$} & \small{$0.01$}\\
	\hline
	\small{$\hat{f}_u$} & \small{$0.25$} & \small{$0.26$} & \small{$0.18$} & \small{$0.19$} & \small{$0.18$} & \small{$0.18$} & \small{$0.18$} & \small{0.19}\\
	\hline 
    \end{tabular}
}\\
\vspace{-2mm}
\subfloat[$\beta=0.2$\label{tab:beta0.2}]{
    \begin{tabular}{|c|c|c|c|c|c|c|c|c|}	
	\hline
	\small{$f_j$} & \small{$0.01$}  & \small{$0.01$} & \small{$0.1$}  & \small{$0.1$}& \small{$0.5$} & \small{$0.5$} & \small{$0.9$} & \small{$0.9$}\\
	\hline
         	\small{$f_t$} & \small{$0$} & \small{$0.01$} & \small{$0$} & \small{$0.01$} & \small{$0$} & \small{$0.01$} & \small{$0$} & \small{$0.01$}\\
	\hline
	\small{$\hat{f}_u$} & \small{1.8} & \small{1.8} & \small{$0.87$} & \small{$0.88$} & \small{$0.82$} & \small{$0.84$} & \small{$0.82$} & \small{0.83}\\
	\hline 
    \end{tabular}
}
\vspace{-2mm}
    \caption{{Threshold $\hat{f}_u$ for different $f_j$ and $f_t$.} } 
    \label{tab:results}
\vspace{-2mm}
\end{table}

Suppose both $t$ and $u$ are among the top-$N$ items with the largest estimated frequencies. The true frequency $f_t$ for the target item $t$ is small, since our attack aims to promote an unpopular item. We have $\text{Pr}(y_j=1|y_t=1) <\text{Pr}(y_j=1|y_u=1)$ when $f_u$ is smaller than a threshold $\hat{f}_u$. Table \ref{tab:results} shows such threshold for different values of $f_j$ and $f_t$, where $\beta=0.05$ and $\beta=0.2$.  We observe that $f_u$ is highly likely smaller than the threshold $\hat{f}_u$ for a variety of $f_j$ when $\beta=0.2$, as $\hat{f}_u$ is very large (sometimes even larger than 1). This observation shows that if we randomly pick an item as $j$ and compare the conditional probabilities $\text{Pr}(y_j=1|y_u=1)$ for each item $u$ in the top-$N$ items, then we can detect the item with the smallest conditional probability as the target item. However, when $\beta=0.05$, the effectiveness of such detection method depends on the true frequencies $f_j$ and $f_u$. 

\myparatight{OLH} The conditional probability based detection method can also be used for OLH when $r=1$. Specifically, we can construct a $d$-bit binary vector $\bm{y}$ for each user whose $v$th entry $y_v=1$ if and only if $H(v)=a$, where $(H,a)$ is the user's perturbed value. Assuming the hash function hashes an item uniformly at random to a hash value in $[d']$. Then, we have the following conditional probabilities:
{\small{\begin{align}
		\label{eq:olh-target}
\text{Pr}(y_j=1|y_t=1)
	&= q + \frac{f_jq(p-q)}{f_tp+(1-f_t)q+\frac{\beta}{1-\beta}},\\
	\text{Pr}(y_j=1|y_u=1)
	&= q + \frac{f_jq(p-q)}{f_up+(1-f_u)q+\frac{\beta}{1-\beta}\cdot q}.
	\label{eq:olh-non-target}
\end{align}}}%
}

 \begin{figure*}[!t]
	 \centering
\subfloat[OUE]{\includegraphics[width=0.2\textwidth]{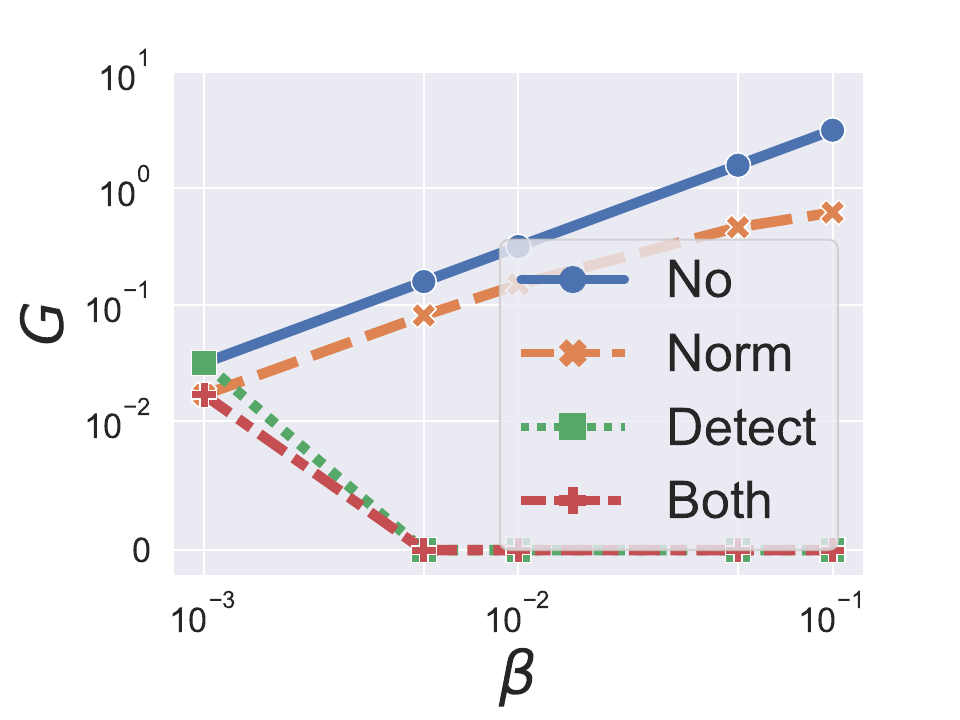}\label{fig:detect_beta_oue}}
\subfloat[OLH]{\includegraphics[width=0.2\textwidth]{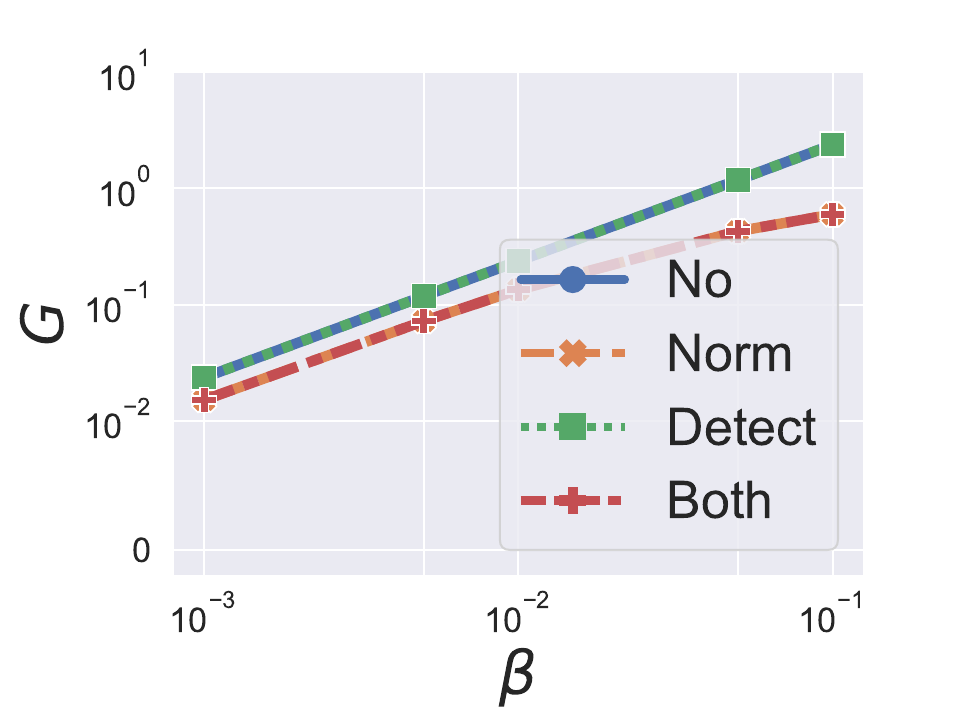}\label{fig:detect_beta_olh}}
\subfloat[OUE]{\includegraphics[width=0.2\textwidth]{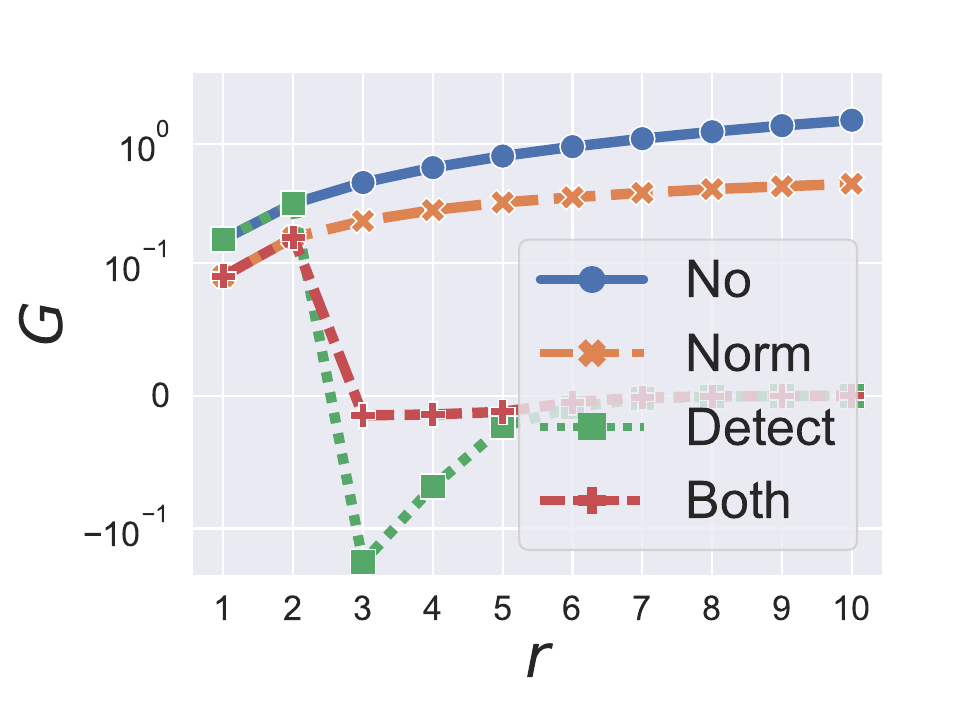}\label{fig:detect_r_oue}}
\subfloat[OLH]{\includegraphics[width=0.2\textwidth]{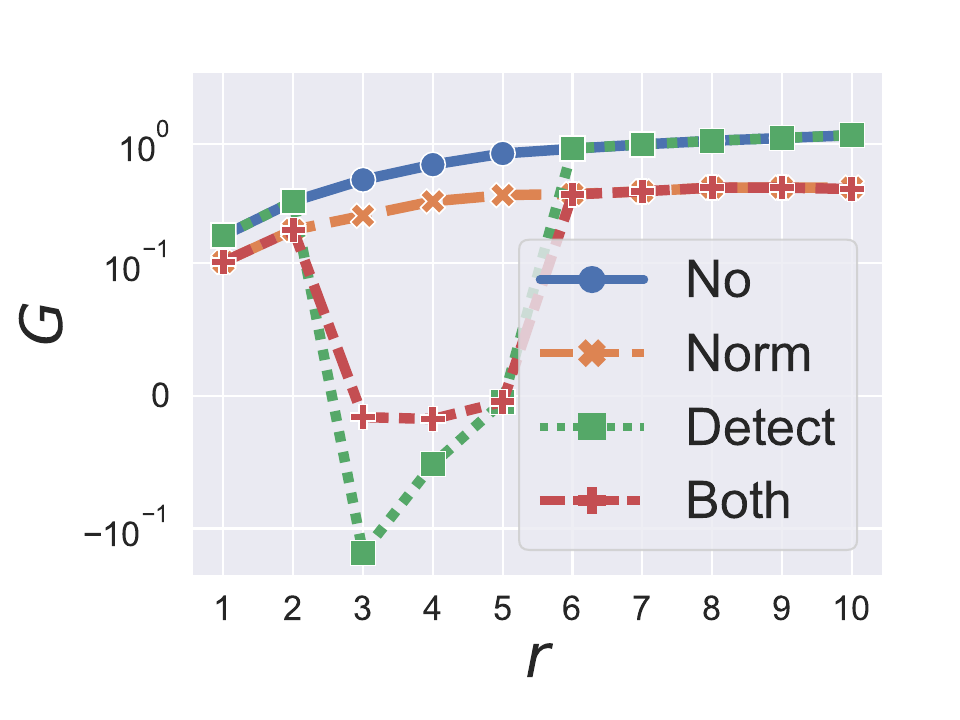}\label{fig:detect_r_olh}}
\subfloat[Adaptive MGA]{\includegraphics[width=0.2\textwidth]{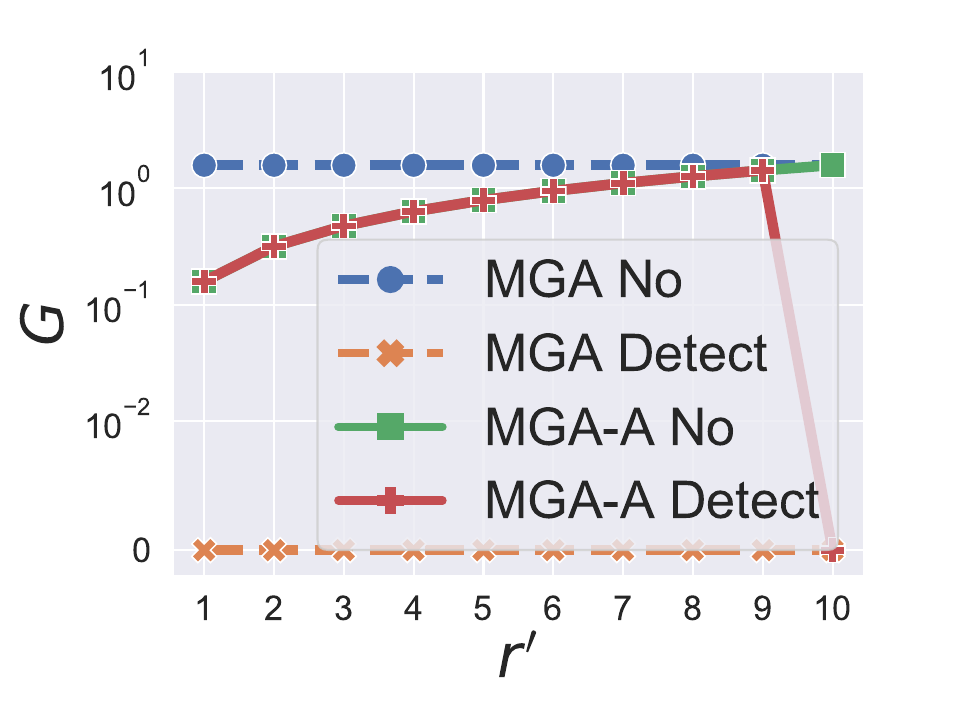}\label{fig:adaptive_rp}}
\vspace{-2mm}
	 \caption{{(a)-(b) Impact of $\beta$ on the countermeasures against MGA when $r=10$. (c)-(d) Impact of $r$ on the countermeasures against MGA when $\beta=0.05$. (e) Impact of $r'$ on the adaptive MGA (MGA-A) to OUE when $r=10$.}}
	 \vspace{-4mm}
\end{figure*}

 \begin{figure}[!t]
	 \centering
	 \vspace{-2mm}
\subfloat[$N$]{\includegraphics[width=0.22\textwidth]{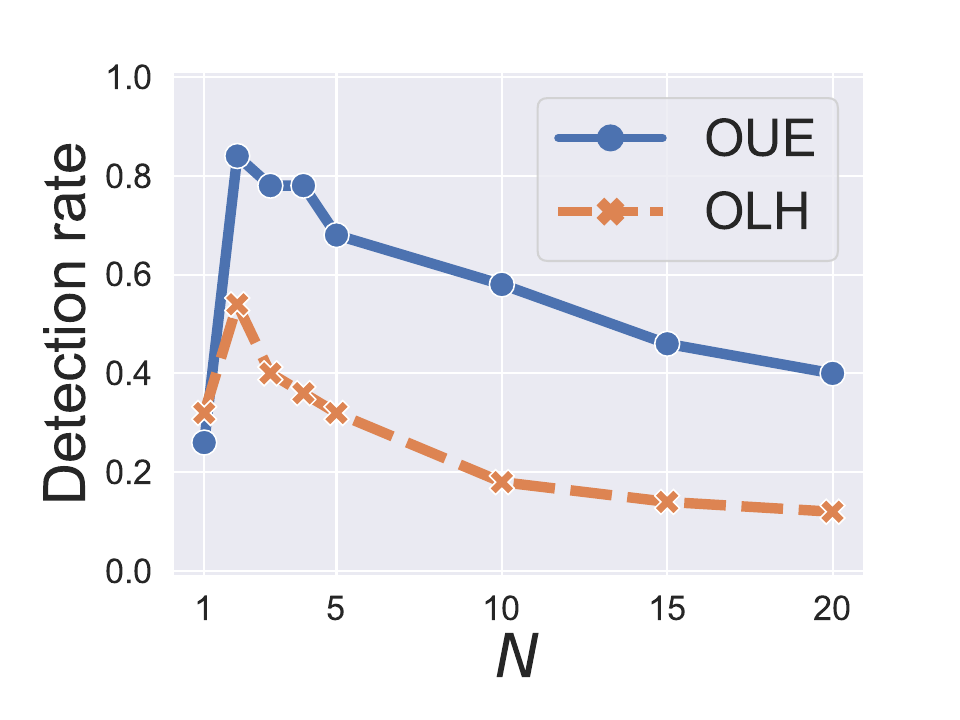}\label{fig:r1_N}}
\subfloat[$\beta$]{\includegraphics[width=0.22\textwidth]{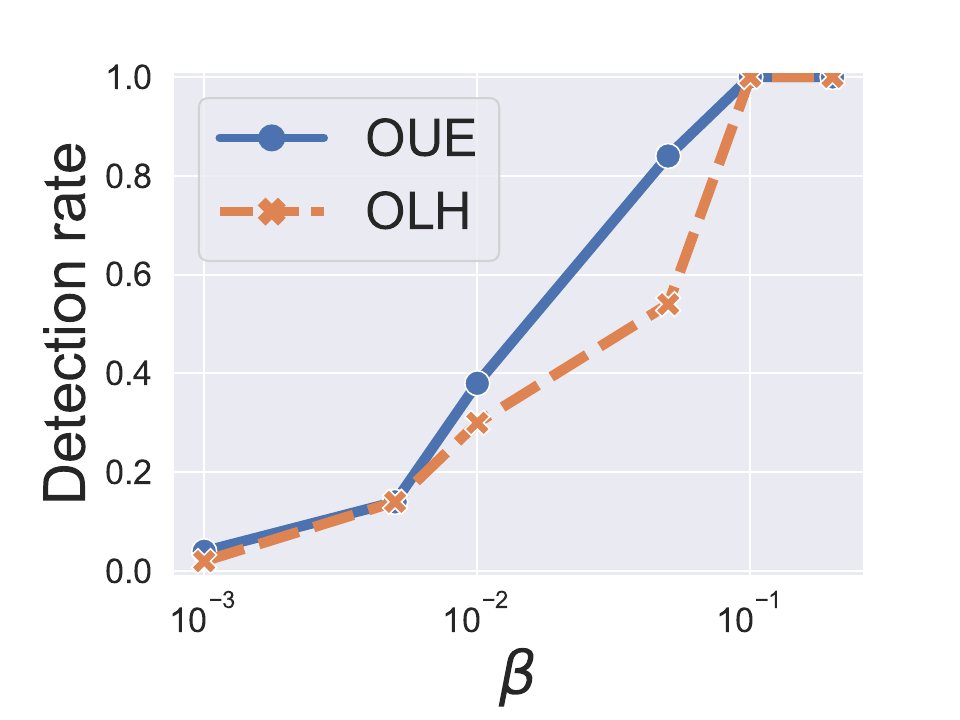}\label{fig:r1_beta}}
\vspace{-2mm}
 \caption{{Impact of $N$ and $\beta$ on the detection rate of the conditional probability based method for $r=1$.}}\label{fig:r1}
 \vspace{-4mm}
\end{figure}

\begin{table}[!t]\renewcommand{\arraystretch}{1.8}
\setlength{\tabcolsep}{2pt}
    \centering
     \tiny
    \scalebox{0.9}{
     \begin{tabular}{|c|c|c|c|c|c|c|c|c|c|c|} \hline
         \multirow{ 2}{*}{} &   \multicolumn{2}{|c|}{\small kRR} &   \multicolumn{4}{|c|}{\small OUE} &   \multicolumn{4}{|c|}{\small OLH}\\ \cline{2-11}
         {} & {\small No} & {\small Norm} & {\small No} & {\small Norm} & {\small Detect} & {\small Both} & {\small No} &  {\small Norm} & {\small Detect} & {\small Both}\\ \hline
 	{\small RPA} & {\small 2e-3} & {\small -1e-3} & {\small 0.50} & {\small 2e-3} & {\small 0.50} & {\small 2e-3} & {\small -2e-3} & {\small -2e-3} & {\small -2e-3} & {\small -2e-3}\\ \hline
	{\small RIA} & {\small 0.05} & {\small -4e-3} & {\small 0.05} & {\small 0.03} & {\small --} & {\small --} & {\small 0.05} & {\small 0.03} & {\small --} & {\small --}\\ \hline
	{\small MGA} & {\small 2.72} & {\small 0.43} & {\small 1.58} & {\small 0.46} & {\small 7e-17} & {\small -2e-16} & {\small 1.18} & {\small 0.43} & {\small 1.18} & {\small 0.43}\\ \hline
    \end{tabular}}
   \vspace{-2mm}
    \caption{Overall gains of the three attacks on the IPUMS dataset after countermeasures are deployed. The column ``No'' means no countermeasure is used. The column ``Both'' means the combined countermeasure. ``--'' means that the countermeasure is not applicable. Only normalization is applicable for kRR.}
    \label{tab:defense}
 \vspace{-2mm}
\end{table}

\vspace{-2mm}
\subsection{Experimental Results}  
{We empirically evaluate the effectiveness of the three countermeasures. Unless otherwise mentioned,  we focus on  normalization and detecting fake users as the conditional probability based detection is only applicable for one target item.}  
Note that normalization and detecting fake users can also be used together. Specifically, the central server can first detect and remove the fake users, and then perform normalization. Therefore, we will also evaluate the combined countermeasure. We use the same default experimental setup as those in Section~\ref{exp:setup}.  
Moreover, we use the FP-growth algorithm implemented in the Python package \emph{mlxtend} \cite{raschkas_2018_mlxtend} to identify frequent itemsets. 

\subsubsection{Frequency Estimation}
\label{defense-fe}

\myparatight{Overall results} Table \ref{tab:defense} shows the experimental results with no countermeasure, normalization, detection, and combined countermeasure, where $\beta=0.05$ and $r=10$. 
We observe that the countermeasures are effective in some scenarios. For example, for OUE, combining the two countermeasures leads to an overall gain of -2e-16 for MGA, which means that the estimated total frequency of the target items is even smaller than the one before attack. However, the countermeasures are ineffective in other scenarios. For instance,  MGA can still achieve a large overall gain of 0.43 for OLH even if  both countermeasures are used. Normalization can reduce the overall gains of all the three attacks for the three protocols except RPA for OLH. However, MGA still achieves large overall gains after normalization. Detecting fake users is ineffective for RPA because RPA randomly samples perturbed values in the encoded space for the fake users and thus the perturbed values do not have meaningful statistical patterns. When the countermeasures are used, MGA is still the most effective attack in most cases. Therefore, we focus on MGA and further study the impact of $\beta$ and $r$ on the countermeasure effectiveness.

{\myparatight{Impact of $\beta$ and $r$ on MGA} Figure \ref{fig:detect_beta_oue}-\ref{fig:detect_beta_olh} show the impact of $\beta$ on the countermeasures against MGA when we fix $r=10$, while Figure \ref{fig:detect_r_oue}-\ref{fig:detect_r_olh} show the results for $r$ when we fix $\beta=0.05$ on the IPUMS dataset. First, we observe that for OUE,  detecting fake users and the combined countermeasure can effectively defend against the MGA attacks (i.e., reduce the overall gains to almost 0) when $\beta$ and $r$ are larger than some thresholds, e.g., $\beta >0.001$ and $r\geq 3$. The countermeasures are ineffective when $\beta$ or $r$ is small (e.g., $\beta \leq 0.001$ or $r \leq 2$). This is because the detection method relies on that the target itemset is frequent and abnormal, but the target itemset is not frequent when $\beta$ is small and is not abnormal among the users' perturbed values when $r$ is small. 

Second, for OLH, detecting fake users and the combined countermeasure can effectively defend against the MGA attacks only when $r$ is not too small nor large, e.g., $3\leq r \leq 5$ in our experiments. Recall that, to attack OLH, our MGA randomly samples 1,000 hash functions and uses the one that hashes the largest number of target items to the same value for each fake user. When $r \leq 5$, our MGA can find a hash function that hashes all target items to the same value. Therefore, the target itemset is frequent among the users' perturbed values. Moreover, when $r \geq 3$, the frequent target itemset is also abnormal. As a result,  the detection method can detect MGA when $3\leq r \leq 5$. When $r \geq 6$, our MGA can only find a hash function among the 1,000 random ones that hashes a subset of the target items to the same value for each fake user. In other words, each fake user essentially randomly picks a subset of the target items and promotes them.  Therefore, the entire target itemset is not frequent enough and MGA evades detection. Our MGA evades detection for all the explored $\beta$ in Figure~\ref{fig:detect_beta_olh} because $r=10$ in these experiments. } 

{\myparatight{Adaptive MGA to OUE} Inspired by the evasiveness of MGA to OLH, we can also adapt MGA to OUE that evades detection. Specifically, for each fake user, instead of using a perturbed value that supports all $r$ target items, we randomly select $r'$ of the $r$ target items and find a perturbed value that supports the $r'$ selected target items. The adaptive attack splits the frequency of the target itemset with size $r$ to ${r\choose r'}$ itemsets with size $r'$, which becomes much harder to detect. We call such adaptive attacks MGA-A. Figure~\ref{fig:adaptive_rp} shows the impact of $r'$ on MGA-A to OUE when $r=10$. We observe that our adaptive MGA achieves smaller overall gains as $r'$ becomes smaller when no countermeasures are deployed. However, our adaptive MGA evades detection when $r'<r$. }

{\myparatight{Attack stealthiness} If the frequent itemset mining based detection method returns an abnormal frequent itemset, then the central server predicts that it is under our MGA attack. Our attack is stealthy if the central server cannot detect it. Our results show that, for OUE, our MGA is stealthy when $\beta$ or $r$ is small (e.g., $\beta \leq 0.001$ or $r \leq 2$), and our adaptive MGA is stealthy when $r' < r$. For OLH, our MGA is stealthy when $r$ is small or large enough, e.g., $r \leq 2$ or  $r \geq 6$ in our experiments. }
 
{\myparatight{Conditional probability based detection for $r=1$} We measure the effectiveness of the conditional probability based detection method using \emph{detection rate}. Specifically, in each experiment, we perform our MGA attack with a random target item 50 times and the detection rate is the fraction of the 50 experiment trials in which the target item is correctly detected. Figure \ref{fig:r1_N} shows the impact of $N$ on the detection rate  when we fix $\beta=0.05$ on the IPUMS dataset. We observe that the detection rate first increases and then decreases as $N$ grows. This is because when $N$ is too small, e.g., $N=1$, the target item is likely not in the top-$N$ items; and when $N$ is too large, it's more likely that there exists a non-target item in the top-$N$ items that has a smaller conditional probability than the target item. We notice that the detection rate is lower for OLH than for OUE. This is because the threshold $\hat{f}_u$ for OLH is smaller than that for OUE, e.g., $\hat{f}_u=0.18$ for OLH and $\hat{f}_u=0.26$ for OUE when $f_t=f_j=0.01$. Figure \ref{fig:r1_beta} shows the impact of $\beta$ on the detection rate, where we explore $N=1$ to $20$ to find the $N$ that achieves the highest detection rate for each given $\beta$. We observe that the detection rate increases as $\beta$ increases, which implies that the MGA attack with $r=1$ is easier to detect when there are more fake users. Once the target item is detected, the server can compute the sum of the estimated frequencies of all non-target items as $\tilde{f}_U=\sum_{u\neq t}\tilde{f}_u$ and set the estimated frequency of the target item as $\tilde{f}_t=1-\tilde{f}_U$, which can reduce the overall gain of MGA. For instance, the overall gain decreases from 2.37 to 0.095 for OLH when $\beta=0.1$. }

\subsubsection{Heavy Hitter Identification}
Normalization is ineffective for heavy hitter identification because normalization does not impact the ranking of the items' estimated frequencies.  Moreover, the conditional probability based detection is only applicable to one target item. Therefore, we  perform experiments on detecting fake users for heavy hitter identification. Moreover, we focus on MGA because RIA and RPA are ineffective even without detecting fake users (see Figures~\ref{fig:hh_synth}). We observe that detecting fake users is effective in some scenarios but not in others. For instance, when $r=5$, detecting fake users can reduce the success rate of MGA from 1 to 0, as all  fake users can be detected. However,  when $r=10$, our MGA can still achieve a success rate of 1. 

\vspace{-2mm}
\subsection{Other Countermeasures} 
Detecting fake users is related to Sybil detection in distributed systems and social networks. Many methods have been proposed to mitigate Sybil attacks. For instance, methods~\cite{stringhini2010detecting,yu2006sybilguard,danezis2009sybilinfer,gong2014sybilbelief,wang2013you,cao2014uncovering,wang19ndss,yuan2019detecting} that leverage content, behavior, and social graphs are developed to detect fake users in social networks. Our detection method can be viewed as a content-based method. Specifically, our detection method analyzes the statistical patterns of the user-generated content (i.e., perturbed values sent to the central server) to detect fake users. However, our detection method is different from the content-based methods to detect fake users in social networks, as the user-generated content and their statistical patterns differ. Social-graph-based methods are inapplicable when the social graphs are not available. 

Another countermeasure is to leverage Proof-of-Work~\cite{dwork1992pricing}, like how Sybil is mitigated in Bitcoin. In particular, before a user can participate in the LDP protocol, the central server sends a random string to the user; and the user is allowed to participate the LDP protocol after the user finds a string such that the cryptographic hash value of the concatenated string has a certain property, e.g., the first 32 bits are all 0. However, such method incurs a large computational cost for genuine users, which impacts user experience. Moreover, when users use mobile devices such as phones and IoT devices, it is challenging for them to perform the  Proof-of-Work. Malicious-party-resistant SMPC could also be used to limit the impact of fake users (e.g.,~\cite{naor2019how}). However, such methods generally sacrifice computational efficiency. 

%%%%%%%%%%%%%%%%%%%%%%%%%%%%%%%%%%%%%%%%%%%%%%%%%%%%%%%

\section{Discussion}
\label{sec:discussion}
\vspace{-2mm}
{\myparatight{Applicability to shuffling-based and SMPC-based protocols} Shuffling-based protocols~\cite{erlingsson2019amplification} apply shuffling to the users' perturbed vectors such that a better  DP guarantee can be derived. Since they still encode and perturb each user's data, our attacks are applicable. When SMPC-based protocols have local encoding and perturbation steps like~\cite{kairouz2015secure}, our attacks are applicable and the security-privacy trade-off still holds. When there is no local encoding or perturbation step in the SMPC-based DP protocols like~\cite{roy2020crypt}, our RPA and MGA  are not applicable because an attacker cannot manipulate the perturbed vectors. However, our RIA is still applicable because it only needs to modify the item value. In this case, we do not have the security-privacy trade-off because the overall gain of RIA does not rely on the privacy budget.} 

{\myparatight{RIA without perturbation} A variant of RIA is that a fake user samples one of the $r$ target items randomly, encodes it, and sends the encoded value to the central server without perturbing it. 
	When $r=1$, this RIA variant has the same overall gain as MGA. 
	When $r>1$, the RIA variant uses a fake user to promote only one target item. However, MGA uses a fake user to simultaneously promote multiple target items, which means that its overall gain is multiple times of the RIA variant's overall gain. Moreover, it may be easy for the central server to detect the RIA variant for OUE. Specifically, the server can count the number of 1's in a vector from a user. If there is only one entry that is 1, then it is likely that the vector is from a fake user as the probability that a genuine vector contains a single 1 is fairly small.}
 
{\myparatight{Defending OLH by restricting the hash functions} Since MGA to OLH relies on searching a hash function that maps target items to the same hash value, the server could restrict the space of seeds of the hash function or select the hash function by itself to defend OLH against MGA. However, the defense may break the privacy guarantees. In particular,  an untrusted server could carefully select a space of seeds or a hash function that does not have collisions in the item domain. For instance, a hash value $h$ corresponds to a unique item. When receiving a hash value $h$ from a user, the server knows the user's item, which breaks the LDP guarantee.} 

\vspace{-2mm}
\section{Conclusion}
In this work, we perform the first systematic study on data poisoning attacks to LDP protocols. Our results show that an attacker can inject fake users to an LDP protocol and send carefully crafted data to the server such that the target items are estimated to have high frequencies or promoted as heavy hitters. We show that we can formulate such an attack as an optimization problem, solving which an attacker can maximize its attack effectiveness. We theoretically and/or empirically show the effectiveness of our attacks. Moreover, we explore three countermeasures against our attacks. Our empirical results show that these countermeasures have limited effectiveness in some scenarios, highlighting the needs for new defenses against our attacks. 

Interesting future work includes generalizing our attacks to other LDP protocols, e.g., LDP protocols for itemset mining~\cite{wang2018locally} and key-value pairs~\cite{ye2019privkv}, as well as developing new defenses to mitigate our attacks.  

{\section*{Acknowledgements}
We thank the anonymous reviewers for their constructive comments. The conditional probability based detection method for one target item was suggested by a reviewer. This work was supported by NSF grant No.1937786.}

%%%%%%%%%%%%%%%%%%%%%%%%%%%%%%%%%%%%%%%%%%%%%%%%%%%%%%%

\bibliographystyle{plain}
\bibliography{refs}

%%%%%%%%%%%%%%%%%%%%%%%%%%%%%%%%%%%%%%%%%%%%%%%%%%%%%%%

\appendix

\section{Proof of Theorem 2}\label{proof:thm2}
\begin{proof}
Let $\beta(1-f_T)+\frac{\beta (d-r)}{e^\epsilon-1} > \beta(2r-f_T)+\frac{2\beta r}{e^\epsilon-1}$, we have:
{\small{\begin{align}
	1 + \frac{d-r}{e^\epsilon-1} > 2r + \frac{2r}{e^\epsilon-1} \iff \frac{d-3r}{e^\epsilon-1} > 2r-1.
\end{align}}}%
Since $e^\epsilon>1$,  the inequality above is equivalent to $d > (2r-1)(e^\epsilon-1) +3r$. 
\end{proof}

\section{FPRs of Detecting Fake Users}
\label{proof:tau_oue} 
\myparatight{OUE}
If a user's perturbed binary vector $\bm{y}$ follows the OUE protocol, then we can calculate the probability that the items in a set $B$ of size $z$, are all 1 in the perturbed binary vector as follows: $\text{Pr}(y_b=1, \forall b\in B)=pq^{z-1}$ if $v\in B$ and $\text{Pr}(y_b=1, \forall b\in B) =q^z$ otherwise, where $y_b$ is the $b$th bit of the perturbed binary vector $\bm{y}$ and $v$ is the user's item. Let $f_B=\sum_{b\in B}f_b$ denote the sum of true frequencies of all items in $B$, $X_1$ denote the random variable representing the number of users whose items are in $B$ and whose perturbed binary vectors are 1 for all items in $B$, and $X_2$ denote the random variable representing the number of users whose items are not in $B$ and whose perturbed binary vectors are 1 for all items in $B$. If all the $n+m$ users follow the OUE protocol, then we have the following distributions: $X_1 \sim Binom(f_B(n+m), pq^{z-1})$ and $X_2 \sim Binom((1-f_B)(n+m), q^{z})$, where $Binom$ is a binomial distribution. Now we consider another random variable $X=X_1+X_2$, which represents the number of users whose perturbed binary vectors are 1 for all items in $B$. $X$ follows a distribution with mean $\mu$ and variance $Var$ as follows:
{\small{\begin{align}
	\mu &= f_B(n+m)pq^{z-1} + (1-f_B)(n+m)q^z\\
	 &\le (n+m)pq^{z-1}\\
	Var &= f_B(n+m)pq^{z-1}(1-pq^{z-1}) \nonumber\\&+ (1-f_B)(n+m)q^z(1-q^z) \\
	&\le (n+m)pq^{z-1}(1-pq^{z-1}).
\end{align}}}%
Based on the Chebyshev's inequality,  for any $\tau_z>(n+m)pq^{z-1}$, we have:
{\small{\begin{align}\label{eq:chebyshev}
		\text{Pr}(X\ge \tau_z)	
		&=\text{Pr}(X-\mu\ge \tau_z-\mu) \nonumber\\
		&\le\text{Pr}(|X-\mu|\ge \tau_z-\mu)\nonumber\\
		&\le \frac{Var}{(\tau_z-\mu)^2} \nonumber\\
		&\le \frac{(n+m)pq^{z-1}(1-pq^{z-1})}{[\tau_z-(n+m)pq^{z-1}]^2} 
\end{align}}}%
Here, if we choose $\tau_z$ as the threshold, the probability $\text{Pr}(X\ge \tau_z)$ is the false positive rate, which is upper bounded by $\frac{(n+m)pq^{z-1}(1-pq^{z-1})}{[\tau_z-(n+m)pq^{z-1}]^2}$. 

\myparatight{OLH}
As discussed in Section \ref{sec:detection}, we first construct a $d$-bit binary vector $\bm{y}$ for each user with a tuple $(H, a)$ such that $y_v=1$ if and only if $H(v)=a$. For an item set $B$ of size $z$, assume $X$ is a random variable that represents the number of users whose constructed binary vectors are 1's for all items in $B$. If all the $n+m$ users follow the OLH protocol, then for any $\tau_z > 0$, the probability that $X\geq \tau_z$ is bounded as follows: 
{\small{\begin{align}\label{eq:uni_abn_olh}
	\text{Pr}(X\geq \tau_z) &= 1-\text{Pr}(X \leq \tau_z-1)\nonumber\\
	&=  1-I({1-q^{z-1}};n+m-\tau_z+1, \tau_z)\nonumber\\
	&= I(q^{z-1};\tau_z, n+m-\tau_z+1)
\end{align}}}%
Note that if we set $\tau_z$ as the threshold, the probability $\text{Pr}(X\ge \tau_z)$ is the false positive rate. 

%%%%%%%%%%%%%%%%%%%%%%%%%%%%%%%%%%%%%%%%%%%%%%%%%%%%%%%

\end{document}